\PassOptionsToPackage{dvipsnames}{xcolor}

\newif\iftechreport
\techreporttrue

\iftechreport
\documentclass[sigconf,screen,nonacm]{acmart}
\else
\documentclass[sigconf,screen]{acmart}
\fi

\usepackage{natbib}

\usepackage{tikz}
\usepackage[nointegrals]{wasysym}
\usepackage{amsfonts}
\usepackage{siunitx}

\usepackage{tabularx}
\usepackage{xcolor}
\usepackage[inline]{enumitem}   \setlist[enumerate,1]{label = (\arabic*), ref = (\arabic*)}
\usepackage{xspace}
\usepackage{xargs,xstring}
\usepackage{ifthen}
\usepackage{makecell}
\usepackage{multirow}
\usepackage[utf8]{inputenc}

\usepackage{amsmath}
\usepackage{mathtools} \usepackage{mathpartir}
\usepackage{stmaryrd}  \usepackage{halloweenmath}

\usepackage{amsthm}    \usepackage{thmtools}  \declaretheoremstyle[bodyfont=\normalfont]{normalstyle}

\declaretheorem[style=normalstyle]{definition}

\newtheorem{lemma}{Lemma}
\newtheorem*{lemma-non}{Lemma}

\newtheorem{corollary}{Corollary}

\usepackage{hyperref}
\usepackage[capitalize,nameinlink]{cleveref}
\crefname{hypothesis}{Hyp.}{Hyp.}
\Crefname{hypothesis}{Hypothesis}{Hypotheses}
\crefname{contract}{Contract}{Contracts}
\crefname{proposition}{Prop.}{Prop.}
\Crefname{proposition}{Proposition}{Propositions}
\crefname{lstlisting}{Listing}{Listings}
\Crefname{lstlisting}{Listing}{Listings}
\crefname{equation}{}{}
\Crefname{equation}{Eq.}{Equations}
\crefname{requirement}{Req.}{Req.}
\Crefname{requirement}{Requirement}{Requirements}
\crefformat{section}{#2§#1#3}
\Crefformat{section}{#2§#1#3}

\usepackage{tikz-cd}
\usetikzlibrary{decorations.pathmorphing}

\usepackage[font=footnotesize]{caption}
\usepackage{subcaption}
\usepackage{newfloat}
\DeclareFloatingEnvironment[
fileext = los,
placement={!ht},
name = Listing
]{listing}
\crefname{sublisting}{Listing}{Listings}
\Crefname{sublisting}{Listing}{Listings}
\DeclareCaptionSubType[alph]{listing}
\makeatletter
\renewcommand\p@subfigure{}
\makeatother

\usepackage{fancyvrb}
\usepackage{listings}
\usepackage{color}
\lstdefinestyle{nonumbers}{
  numbers=none,
  xleftmargin=0em,
}

\usepackage[most]{tcolorbox}

\tcbset {
  base/.style={
    arc=0mm,
    boxrule=0mm,
    colbacktitle=black!10!white,
    coltitle=black,
    fonttitle=\bfseries,
    top=0.5mm,
    bottom=0.5mm,
    left=1.5mm,
    right=1.5mm,
    leftrule=1mm,
    title={#1},
    toptitle=0.5mm,
    bottomtitle=0.5mm,
  }
}

\newtcolorbox{mainbox}[1]{
  colframe=black!30!white,
  base={#1}
}

\hypersetup{
    colorlinks,
    linkcolor={red!70!black},
    citecolor={green!50!black},
    urlcolor={blue!70!black}
}
 \RequirePackage{thmtools}
\RequirePackage{thm-restate}

\newif\ifnotes
\notesfalse

\newcommand{\zenodourl}{\url{https://doi.org/10.5281/zenodo.20770661}}
\newcommand{\proteusurl}{\url{https://github.com/proteus-core}}

\ifnotes
\newcommand\lesly[1]{\textcolor{red}{Lesly: #1}}
\newcommand\DD[1]{\textcolor{orange}{Davide: #1}}
\newcommand\marton[1]{\textcolor{magenta}{Marton: #1}}
\newcommand\tamara[1]{\textcolor{green!70!black}{Tamara: #1}}
\newcommand\benj[1]{\textcolor{blue}{Benj: #1}}
\else
\newcommand\lesly[1]{}
\newcommand\DD[1]{}
\newcommand\marton[1]{}
\newcommand\tamara[1]{}
\newcommand\benj[1]{}
\fi

\newcommand{\ie}{i.e.,\xspace}
\newcommand{\eg}{e.g.,\xspace}
\newcommand{\cf}{cf.\xspace}
\newcommand{\wrt}{w.r.t.\xspace}

\newenvironment{diff}
{\begingroup}{\endgroup}

\usepackage{tikz,pifont}
\newcommand\shield{
\resizebox{.8em}{.8em}{\tikz [baseline] \draw (3ex,2.5ex) -- (1.5ex,3ex) -- (0ex,2.5ex)  to [out=-90, in=165] (1.5ex,0ex) to [out=15, in=-90] cycle (1.5ex,1.5ex) node {\ding{51}};}
}

\newcommand{\riscv}{RISC\nobreakdashes-V\xspace}

\colorlet{directive_color}{black}
\colorlet{leak_color}{RedOrange}
\definecolor{instruction_color}{RGB}{5,59,128}
\definecolor{codefn}{RGB}{14,113,89}
\definecolor{codekwd}{RGB}{22, 35, 115}

\newcommand{\mydef}{\triangleq}

\newcommand{\ctrue}{\mathtt{true}}
\newcommand{\cfalse}{\mathtt{false}}

\newcommand{\Ar}{\mathsf{Arr}}
\newcommand{\ar}{\sym{a}}
\newcommand{\artwo}{\sym{b}}

\newcommand{\Arr}{{\mathsf{Arr}}}

\newcommand\Vals{\mathbb{V}}
\newcommand\val[1]{\def\param{#1}\normalfont{\mathtt{\ifx\param\empty{v}\else{#1}\fi}}}

\newcommand{\Nat}{\mathbb N}
\newcommand{\nat}{n}

\newcommand{\bool}{b}

\newcommand{\bms}{\beta}

\newcommand\Loc{\ensuremath \mathbb{L}}
\newcommand\loc{\ell}
\newcommand\Regs{\mathsf{Reg}}

\newcommand\register[1]{\texttt{#1}}

\newcommand\inststyle[1]{\kwd{#1}}
\newcommand\fence{\inststyle{fence}\xspace}
\newcommand\cprotect{\inststyle{protect}\xspace}
\newcommand\lfence{\inststyle{lfence}\xspace}
\newcommand\dsb{\inststyle{dsb}\xspace}
\newcommand\dfence{\texorpdfstring{\inststyle{dfence}}{dfence}\xspace}
\newcommand\dfencex{\dfence \vx\xspace}

\newcommand\load[2]{\def\parama{#1}\def\paramb{#2}\ifx\parama\empty{}\else{\register{#1}\ensuremath{\leftarrow}}\fi \inststyle{load}\ifx\paramb\empty{}\else{\ #2}\fi}
\newcommand\store[2]{\def\parama{#1}\def\paramb{#2}\inststyle{store}\ifx\parama\empty{}\else{\ #1}\fi \ifx\paramb\empty{}\else{\ #2}\fi} \newcommand\jmp[1]{\def\parama{#1}\inststyle{jmp}\ifx\parama\empty{}\else{\ #1}\fi} \newcommand\beqz[2]{\def\parama{#1}\def\paramb{#2}\inststyle{beqz}\ifx\parama\empty{}\else{\ #1}\fi \ifx\paramb\empty{}\else{\ #2}\fi} \newcommand\prog[1]{\def\param{#1}\ifx\param\empty{P}\else{{P[#1]}}\fi}

\newcommand\mem{{\ensuremath \mathnormal{m}}}
\newcommand\reg{\ensuremath \mathnormal{r}}

\def\ws{\ }\def\wsind{\ }\def\sep{\kwd{;\,}}

\newcommand{\calign}[2][t]{\begin{array}[#1]{@{}l@{}}#2\end{array}}

\newcommand{\expression}[1]{\mathtt{#1}}
\newcommand{\sym}[1]{{\color{codefn}{\mathtt{#1}}}}
\newcommand{\command}[1]{\mathtt{#1}}
\DeclareRobustCommand{\kwd}[1]{{\text{\color{codekwd}\texttt{\textbf{#1}}}}}

\newcommand{\Expr}{\mathsf{Expr}}
\newcommand{\expr}{\expression{E}}
\newcommand{\exprtwo}{\expression{F}}
\newcommand{\vx}{\register{x}}
\newcommand{\vy}{\register{y}}

\newcommand{\Op}{\mathsf{Ops}}
\newcommand{\op}{\sym{op}}

\newcommand{\size}[1]{|#1|}

\newcommand{\Instr}{\mathsf{Instr}}
\newcommand{\stat}{\command{I}}
\newcommand{\stattwo}{\command{J}}
\newcommand{\asgn}{\mathrel{\kwd{{:}{=}}}}
\newcommand{\cmemread}[3][]{#2 \asgn_{#1} {#3}}
\newcommand{\cmemasgn}[3][]{{#2} \asgn_{#1} {#3}}

\newcommand{\cwhile}[2]{\kwd{while}\ #1\ \kwd{do} \wsind \calign{#2} \ws \kwd{od}}
\newcommand{\cif}[3]{\kwd{if}\ #1\ \kwd{then}\wsind \calign{#2}\ws \kwd{else}\wsind \calign{#3}  \ws \kwd{fi}}
\newcommand{\cnil}{\epsilon}

\newcommand{\Cmd}{\mathsf{Prg}}
\newcommand{\cmd}{\command{P}}
\newcommand{\cmdtwo}{\command{Q}}

\newcommand{\States}{\mathsf{States}}
\newcommand{\st}{s}
\newcommand{\sttwo}{t}

\newcommand{\bufreadi}[2]{\mathsf{sread}^{#2}(#1)}
\newcommand{\buflookup}[2]{{\mathsf{read}(#1, #2)}}

\newcommand{\confone}{C}

\newcommand{\sframe}[3]{#1, #2, #3}

\newcommand{\cons}{\mathrel{:}}
\newcommand{\nil}{\varepsilon}

\newcommand{\upd}{\leftarrow}
\newcommand{\supdate}[3]{{#1[{#2} \upd {#3}]}}

\newcommand{\mbuf}{{\mu}}
\newcommand{\bm}[2]{{(#1,#2)}}

\newcommand{\Dir}{\mathsf{Dir}}
\newcommand{\dir}{d}
\newcommand{\dstep}{\mathsf{step}}
\newcommand{\doob}[1]{\mathsf{oob}\ {#1}}
\newcommand{\dload}[2][]{\mathsf{ld}_{#1}\,#2}

\newcommand{\dbranch}[2][]{\mathsf{br}_{#1}\,#2}

\newcommand{\Ds}{D}

\newcommand{\Obs}{\mathsf{Obs}}
\newcommand{\Os}{O}
\newcommand{\obs}{o}

\newcommand{\onone}{\bullet}
\newcommand{\ojmp}[1]{\mathsf{jmp}\ {#1}}

\newcommand{\omem}[1]{{\mathsf{mem}\,#1}}

\newcommand{\obranch}[1]{{\mathsf{br}\,#1}}

\newcommand{\bitem}[2]{{[#1 \mapsto #2]}}

\newcommand{\sem}[1]{\llbracket #1 \rrbracket}
\newcommand{\sto}[2]{\xrightarrow[#1]{#2}}
\newcommand{\jmpcont}{\phi}
\newcommand{\step}[3][]{\jmpcont\vdash #2 \to^{#1}\, #3}
\newcommand{\sstep}[5][]{\jmpcont\vdash #2 \sto{#4}{#5}\!\!{}^{#1}\, #3}

\newcommand{\nsstep}[5]{#2 \xrightarrow[{#4}]{#5}{}^{#1}\, #3}

\newcommand{\pt}{\pi}

\newcommand\deriv{\mathrel\triangleright}

\newcommand{\env}{\Gamma}
\newcommand{\envtwo}{\Sigma}

\newcommand{\tylow}{\mathsf L}
\newcommand{\tyhigh}{\mathsf H}

\newcommand{\ty}{\tau}

\newcommand{\typair}[1]{ (#1, #1)}

\newcommand{\tty}{\sigma}

\newcommand{\steq}[2][]{\mathrel{\simeq_{#2}^{#1}}}
\newcommand{\confeq}[1]{\approxeq_{#1}}
\newcommand{\boteq}{\mathrel{\sim_\bot}}

\RequirePackage{tabularray}
\UseTblrLibrary{booktabs}

\newcommand{\bnfdef}{\mathrel{::=}}
\newcommand{\bnfmid}{\mid}
\newcommand{\bnfcmt}[1]{#1}
\def\bnfsep{\bnfdef}

\newcommand{\bnfmksep}[1]{\phantom{{}\bnfdef{}}\llap{\text{\ensuremath{{}#1{}}}}{}}
\newcommand{\checkemptydecl}[1]{\ifx\\#1\\\else #1\global\def\bnfsep{{}\bnfdef{}}\fi}\newcommand{\checkemptycase}[1]{{\bnfsep} \ifx\\#1\\ \global\def\bnfsep{\bnfmksep{}}\else{#1}\global\def\bnfsep{\bnfmksep{\bnfmid}}\fi}

\newcommand{\bnfcmtcmd}[1]{\ifx\\#1\\\else\bnfcmt{#1}\fi}
\NewColumnType{H}{X[-1,r,mode=math,cmd=\checkemptydecl]}
\NewColumnType{C}{X[l,mode=math,cmd=\checkemptycase]}
\NewColumnType{D}{X[-1,r,mode=text,cmd=\bnfcmtcmd]}
\newenvironment{bnf}
{\par\medskip
\begin{tblr}{
      colspec = {HCD},
      rowsep={1pt},
      column{1} = {rightsep=0pt},
      column{2} = {leftsep=0pt}
    }}
{\end{tblr}\par\medskip}
\newenvironment{cbnf}{\begin{center}\begin{minipage}{.98\linewidth}\begin{bnf}}{\end{bnf}\end{minipage}\end{center}}

\newcommand*{\unlessempty}[2]{\ifthenelse{\equal{#1}{}}{}{#2}}

\makeatletter
\newcommand{\customlabel}[2]{\protected@write \@auxout {}{\string \newlabel{#1}{{#2}{\thepage}{#2}{#1}{}} }}
\makeatother

\newenvironment{proofcases}
{\begin{list}{\labelitemi}{
 \setlength{\itemsep}{0pt}
 \setlength{\topsep}{0pt}
 \setlength{\parsep}{0pt}
 \setlength{\partopsep}{0pt}
 \setlength{\leftmargin}{10pt}
 \setlength{\rightmargin}{0pt}
 \setlength{\itemindent}{0pt}
 \setlength{\labelsep}{5pt}
 \setlength{\labelwidth}{10pt}}}
{\end{list}}
\newcommand{\proofcase}[1]{\item[-]\textsc{Case} #1.}

\newcommandx{\InferDef}[5][1=,2=,3=]{
  \ifx\\#3\\\def\lab{\textsc{#2}}\def\name{\textsc{#2}}\else \def\lab{#3}\expandarg \def\pattern{\\}
    \def\replace{}
    \expandarg
    \StrSubstitute{#3}{\pattern}{\replace}[\name]
  \fi \infer[\unlessempty{#1#2}{\hypertarget{#1:#2}{\ensuremath{\lab}}}]{#5}{#4}\unlessempty{#1#2}{\customlabel{#1:#2}{[\ensuremath{\name}]}}
}
\newcommandx{\InferStar}[4][1=,2=]{\inferrev[\unlessempty{#2}{\textsc{\ref{#1:#2}}}]{#3}{#4}}
\NewDocumentCommand\Infer{s}{\IfBooleanTF#1\InferStar\InferDef}

\newcommand{\inferrev}[3][]{\infer[#1]{#3}{#2}}

\everymath=\expandafter{\the\everymath\displaystyle}

\newcommand{\repeatable}[1]{
  \newenvironment{#1*}[2][]
  {\expandafter\restatable[##1]{#1}{##2}\label{##2}}
  {\endrestatable}
}

\newcommand{\again}[1]{\csname#1\endcsname*}

\lstdefinestyle{C}{
  language=C,
  basicstyle=\linespread{0.9}\ttfamily,
  keywordstyle=\color{codekwd},
  emphstyle=\color{codefn},
}

\lstnewenvironment{code}[1][]{\lstset{style=C,#1}}{}

\newcommand*\cc[1][]{\lstinline[style=C,emph={#1}]}
 
\unless\iftechreport
\copyrightyear{2026}
\acmYear{2026}
\setcopyright{cc}
\setcctype{by}
\acmConference[CCS '26]{Proceedings of the 2026 ACM SIGSAC Conference on Computer and Communications Security}{November 15--19, 2026}{The Hague, Netherlands}
\acmBooktitle{Proceedings of the 2026 ACM SIGSAC Conference on Computer and Communications Security (CCS '26), November 15--19, 2026, The Hague, Netherlands}
\acmDOI{10.1145/3830454.3832748}
\acmISBN{979-8-4007-2871-6/2026/11}
\fi

\repeatable{theorem}
\repeatable{lemma}
\repeatable{proposition}
\repeatable{corollary}

\begin{document}

\iftechreport
\title{\dfence: Fine-Grained Speculation Barriers for Efficient and Effective Hardware-Software Protection in the Spectre Era (Extended Version)}
\else
\title[\dfence: Fine-Grained Speculation Barriers for Efficient and Effective Hardware-Software Protection]{\dfence: Fine-Grained Speculation Barriers for Efficient and Effective Hardware-Software Protection in the Spectre Era}
\fi

\author{Davide Davoli}
\email{davide.davoli@mpi-sp.org}
\orcid{0009-0009-2981-2962}
\affiliation{\institution{MPI-SP}
  \city{Bochum}
  \country{Germany}
}
\author{Marton Bognar}
\email{marton.bognar@kuleuven.be}
\orcid{0000-0002-8641-7549}
\affiliation{\institution{DistriNet, KU Leuven}
  \city{Leuven}
  \country{Belgium}
}
\author{Lesly-Ann Daniel}
\email{lesly-ann.daniel@eurecom.fr}
\orcid{0000-0002-2772-3722}
\affiliation{\institution{EURECOM}
  \city{Sophia Antipolis}
  \country{France}
}
\author{Benjamin Grégoire}
\email{benjamin.gregoire@inria.fr}
\orcid{0000-0001-6650-9924}
\affiliation{\institution{Inria}
  \city{Sophia Antipolis}
  \country{France}
}
\author{Frank Piessens}
\email{frank.piessens@kuleuven.be}
\orcid{0000-0001-5438-153X}
\affiliation{\institution{DistriNet, KU Leuven}
  \city{Leuven}
  \country{Belgium}
}
\author{Tamara Rezk}
\email{tamara.rezk@inria.fr}
\orcid{0000-0003-3744-0248}
\affiliation{\institution{Inria}
  \city{Sophia Antipolis}
  \country{France}
}

\begin{abstract}
Speculative execution attacks such as Spectre-PHT and Spectre-STL remain a critical security concern in modern processors.
  While software-based mitigations like Speculative Load Hardening (SLH) offer effective protection against Spectre-PHT, they are limited in scope and require software-managed speculative masks, which can be error-prone and costly.
  Defenses against Spectre-STL, such as the Speculative Store Bypass Disable bit (SSBD), incur additional performance overhead and lack fine-grained control.
  In this work, we introduce \dfence, a new CPU instruction that generalizes SLH to mitigate both Spectre-PHT and Spectre-STL with minimal hardware support.
  \dfence enables developers to annotate sensitive registers, with the hardware ensuring that these values do not leak transiently.
  We implement \dfence in the Proteus CPU and evaluate its security and performance, demonstrating less than 1\% average performance overhead for our benchmarks.
  In addition, to support easy and secure adoption, we design a type system that statically verifies the correct placement of \dfence instructions in code.
\end{abstract}

\unless\iftechreport



\fi

\maketitle

\marton{TODO: remove obsolete files and comments (definitely before putting on arxiv)}

\section{Introduction}

After the discovery of Spectre attacks~\cite{spectre}, vendors scrambled to develop rapid software patches to mitigate these hardware-based vulnerabilities.
A common mitigation strategy has been to repurpose synchronization primitives, such as fence instructions, as speculation barriers to block the leakage of secrets during speculative execution.
On x86 systems, this is typically implemented using the \lfence instruction~\cite{x86-manual}, while Arm processors use the \dsb instruction~\cite{armv8} for similar purposes.
However, these full-fence approaches, while effective, impose significant performance overheads when applied broadly across speculative execution paths.

To address this performance penalty, hardware vendors have advocated for more fine-grained solutions.
For example, Arm recommends techniques such as software-based index masking to mitigate Spectre-PHT (Spectre-v1)~\cite{arm-spectrev1}, which selectively restricts speculative execution rather than halting it entirely.
Similarly, Intel and AMD recommend hybrid approaches that combine fences with other software mitigations such as retpolines~\cite{amd-speculation,intel-bcb,intel-lvi,intel-compilers}.
These techniques aim to balance security with performance by limiting the scope of speculation barriers to the necessary minimum.

One of these fine-grained mitigations, Speculative Load Hardening (SLH), was introduced in LLVM~\cite{SpeculativeLoadHardening} to address Spectre-PHT attacks.
SLH is implemented at the compiler level and is based on conditional masking.
SLH maintains a global misspeculation flag in a dedicated register, updated for each branch condition that could be mispredicted.
This misspeculation flag is used to mask operands of unsafe instructions, preventing data leakage during transient execution.
Initially, SLH was used to automatically mask load addresses or values, but this approach was found to be incomplete and was later extended to mask all potential sources of leakage~\cite{ultimate-slh,exorcising-spectres}.
This extended approach incurs an average performance overhead of approximately $150\%$ in the SPEC2017 benchmark when protecting all branches---still a noticeable improvement over the $300\%$ overhead of full fence insertion~\cite{ultimate-slh}.
An alternative approach is to let the user manually insert SLH instrumentations, guided by a type system ensuring that their placement prevents confidential information leakage~\cite{typing-high}.
This strategy reduces the number of inserted protections and the resulting overhead, but it demands manual management of the misspeculation flag and often requires extensive code modifications.

More importantly, SLH can only be applied to Spectre-PHT due to the inherent limitation of detecting mispredictions from software.
For instance, in the case of Spectre-STL (v4), it is impossible to detect from software when a load uses transiently forwarded sensitive data from a prior store.
The vendor-recommended defense for Spectre-STL is Speculative Store Bypass Disable (SSBD), a hardware-based mitigation that fully disables STL speculation, incurring additional performance overhead.
To address the lack of generalization and performance penalties of current mitigations, hardware modifications like taint-tracking solutions~\cite{prospect, spt, stt, nda} have been proposed, but their adoption is hindered by often requiring extensive hardware changes, increasing implementation complexity and cost.

In this work, we investigate the following research question: \emph{Can we apply the principles of SLH to prevent Spectre-STL in addition to Spectre-PHT with minimal hardware changes?}

We answer this question affirmatively by introducing a new assembly instruction called \dfence.\footnote{A play on 'defense', with the `d' loosely suggesting its selective or discerning behavior.}
This instruction generalizes the core principles of SLH beyond Spectre-PHT and makes it more lightweight to use for programmers by alleviating the need to track speculation in software.
The software---following the principle of SLH---can use \dfencex to protect the contents of register $\vx$ from transient leakage.
The hardware tracks the sources of speculative execution and ensures that $\vx$ does not leak speculatively.
Unlike SLH, \dfence eliminates the need to track a speculative mask in software.
This reduces the likelihood of errors, as it removes the need to continuously update the misspeculation flag and provides greater flexibility compared to SLH, where the register holding the misspeculation cannot be spilled.
Additionally, tracking speculation in hardware enables detecting additional types of speculation that are invisible to software (such as Spectre-STL).
To protect programs, we propose and implement a low-level type system for the Jasmin~\cite{jasmin} language, which automatically verifies that \dfence protections (or standard serializing fences) are correctly inserted.
We formally prove the effectiveness of our defense against Spectre-PHT, and variants of Spectre-STL---namely Speculative Store Bypass (SSB) and Predictive Store Forwarding (PSF).
We also discuss extensions to other Spectre variants, including the recent SLAP attack~\cite{slap}, which \dfence can mitigate with very small changes (\cref{sec:discussion}).

To validate our solution, we implement \dfence in a prototype \riscv CPU, Proteus~\cite{bognar2023proteus}, showcasing its low hardware cost and effectiveness in real-world scenarios.
We evaluate its security by using non-interference testing, and its performance on different cryptographic schemes such as ML-DSA~\cite{Dilithium} and ML-KEM~\cite{Kyber}.
We demonstrate that our generalized approach incurs lower performance overhead than the combined costs of SLH and SSBD or serializing fences, offering a practical and scalable alternative.
By achieving robust protection against speculative execution vulnerabilities without prohibitive costs or complexity, our work paves the way for high-assurance software-hardware Spectre defenses.

In summary, our contributions are the following:
\begin{itemize}
  \item Introducing \dfence, generalizing SLH to mitigate Spectre-PHT and -STL. \dfence requires minimal hardware support (\cref{sec:overview}) and simplifies software instrumentation, resulting in fewer constraints and lower complexity (\cref{sec:evaluation});
  \item Developing a formally proven type system that enables verification of correct \dfence insertion (\cref{sec:ts});
  \item Implementing \dfence in the Proteus \riscv core and demonstrating its minimal hardware overhead (\cref{sec:implementation});
  \item Evaluating \dfence's security and performance, providing evidence for protection against Spectre\nobreakdashes-PHT, \nobreakdashes-SSB and \nobreakdashes-PSF, and an average overhead of less than 1\%, smaller than standard protection mechanisms (\cref{sec:evaluation}).
\end{itemize}

\iftechreport
This version of the paper is extended from the conference version~\cite{davoli2026dfence}, containing the operational semantics of \dfence (\Cref{sec:source-language}), additional proofs (\Cref{sec:appsoundness,sec:jta}), and benchmark results (\Cref{app:benchs}). 
\else
An extended version of this paper with the operational semantics of \dfence, additional proofs, and benchmark results is available~\cite{dfence-arxiv}.
\fi

\paragraph{Open science.}
To facilitate reproducing and building on our research, we archive all implementation and evaluation materials at \zenodourl, and will incorporate most components in the upstream Proteus ecosystem at \proteusurl.
 \section{Background}
\label{sec:background}

In this section, we provide some background information on Spectre attacks and their mitigations.

\subsection{Spectre Attacks}
\emph{Microarchitectural attacks} exploit microarchitectural side effects of a program's execution to extract sensitive information.
Any computation affecting the microarchitectural state in a way that can be measured by an attacker reveals information about data processed by a victim (via e.g., timing~\cite{DBLP:conf/crypto/Kocher96}, cache~\cite{bernstein2005cache}, predictor state~\cite{chowdhuryy2021leaking}, micro-op cache~\cite{renSeeDeadMu2021}, and other microarchitectural resource contention~\cite{DBLP:conf/sp/AldayaBHGT19,DBLP:conf/ccs/BhattacharyyaSN19,DBLP:conf/ccs/FustosBY20,DBLP:conf/esorics/0001SLMG19}).
For instance, a memory load $\cmemread \vx {\ar[\expr]}$ affects the cache state based on the accessed address \(\ar[\expr]\).
An attacker can exploit this change through cache side channels to infer the address.
We refer to instruction operands that expose information about their values through microarchitectural side-channels as \emph{unsafe operands}.
For simplicity, in this paper, we assume the set of unsafe operands, a.k.a. the \emph{leakage model}, to be restricted to memory operands and guards of control-flow instructions.
However, our results can trivially be extended to other types of unsafe operands, such as variable-time divisions or multiplications~\cite{DBLP:conf/icisc/GrossschadlOPT09}.

Processors can execute instructions out-of-order and employ various \emph{speculation} mechanisms to predict, for instance, the next instruction to execute.
If the prediction is wrong, the processor simply reverts the incorrectly executed instructions---called \emph{transient} instructions---and continues along the correct path.
Spectre attacks~\cite{spectre} exploit these speculation mechanisms to force a victim to leak secrets during transient execution, as this still produces visible microarchitectural side effects.
By mistraining predictors, a victim can be forced into transiently executing a sequence of instructions chosen to encode secrets in the microarchitectural state, which can later be extracted.
Many Spectre attack variants exist~\cite{DBLP:conf/ccs/MaisuradzeR18,
  DBLP:conf/woot/KoruyehKSA18,spectrev4,SecurityAnalysisAMD2021,
  spectre,
  DBLP:conf/uss/RagabBBG21,DBLP:conf/sp/OleksenkoGKS23}, classified in various
surveys~\cite{DBLP:journals/corr/abs-2309-03376,
  DBLP:conf/uss/CanellaB0LBOPEG19, DBLP:conf/uss/RagabBBG21} according to
which speculation mechanisms they exploit.
In this paper, we focus on \emph{Spectre-PHT} (also known as Spectre-v1)~\cite{spectre} and \emph{Spectre-STL}.
We use Spectre-STL as an umbrella term for \emph{Store-to-Load dependency speculation}, which encompasses two distinct variants: \emph{Spectre-SSB}~\cite{spectrev4,SpeculativeStoreBypass} (also known as Spectre-v4) and \emph{Spectre-PSF}~\cite{SecurityAnalysisAMD2021,FastStoreForwarding}.
We illustrate all these variants in \cref{lst:transient_execution_examples}.

\begin{listing}[t]
\centering
\begin{minipage}{.49\linewidth}
  \begin{lstlisting}[style=nonumbers]
    0 - 15: <@$\ar\expression{[16]}$@>
        16: <@$\sym{s}\expression{[1]}$@>
  17 - 272: <@$\sym{b}\expression{[256]}$@>
\end{lstlisting}
\subcaption[listing]{Memory layout.}\label{lst:memory}
   \begin{lstlisting}[firstnumber=4]
if (i < $\size{\ar}$) $\mathghost$ <@\label{line:pht:transient_cause}@>
  $\cmemread \vx {\ar\expression{[i]}}$ <@\label{line:pht:load}@>
  $\cmemread \vy {\sym{b}\expression{[\vx{}]}}$ $\skull$ <@\label{line:pht:leak}@>
\end{lstlisting}
\subcaption[listing]{Spectre-PHT (Spectre-v1).}\label{lst:spectre-pht}
 \end{minipage} \hfill
\begin{minipage}{.40\linewidth}
  \begin{lstlisting}[firstnumber=7]
$\cmemasgn{\sym{s}\expression{[0]}}{\expression{0}}$ <@\label{line:ssb:sanitize}@>
$\cmemread \vx{} {\sym{s}\expression{[0]}}$ $\mathghost$ <@\label{line:ssb:transient_cause}@>
$\cmemread \vy{} {\sym{b}\expression{[\vx{}]}}$ $\skull$ <@\label{line:ssb:leak}@>
\end{lstlisting}
\subcaption[listing]{Spectre-SSB (Spectre-v4).}\label{lst:spectre-ssb}
   \begin{lstlisting}[firstnumber=10]
$\cmemasgn{\sym{s}\expression{[0]}}{\sym{s}\expression{[0] + 1}}$ <@\label{line:psf:store}@>
$\cmemread \vx{} {\sym{a}\expression{[0]}}$ $\mathghost$ <@\label{line:psf:transient_cause}@>
$\cmemread \vy{} {\sym{b}\expression{[\vx{}]}}$ $\skull$ <@\label{line:psf:leak}@>
\end{lstlisting}
\subcaption[listing]{Spectre-PSF.}\label{lst:spectre-psf}
 \end{minipage} \hfill
\addtocounter{listing}{-1} \addtocounter{lstlisting}{1} \captionof{listing}{Examples of code snippets vulnerable to transient execution attacks. All programs share the memory layout in \cref{lst:memory} where \(\sym{s}\) is secret.
  The symbol \(\protect\mathghost\) indicates an instruction triggering transient execution and $\protect\skull$ indicates leakage.
}\label{lst:transient_execution_examples}
\end{listing}

\textbf{Spectre-PHT} (Pattern History Table) exploits the conditional branch outcome predictor.
In \cref{lst:spectre-pht}, an attacker can train the branch predictor to predict the guard to be true (line~\ref{line:pht:transient_cause}), execute the code with an out-of-bounds index $\expression{i} = 16$ that transiently accesses \(\sym{s}[0]\) (line~\ref{line:pht:load}), and leak it to the microarchitectural state (line~\ref{line:pht:leak}).

\textbf{Spectre-SSB} (Speculative Store Bypass) exploits the ability of a load instruction to speculatively bypass preceding stores with unresolved addresses.
In \cref{lst:spectre-ssb}, the load of \(\sym{s}[0]\) (line~\ref{line:ssb:transient_cause}) may speculatively bypass the prior store to \(\sym{s}[0]\) (line~\ref{line:ssb:sanitize}).
This enables transient access to the secret value before it is zeroed out in memory, leading to its leakage (line~\ref{line:ssb:leak}).

\textbf{Spectre-PSF} (Predictive Store Forwarding) exploits the fact that a store can speculatively forward its value to a subsequent load.
In \cref{lst:spectre-psf}, the store (line~\ref{line:psf:store}) can transiently forward the secret \(\sym{s}\expression{[0] + 1}\) to the next load (line~\ref{line:psf:transient_cause}), which leaks at line~\ref{line:psf:leak}.

\subsection{Countermeasures}
Spectre attacks are possible when secrets speculatively flow to unsafe operands.
Two main policies have been proposed to protect against Spectre~\cite{DBLP:conf/sp/CauligiDMBS22, DBLP:conf/sp/GuarnieriKRV21}: \emph{speculative constant-time}~\cite{DBLP:conf/pldi/CauligiDGTSRB20} and \emph{speculative sandboxing}~\cite{DBLP:conf/sp/GuarnieriKRV21}.

\emph{Speculative constant-time} is an extension of \emph{constant-time programming}, the traditional defense against (non-speculative) microarchitectural attacks, already largely followed by cryptographic libraries~\cite{DBLP:conf/latincrypt/BernsteinLS12, DBLP:conf/uss/AlmeidaBBDE16}.
While constant-time requires that program secrets do not flow to unsafe operands (architecturally), speculative constant-time extends this requirement to speculative paths as well.

\emph{Speculative sandboxing} requires that programs do not speculatively leak data that is not architecturally accessed (in this model, we can consider transiently accessed data to be secret).
Thus, if a program reads out-of-bounds data during speculation, speculative sandboxing forbids this data from speculatively flowing to an unsafe instruction.
Speculative sandboxing is weaker than speculative constant-time as it does not protect architecturally accessed secrets; however, contrary to speculative constant-time, it does not require identifying explicit program secrets and is therefore better suited for protecting general-purpose programs.

Both policies disallow secret data to speculatively flow to unsafe instructions.
To enforce this, a first solution is to insert speculation barriers (\eg{} \lfence) to cut speculative paths before they can leak secrets.
For instance, all programs in \cref{lst:transient_execution_examples} can be protected by inserting a barrier right before the last load.
However, existing fences needlessly restrict speculative execution of \emph{all} instructions after the barrier, increasing the performance overhead.
To alleviate this performance cost, a more efficient defense, \emph{Speculative Load Hardening} (SLH), is usually employed~\cite{SpeculativeLoadHardening,ultimate-slh,typing-high}.
As illustrated in \cref{lst:spectre-slh}, this defense consists of keeping track of possibly misspeculated branch conditions as a mask in a dedicated register (line~\ref{line:slh:maskupdate}) and using this to mask the secrets before they can reach an unsafe operand (line~\ref{line:slh:defense}).

\begin{lstlisting}[firstnumber=1, caption={\Cref{lst:spectre-pht} secured with SLH.}, label={lst:spectre-slh}, float]
$\register{m} \asgn 0\sep \register{c} \asgn \expression{\register{i} < \size{\ar}}$
if ($\expression{\register{c}}$) $\mathghost$ <@\label{line:slh:transient_cause}@>
  $\register{m} \asgn \expression{\register{m} \& (- \register{c})}\sep$<@\label{line:slh:maskupdate}@>
  $\cmemread{\vx}{\ar\expression{[\register{i}]}}$ <@\label{line:slh:load}@> 
  $\vx \asgn \expression{\vx \& \register{m}}\sep$ <@\label{line:slh:defense}@>
  $\cmemread{\vy}{\sym{b}\expression{[\vx{}]}}$ <@\label{line:slh:leak}@> 
\end{lstlisting}

\textbf{Problem:}
SLH faces two significant limitations.
First, its \emph{generality} is restricted: it is effective only against Spectre-PHT; protecting against Spectre-STL requires either inserting costly \fence instructions or disabling SSB and PSF predictions (e.g., using SSBD~\cite{SpeculativeStoreBypass} on Intel processors).
In our benchmarks, SSBD alone introduces up to 56\% overhead.
Second, SLH requires a dedicated register to maintain the misspeculation flag and extra instructions, increasing register pressure and degrading performance.
This effect is reflected in our results, where SLH alone introduces up to 9\% overhead.

\section{Design of \dfence}\label{sec:overview}
We propose a new instruction, \dfencex, which, in the spirit of SLH, can be used to protect a register $\vx$ before it speculatively reaches an unsafe operand.

\begin{mainbox}{Specification of \dfence}
\dfencex acts as a \emph{selective register fence}: it guarantees that its operand $\vx$ is not forwarded to subsequent instructions until this value becomes \emph{non-speculative}.
\end{mainbox}

\cref{lst:dfence_examples} shows how to protect the programs from \cref{lst:transient_execution_examples} with \dfence.
The idea is to protect the register $\vx$, which can transiently contain secret data, before it reaches the unsafe load operand.
The final load instruction can therefore only be executed when $\vx$ becomes non-speculative---\ie{} in \cref{lst:spectre-phtd}, the branch is resolved, and in \cref{lst:spectre-psfd}, dependencies between the load (line~\ref{line:psfd:transient_cause}) and store (line~\ref{line:psfd:store}) are resolved---at which point $\vx$ only contains public data.

\begin{listing}[t]
\refstepcounter{listing}
\centering
\begin{minipage}{.45\linewidth}
  \begin{lstlisting}[firstnumber=1]
if (i < $\size{\ar}$) $\mathghost$ <@\label{line:phtd:transient_cause}@>
  $\cmemread \vx {\ar\expression{[i]}}$ <@\label{line:phtd:load}@>
  $\text{\dfencex} \ \shield{}$ <@\label{line:phtd:dfence}@>
  $\cmemread \vy {\sym{b}\expression{[\vx{}]}}$  <@\label{line:phtd:leak}@>
\end{lstlisting}
\subcaption[listing]{\cref{lst:spectre-pht} protected with \dfence.}\label{lst:spectre-phtd}
 \end{minipage} \hfill
\begin{minipage}{.48\linewidth}
  \begin{lstlisting}[firstnumber=5]
$\cmemasgn{\sym{s}\expression{[0]}}{\sym{s}\expression{[0] + 1}}$ <@\label{line:psfd:store}@>
$\cmemread \vx{} {\sym{a}\expression{[0]}}$ $\mathghost$ <@\label{line:psfd:transient_cause}@>
$\text{\dfencex} \ \shield{}$ <@\label{line:psfd:dfence}@>
$\cmemread \vy{} {\sym{b}\expression{[\vx{}]}}$  <@\label{line:psfd:leak}@>
\end{lstlisting}
\subcaption[listing]{\cref{lst:spectre-psf} protected with \dfence.}\label{lst:spectre-psfd}
 \end{minipage} \hfill
\addtocounter{listing}{-1} \addtocounter{lstlisting}{1} \captionof{listing}{Spectre-PHT and Spectre-PSF gadgets protected with \dfence. Spectre-SSB can be handled similarly to Spectre-PSF.}\label{lst:dfence_examples}
\end{listing}

\subsection{Design Choices}
The \dfence instruction relies on a hardware-software co-design: to achieve end-to-end security, the responsibilities are split between hardware and software, following a security contract~\cite{DBLP:conf/sp/GuarnieriKRV21}.
On the hardware side, in addition to implementing \dfence, hardware vendors must provide a \emph{leakage model} determining unsafe operands.
This leakage model is similar to the list of data-independent-timing instructions already provided by hardware vendors such as AMD and Intel in their DIT/DOIT modes~\cite{armDITDataIndependent, intelDataOperandIndependent}.
On the software side, developers are required to (R1) handle architectural leakage, and (R2) correctly insert \dfence to make sure that secrets do not flow speculatively to unsafe operands.
How to achieve these requirements depends on the policy.
In the case of speculative constant-time, requirement R1 is achieved by following the standard constant-time programming discipline, and requirement R2 is achieved by additionally protecting secrets on speculative paths.
In the case of speculative sandboxing, requirement R1 is achieved by ensuring (architectural) memory safety, and requirement R2 is achieved by protecting the result of potentially speculative loads from leaking.

\paragraph{Type system}
To address the contract requirements for software and assist developers in correctly inserting \dfence protections, we propose a type system that accepts only programs adhering to the constant-time discipline (based on secrecy annotations) with properly placed \dfence instructions.
Our type system considers the results of speculative loads to be secret and ensures that they are protected with \dfence.
This approach guarantees compliance with the speculative constant-time policy (\cref{sec:ts}).
Moreover, our type system still guarantees compliance with the speculative sandboxing policy for general-purpose programs without secrecy annotations.

Our type system tracks the confidentiality level, public ($\tylow$) or secret ($\tyhigh$), of registers and arrays.
For instance, in \cref{lst:dfence_examples}, array $\ar$ initially has type \((\tylow)\), while \(\sym{s}\) has type $(\tyhigh)$.
Registers have two types to account for both \emph{non-speculative} and \emph{speculative} execution.

We assume that programs are memory-safe in normal execution.
Importantly, to account for the fact that load instructions can be memory-unsafe in transient execution, or speculatively get forwarded secret values, the speculative type of their destination register is always set to $\tyhigh$.
For instance, in \cref{lst:dfence_examples}, after the load from array $\ar$ (lines~\ref{line:phtd:load} and~\ref{line:psfd:transient_cause}), $\vx$ is typed \((\tylow,\tyhigh)\): non-speculatively public but speculatively secret.

After a \dfencex instruction, our type system sets the speculative type of $\vx$ to its non-speculative type.
This captures the intuition that, after the fence, the value of \(\vx\) is non-speculative and cannot speculatively contain a secret if \(\vx\) is public.
Therefore, in \cref{lst:dfence_examples}, the \dfencex instructions set the type of \(\vx\) to \((\tylow, \tylow)\) (lines~\ref{line:phtd:dfence} and~\ref{line:psfd:dfence}).

The type system accepts instructions with unsafe operands (\eg loads) only if the unsafe operand is typed $(\tylow, \tylow)$.
Therefore, the insecure programs in \cref{lst:transient_execution_examples}, without \dfence, do not typecheck as \(\vx\) has type $(\tylow, \tyhigh)$.
However, the secured versions in \cref{lst:dfence_examples} do typecheck as \(\vx\) has type $(\tylow, \tylow)$.

We implemented the type system in Jasmin~\cite{jasmin}, a low-level language designed for writing high-speed cryptography.

\paragraph{Hardware}
We propose an implementation of \dfence in hardware (\cref{sec:implementation}). The processor tracks all sources of speculative execution, and upon encountering a \dfencex instruction, delays forwarding $\vx$ until older sources of speculation have been resolved. We also implement and discuss other strategies in \cref{sec:fenceimpl} and \cref{sec:improvements}.

\begin{diff}
  \subsection{Evaluation Setting and Limitations}
  \label{sec:evallim}
  We evaluate \dfence on cryptographic code written in the Jasmin low-level programming language.
  In this setting, our type system makes assumptions about memory safety and predictability of jump targets, which significantly improve the precision of static analysis for the protected program.
  These assumptions are standard and practically relevant in this context, as they typically hold for assembly generated from high-level languages and cryptographic software.
  This choice was motivated by the goal of demonstrating that \dfence can enforce strong security guarantees (SCT) while incurring modest performance costs.

  However, these limitations should not be regarded as inherent to the
  approach: \dfence is not restricted to languages or settings where
  (i) security annotations are available, (ii) jump targets are
  predictable, and (iii) programs are sequentially safe.  When such
  assumptions do not hold, \dfence may be combined with more
  conservative static analyses or target weaker security properties,
  such as relative SCT, as also considered by Blade~\cite{blade}
  (cf. \Cref{sec:blade}) and SLH~\cite{SpeculativeLoadHardening}.  In
  addition, compiler-based hardening passes such as SESES (LLVM's
  \texttt{x86-seses*} flags), USLH~\cite{ultimate-slh}, and
  Eclipse~\cite{Eclipse} use a static analysis to identify where an
  SLH-style protection mechanism should be inserted, providing
  protection guarantees that are often more suitable to
  non-cryptographic code, not requiring security
  annotations. The \dfence instruction can serve as a drop-in
  replacement for the protection primitive in these compiler-based hardening passes,
  likely improving performance (cf. \cref{sec:performance-evaluation}).
  Evaluating \dfence beyond SCT is an interesting area
  for future work.

  Finally, our implementation is currently limited to the
  RISC\nobreakdash-V architecture, specifically the Proteus processor.
  Consequently, the results of our evaluation are representative of
  this platform.  However, \dfence is not inherently tied to
  RISC\nobreakdash-V.  The proposed mechanism is ISA-agnostic and can
  be integrated into other architectures and (commercial)
  processors. \end{diff}

\subsection{Threat Model}
We assume sequential safety---\ie{} memory safety during non-specu\-lative execution---but safety during speculative execution is not required.
For speculative constant-time (but not for speculative sandboxing), we assume a programming model where developers annotate secret data in their program (typical for cryptographic implementations).

Our threat model includes microarchitectural timing attacks.
In particular, an attacker can measure execution time with clock cycle precision and observe changes to the shared microarchitectural state, restricted by the leakage contract.
For simplicity, we assume the standard constant-time leakage model (leaking memory operands and the program counter), but our results can be extended to other leakage models.
This paper focuses on PHT and STL speculation.
Although we model indirect branches, we assume RSB and BTB to be mitigated with complementary defenses (\cf{} \cref{sec:other-spectre-variants}).
We discuss extensions to other forms of speculation in \cref{sec:other-spectre-variants}.

We exclude physical side channels like power consumption or electromagnetic emissions.
These attacks assume stronger attacker models and are subject to orthogonal defenses like masking~\cite{DBLP:conf/mark2/Blakley79,DBLP:journals/cacm/Shamir79}.
Additionally, dynamic frequency throttling attacks (\eg{} Hertzbleed)~\cite{DBLP:conf/uss/WangPHSFK22,DBLP:conf/ccs/LiuCCR22} are left out of scope; this feature should be disabled when computing on secret data.
Data memory-dependent prefetch attacks~\cite{DBLP:conf/uss/ChenWSFKPG24,DBLP:conf/sp/VicarteFPG0FK22} are also out of scope.

\begin{diff}
 \paragraph{Integration with non-cryptographic code.}
The \dfence primitive guarantees that values explicitly marked as secret cannot be transiently propagated to unsafe operands, even in the presence of speculative execution.
These guarantees, however, apply only to the protected and annotated portion of the code.
When cryptographic components are linked with generic non-protected code, secret data may still reside in shared memory and therefore become accessible to unprotected components.
If the program is linked to another program, we assume that the linking process robustly preserves speculative noninterference~\cite{exorcising-spectres}.

This guarantee is largely orthogonal to those provided by \dfence.
The role of \dfence is to prevent secret leakage within the protected cryptographic component itself, whereas the surrounding non-cryptographic code must be secured through complementary mechanisms that prevent the introduction of exploitable speculative gadgets.
This approach is consistent with existing practice, where different protection techniques are combined to secure different layers of the system.
For instance, cryptographic routines can be isolated using compartmentalization mechanisms akin to Intel MPK so that unprotected components cannot access secure compartments either architecturally or speculatively without going through specific API calls~\cite{rct}.

\end{diff}
 
\section{Type System for \dfence}
\label{sec:ts}
\newcommand{\jmpctx}{\Phi}
\newcommand{\sjudg}[4][\jmpctx]{#1\unlessempty{#1}{\mid} #2\vdash #3: #4}

In this section, we propose a type system that ensures that software
is correctly protected using \dfence.  For simplicity, our
formalization focuses on the stronger speculative constant-time
policy.  We discuss sandboxing in \cref{sec:generalization}.

\paragraph{Overview of the type system} Our type system is designed to
integrate with Jasmin~\cite{jasmin}, drawing inspiration
from its existing type system for speculative
constant-time~\cite{typing-high}.  While the Jasmin type system is
tailored for SLH protection and focuses solely on Spectre-PHT, our
type system is designed for \dfence and provides protection against
Spectre-PHT, Spectre-SSB, and Spectre-PSF.
The main idea of the type system is to associate two security levels ($\tylow$ or $\tyhigh$)
with each program variable. The first level corresponds to non-speculative execution, the second to speculative execution.
A register with type $(\tylow,\tylow)$ indicates that its contents are public during both non-speculative and speculative execution,
meaning the value is not dependent on any secret data and can safely leak.
The type $(\tylow,\tyhigh)$ signals that the contents of a register are public during non-speculative execution but may depend
on secret data during speculative execution. In this case, the variable must be protected with a
\dfence before its value can be used as an unsafe operand.
If a variable has type $(\tyhigh,\tyhigh)$, its value is considered secret-dependent in both
execution contexts and must not leak.

\begin{figure}
  \centering
      \begin{cbnf}
        \Expr \ni \expr ,\exprtwo
        & \val{} & values \\
        & \vx & register \\
        & \op (\expr_1,\dots,\expr_n) & operation\\[1mm]
        \Instr \ni \stat,\stattwo
        & \vx \asgn \expr & assignment \\
        & \cmemread \vx {\ar[\expr]} & memory load \\
        & \cmemasgn {\ar[\expr]} \exprtwo & memory store\\
        & \fence  & regular fence\\
        & \text{\dfencex} & selective fence\\
        & \cif{\expr}{\cmd}{\cmdtwo} & conditional \\
        & \cwhile{\expr}{\cmd} & while loop \\[1mm]
        & \kwd{jmp}\ \expr & indirect jump \\[1mm]
        \Cmd \ni \cmd, \cmdtwo & \cnil \bnfmid \stat\sep\cmd & programs
      \end{cbnf}
  \vspace{-3mm}
  \caption{Syntax of the language. Here $\mathtt{v} \in \Vals$ is a value, $\vx \in \Regs$ is a register variable and $\ar \in \Arr$ is an array name.}
  \label{fig:syntax}
\end{figure}

\paragraph{Program language} We formalize our type system over a
simple imperative \texttt{while} language (\Cref{fig:syntax}).  Memory
operations are modeled as accesses to non-overlapping fixed-size
arrays. More precisely, we model memory as a collection of arrays
$\ar\in \Ar$, each one associated with a length, $\size\ar$.  This
approach supports the design of a more precise type system, by
associating each memory operation with the security level of the
corresponding region.

Expressions $\expr \in \Expr$ can either be a
ground value $v\in \Vals$, a register $\vx \in \Regs$, or the
application of an $\nat$-ary operation $\op \in \Op$ on expressions
$\expr_1,\dots,\expr_n$.

Programs $\cmd\in \Cmd$ are sequences of instructions. The assignment
$\vx \asgn \expr$ updates the register $\vx$ with the value of
$\expr$.  The load instruction $\cmemread \vx {\ar[\expr]}$
evaluates $\expr$ to an offset $i$ and loads the $i$-th entry of the
array $\ar$ in the register $\vx$.  The store instruction
$\cmemasgn {\ar[\expr]} \exprtwo$ evaluates $\exprtwo$ and writes its
value to the array location $\ar[\expr]$.  We assume that our programs
enjoy basic memory safety properties: the offsets of arrays always
evaluate to natural numbers $\nat \in \Nat$, and
$0\le \nat \le \size \ar$ always holds in non-speculative execution.  Finally,
our language features structured control flow instructions, the full
speculation barrier \fence, and the fine-grained barrier \dfence.

\iftechreport\else For brevity's sake, we omit the operational semantics of the language,
which can be found in our extended paper~\cite{dfence-arxiv}.\fi

\paragraph*{Types and environments}
Now that we introduced the program language targeted by our type system, we introduce its types and typing environments.
Our \emph{ground types} describe the confidentiality levels of registers and arrays.
A variable is assigned type $\tylow$ if it contains public data, or
$\tyhigh$ if it may contain secret data.
These ground types form a security lattice with the ordering $\tylow \leq \tyhigh$,
indicating that public data can always be treated as secret if necessary.
A type environment $\env$ maps registers $\vx \in \Regs$ to pairs $\tty=(\ty_1 ,\ty_2)$, where
$\ty_1$ and $\ty_2$ represent the security level of $\vx$ during non-speculative and speculative execution, respectively.
Additionally, $\env$ maps arrays $\ar \in \Arr$ to a single security level $\ty$,
corresponding to the security level of the array's contents (i.e., its cells) during non-speculative execution.
Speculative types for arrays are not tracked because the type system assumes that values loaded
speculatively from memory are always secret ($\tyhigh$).
For two given pairs $\tty = (\ty_1, \ty_2)$ and
$\tty' = (\ty_1', \ty_2')$ we say that $\tty \le \tty'$
whenever $\ty_1\le \ty_1'$ and $\ty_2\le \ty_2'$.
We stipulate that two environments satisfy $\env \le \envtwo$,
if and only if for every array or register variable $z$, we have
$\env(z) \le \envtwo(z)$.

\paragraph*{Judgments}
The type system establishes judgments of the form
\( \sjudg[]\env\cmd\envtwo \), which expresses that executing the
program~$\cmd$ in a state of type~$\env$ transitions to a state of
type~$\envtwo$. Moreover, the judgment guarantees that all values
leaked during the execution of~$\cmd$ are public.  To deal with
indirect jumps, our type system internally uses \emph{auxiliary
  judgments} such as \( \sjudg \env\cmd\envtwo \), where the
\emph{jump context} $\jmpctx$ is a set of judgments of the form
\( \sjudg[] \env\cmd\envtwo \).  Intuitively, the context records
previously established typings for programs that may be reached
through indirect jumps.
Ordinary judgments are then defined as auxiliary judgments with
empty jump context $\emptyset$. That is, we write
\( \sjudg[]\env\cmd\envtwo \) as shorthand for
\( \sjudg[\emptyset]\env\cmd\envtwo \).
Jump contexts support compositional reasoning about indirect jumps:
when typing an indirect jump, the type system can look up the typing
of the target continuation in~$\jmpctx$ instead of recomputing it.
This mechanism also allows the type system to handle programs with
cyclic patterns of indirect jumps.

\begin{figure*}[t]
    \columnwidth=\linewidth
  \small
  \par\raggedright\emph{\textbf{\bf \em Typing rules for expressions:}}\\
  \[
    \Infer[TY][TVar][\textsc{TVar}]
    {\sjudg[]{\env}{\vx}{\env(x)}}
    {~}
    \qquad
    \Infer[TY][TVal][\textsc{TVal}]
    {\sjudg[]{ \env}{\val{}}{\typair \tylow}}
    {~}
    \qquad
    \Infer[TY][TOp][\textsc{TOp}]
    {\sjudg[] \env {\op(\expr_1, \ldots, \expr_k)} {\tty}  }
    {
     \sjudg[]{\env}{\expr_i}{\tty}
    }
    \qquad
    \Infer[TY][TSub][\textsc{TSub}]
    {\sjudg[]{\env}{E}{\tty'}}
    {\sjudg[]{\env}{E}{\tty} \quad \tty \le \tty'}
  \]\\[1.5mm]
    {\bf \em Typing rules for instructions:}\\
    \[
      \Infer[TY][TAsgn][\textsc{TAsgn}]
      {\sjudg \env {\vx \asgn \expr} {\supdate \env \vx \tty} }
      {\sjudg[] \env \expr {\tty} }
      \qquad
      \Infer[TY][TLoad][\textsc{TLoad}]
      {
        \sjudg {\env} {\cmemread \vx  {\ar[\expr]}} {\supdate \env \vx {\tty}}
      }
      {
        \sjudg[] {\env} {\expr} {\typair \tylow}
        \quad
        \tty = (\env(\ar), \tyhigh)
      }
      \qquad
      \Infer[TY][TStore][\textsc{TStore}]
      {
        \sjudg {\env} {\cmemasgn  {\ar[\expr]} \exprtwo} {\supdate \env \ar {\max(\pi_1(\tty_\exprtwo),~\env(\ar))}}
      }
      {
        \sjudg[] {\env} {\expr} { \typair \tylow}
        \quad
        \sjudg[] {\env} {\exprtwo} {\tty_\exprtwo}
      }
    \]
    \[
      \Infer[TY][TIf]
      {
        \sjudg {\env} {\cif \expr {\cmd_1} {\cmd_2}} {\env'}
      }
      {
        \sjudg[] {\env} {\expr} { \typair \tylow } \quad
        \sjudg {\env} {\cmd_1} {\env'}\quad
        \sjudg {\env} {\cmd_2} {\env'}\quad
      }
    \qquad
      \Infer[TY][TWhile]
      {
        \sjudg {\env} {\cwhile\expr {\cmd}} {\env}
      }
      {
        \sjudg[] {\env} {\expr} {\typair \tylow}\quad
        \sjudg {\env} {\cmd} {\env}
      }
      \qquad
      \Infer[TY][TJmp]
      {\sjudg \env {\kwd{jmp}\ \expr_{V}} {\textstyle\bigsqcup_{\loc\in V}\envtwo_\loc}}
      {
        \sjudg[] \env {\expr} {\typair \tylow} \quad
        \forall \loc \in V. \left(\sjudg[] \env {\phi(\loc)} {\envtwo_\loc}\right) \in \jmpctx }
    \]
\[
      \Infer[TY][TFence]
      {\sjudg \env {\fence} {\envtwo}}
      {\forall \vx \in \Regs.~\envtwo(\vx) = \typair{\pi_1(\env(\vx))}
        \quad
        \forall \ar \in \Arr.~\envtwo(\ar) = \env(\ar)}
\qquad
      \Infer[TY][TDfence]
      {\sjudg \env {\dfencex} {\supdate \env \vx \tty}}
      {\tty = \typair{\pi_1(\env(\vx))}}
    \]\\[1.5mm]
    {\bf \em Typing rules for programs:}
    \[
      \Infer[TY][TEmpty][\textsc{TEmpty}]
      {\sjudg \env {\cnil} \env}
      {~}
      \ \
      \Infer[TY][TSeq][\textsc{TSeq}]
      {\sjudg {\env_1} {\stat\sep \cmd_2} {\env_3} }
      {
        \sjudg {\env_1} {\stat} {\env_2}
        \quad
        \sjudg {\env_2} {\cmd_2} {\env_3}
      }
      \ \
      \Infer[TY][TWeak][\textsc{TWeak}]
      {
        \sjudg {\env} {\cmd} {\envtwo}
      }
      {
        \sjudg {\env'} {\cmd} {\envtwo'}\quad
        \env \le \env'\quad
        \envtwo' \le \envtwo
      }
      \ \
      \Infer[TY][TCont]
      {\sjudg[\emptyset] {\env} {\cmd} \envtwo}
      {\forall \left( \sjudg[] {\env'} {\cmdtwo} {\envtwo'} \right) \in \jmpctx.\
        \sjudg {\env'} {\cmdtwo} {\envtwo'}
        \quad
        \sjudg {\env} {\cmd} \envtwo
      }
    \]
    \caption{Type system.}
    \label{fig:ts}
  \end{figure*}

\paragraph{Typing rules}
The typing rules are detailed in \Cref{fig:ts}.
The rules for expressions propagate the security levels of the registers
to the resulting value of the expression, producing judgments of the form
$\sjudg[] {\env} {\expr} {\tty}$.
This judgment indicates that evaluating the expression
$\expr$ in a state of type $\env$ will produce a value of type $\tty$.
Rule \ref{TY:TSub} reflects the principle that public data can always be treated as secret.
\begin{diff}
  For simplicity's sake, we assume that applying an operator $\op$ to
  a sequence of sub-expressions takes the same amount of time
  regardless of the values involved, so that no information can leak
  through the timing of expression evaluation itself. This assumption
  can be easily relaxed by imposing that arguments of expressions with
  data-dependent timing have type $\typair\tylow$.
\end{diff}

The remaining rules handle the typing of programs, ensuring that the types of
all variables are consistently tracked at each program point.
Rule \ref{TY:TAsgn} updates the type of the assigned variable
$\vx$ to match the type of the expression $\expr$.
For control flow instructions, the rules follow standard patterns (\ref{TY:TIf}, \ref{TY:TWhile}, \ref{TY:TEmpty}, \ref{TY:TSeq}).
The condition $\sjudg {\env} {\expr} {\typair \tylow}$ in rules \ref{TY:TIf} and \ref{TY:TWhile}
ensures that the value leaked by the conditional instruction is public.
\ref{TY:TWhile} imposes that the environment is a fixed point: the output environment
must be identical to the input environment. This can be achieved using the weakening rule \ref{TY:TWeak}.

The operational semantics
\iftechreport in \Cref{sec:source-language} \else~\cite{dfence-arxiv} \fi
handles indirect
jumps by means of a \emph{continuation map}~$\jmpcont$, associating each program location~$\loc$ with the corresponding continuation~$\jmpcont(\loc)$.
Rule~\ref{TY:TJmp} types indirect jump instructions by assuming
that these are annotated with a set $V$ of
possible jump targets under normal and speculative execution.  The
rule requires all the jump continuations $\jmpcont(\loc)$ for
$\loc \in V$ to be typable, described by the triple
$\left(\sjudg[] \env {\phi(\loc)} {\envtwo}\right) \in \jmpctx$.
This rule is only significant in combination with
Rule~\ref{TY:TCont}. When typing a program $\cmd$, this rule allows
populating a jump context with a collection $\jmpctx$ of
auxiliary judgments  $\sjudg[] {\env'} \cmdtwo {\envtwo'}$. In the premise on
the left, the validity of each judgment in $\jmpctx$ is established
under the assumption that the type of jump targets can be retrieved from
$\jmpctx$. Once the validity of judgments in $\jmpctx$ is established,
these judgments can be used to type $\cmd$ in the right premise.

The most significant rules correspond to loads, stores,
\fence, and \dfence.
Rule~\ref{TY:TLoad} ensures that the array's offset has type $(\tylow,\tylow)$ to prevent leakage due to using secrets as memory addresses.
It also enforces that the new type of register $\vx$ is $(\env(a),\tyhigh)$.
In other words, the type system uses the non-speculative type of
$\ar$ for the register $\vx$ in non-speculative execution and $\tyhigh$ during speculation.
\iftechreport
This approach is justified by several factors;
because of Spectre-PHT, the array access can overflow (rule \ref{SIE:SLoad-OOB}), which means the loaded value might be secret. Even if the access is within bounds, PSF speculation (by rule \ref{SIE:SLoad-PSF})
\else
This approach is justified by several factors;
because of Spectre-PHT, the array access can overflow,  which means the loaded value might be secret. Even if the access is within bounds, PSF speculation
\fi
allows for loading any value from the store buffer.
Consequently, the type system conservatively approximates the speculative type to $\tyhigh$.
Rule~\ref{TY:TStore} similarly requires that the offset has type $(\tylow,\tylow)$.
\begin{diff}
  Since only a single array entry is modified, the new type of the array is the least upper bound of the existing element types and the updated value, i.e.,  $\max(\pi_1(\tty_\exprtwo),~\env(\ar))$,
\end{diff}
where $\pi_1$ is the first projection of the pair of types (i.e., the non-speculative type of $\exprtwo$).

The last two rules account for instructions that (partially) stop speculation.
Rule \ref{TY:TFence} sets the speculative type of all register variables to be equal to their non-speculative type.
This rule is sound because the \fence instruction and any subsequent instructions
can only be executed when the state is not in misspeculation.
\begin{diff}
  Array types are unchanged because our type system does not keep track of their speculative type.
\end{diff}
Rule \ref{TY:TDfence} is similar to Rule \ref{TY:TFence}, but it only updates the protected variable's
speculative type.

\Cref{thm:soundness} states the main property of our type system,
i.e., that well-typed programs do not leak confidential data. This is
captured by the notion of $\confeq\env$-SCT, stating that executing
$\cmd$ on states that store the same values in public memory
regions---i.e., $\confeq\env$-related---yields identical leakage. In
the relation $\confeq\env$, the environment $\env$ specifies
which parts of the state (register variables or arrays) are considered
secret and which of them are public. For brevity's sake, we omit the
complete definitions of $\confeq\env$ and SCT, which are given
in\iftechreport~\Cref{sec:appsoundness}\else~\cite{dfence-arxiv}\fi.

\begin{theorem*}[Correctness of the type system]{thm:soundness}
If \( \sjudg[]\env\cmd\envtwo \) then $\cmd$ is $\confeq\env$-SCT.
\end{theorem*}

\paragraph{Compilation}
Note that \cref{thm:soundness} guarantees that typable programs are speculative constant-time only \textit{at the source level}.
This guarantee will translate to binaries only if the compiler preserves the notion of speculative constant-time.
Proving that secure compilers like Jasmin preserve speculative constant-time is still an active research topic~\cite{DBLP:journals/iacr/OlmosBBGL24}.
In particular, to preserve speculative constant-time, it is crucial that the compiler does not introduce new memory operations (as often happens when register spilling is required).
If spilling is necessary, the programmer (or a policy-aware compiler) should handle it by storing or loading registers into an array, as done in the Jasmin language.
Alternatively, the type system should be applied after the memory operations are introduced.

\paragraph{Example of the type system on a Spectre v1.1 gadget}

\begin{listing}[t]

\centering
\begin{minipage}{\linewidth}
\begin{lstlisting}[style=nonumbers]
    0 - 15: <@$\ar\expression{[16]}$@>
  16 - 271: <@$\sym{sp}\expression{[256]}$@>
\end{lstlisting}
\subcaption[listing]{Memory}\label{lst:memory-spectre-v1.1}
\end{minipage} \hfill

\vspace{3mm}

\begin{minipage}{\linewidth}
\begin{lstlisting}[firstnumber=1]
$\vx \asgn 0$;           // x : (L, L)
$\cmemread \vy {\ar[\vx]}$;           // vulnerable instruction<@\label{line:v1.1vuln}@>
$\vx \asgn \val{secret}$;           // x : (H, H)
if ($\expression{in\_bounds(\ar, i)}$){
  $\cmemasgn {\ar\expression{[i]}} {\expression{addr\_of\_line\_2}}$;<@\label{line:v1oob2}@>
}
${\expression{ra}} \asgn  \expression{sp[0]}$;           // load return address
$\kwd{jmp}\ \expression{ra}$;           // return <@\label{line:v1return}@>
\end{lstlisting}
\subcaption[listing]{Spectre-v1.1}\label{lst:spectre-v1.1}
\end{minipage} \hfill
\addtocounter{listing}{-1} \addtocounter{lstlisting}{1} \captionof{listing}{Snippet vulnerable to Spectre v1.1. In \riscv, returns are encoded as unconditional jumps to the address held in the \texttt{ra} register.}\label{fig:v1.1}
\end{listing}

Spectre v1.1 is a class of transient
execution attacks which relies on transient out-of-bounds store
operations and store-to-load forwarding to leak confidential data in
a victim program~\cite{speculative-overflows}, as demonstrated in \Cref{fig:v1.1}. In
the example, \cc{addr_of_line_2} is a constant that corresponds to the
address of the unsafe instruction at line~\ref{line:v1.1vuln}. An
attacker can induce transient execution of the conditional
instruction with \cc{i} = 16, overlapping \cc{a[i]} with the return
address stored in \cc{sp[0]}.  As a result, the out-of-bounds store at
line~\ref{line:v1oob2} speculatively forwards \cc{addr_of_line_2} to
the return instruction at line~\ref{line:v1return}, redirecting the control flow to line~\ref{line:v1.1vuln}, where the
value of \cc{x} is leaked.

In our type system, programs vulnerable to Spectre v1.1 are not
typable.  Indeed, for typing the program in \Cref{fig:v1.1}, the rule
\ref{TY:TJmp} for the jump instruction in line 9 requires each
target of the jump instruction to typable.  Notice that the program starting in line 2 is reachable, as
the out-of-bounds write instruction may transiently overwrite the
stack pointer. Consequently, by rule~\ref{TY:TJmp}, the instruction at
line~\ref{line:v1.1vuln} must be typable under the assumption that
$\vx$ has type $\typair \tyhigh$, which is not the case.

It is important to note that our type system can only protect against Spectre v1.1 when jump targets can be statically over-approximated. In Jasmin, this requirement is not problematic, since the only form of indirect jumps it supports is return instructions. During normal execution, the set of return targets corresponds to the instructions immediately following each call site. Thus, if all return instructions in a program are preceded by a \dfence instruction protecting the register that holds the return address, then in the absence of return address speculation the set of possible jump targets $V$ remains the same under both speculative and normal execution. \iftechreport This fact is shown in \Cref{sec:jta}\else See~\cite{dfence-arxiv} for a detailed proof of this fact\fi.
This also means that Jasmin programs for \riscv can be protected against
Spectre v1.1 by simply replacing every return instruction with
\cc{dfence ra; jmp ra}.

\paragraph{Software implementation}
We implemented our type system in the Jasmin framework. To this end, we
integrated the \dfence instruction into its language and extended
Jasmin's \riscv compiler to emit the corresponding assembly
instruction. To compare \dfence with Jasmin's manual SLH mechanism, we
extended the compiler with support for SLH and \fence
instructions in \riscv, which were only supported for the x86
architecture.

The primary benefit of our approach is that it enables generating multiple versions of the
same cryptographic primitives, each protected by a different mitigation mechanism
(e.g., SLH, \dfence, or \fence), as demonstrated in \Cref{sec:evaluation}.

\section{Hardware Implementation}
\label{sec:implementation}
To assess the feasibility of \dfence in real hardware and evaluate its performance impact, we implement
it as an extension of the Proteus~\cite{bognar2023proteus} \riscv CPU, which has already been used for evaluating Spectre countermeasures~\cite{prospect}.
We extend version 25.08-O of Proteus featuring out-of-order execution and speculation primitives.
This CPU features both SSB- and PSF-type memory disambiguation speculation, and a branch target buffer (BTB) that can induce both Spectre-PHT and -BTB type speculation.

\subsection{Speculative Execution on Proteus}
Proteus features a reorder buffer (ROB) with a configurable number of entries and register renaming. Execution units are managed by reservation stations that wait for unavailable operands, enforcing data dependency resolution. Load operations are issued to dedicated load units after their target addresses have been computed. Instructions in the ROB are retired in program order when ready, and stores are committed to memory during this retirement phase.

\paragraph{PHT speculation}
Proteus includes a BTB for control-flow speculation.
The BTB is a table that maps partial instruction addresses---independent of instruction type---to predicted target addresses. Additionally, Proteus supports straight-line speculation. These mechanisms are employed in the instruction fetch stage to predict the next instructions to fetch and execute.

Note that the control-flow speculation implemented in Proteus is more permissive than what is assumed in our threat model and formal semantics (\cref{sec:ts}), which only consider PHT-based speculation on conditional branches. As a result, achieving full security on Proteus requires either complementing our defense with an orthogonal BTB-specific mitigation that restricts speculation to PHT-style behavior, or extending our defense as discussed in \cref{sec:other-spectre-variants}.
Nevertheless, this broader speculation model does not undermine our security claims.
Our evaluation (\cref{sec:security-evaluation}) focuses on PHT-style leakage and demonstrates that our defense effectively mitigates these.

\paragraph{SSB and PSF speculation}
Proteus supports both SSB and PSF speculation.
In the case of SSB, a load may execute speculatively even in the presence of an older store with an unresolved address. Once the store instruction retires, a check is performed to detect any address aliasing. If the addresses alias, the store was incorrectly bypassed, and the pipeline is flushed.
The PSF implementation speculatively forwards the value from the most recent store to subsequent loads with unresolved addresses. Once the load address is computed, the actual load address is compared with the address of the store.
If a mismatch is detected, the pipeline is flushed.
These predictors are also equipped with a small history table to avoid repeated mispredictions on the same addresses, similar to how similar predictors can be dynamically disabled in Intel processors~\cite{doweck2006intel}.

\subsection{RISC-V Encoding}
The standard \riscv \fence instructions are encoded as register-immediate (I-type) instructions with separate \texttt{fm} identifiers for \texttt{fence} and \texttt{fence.tso}, and two---currently unused---register identifiers (\texttt{rd} and \texttt{rs1}).
For \dfence, we choose a new \texttt{fm} identifier and refine the semantics to use two registers instead of one.
Concretely, \dfence \texttt{rd}, \texttt{rs1} introduces a dependency between \texttt{rd} and \texttt{rs1}. Instructions using the value of the protected \texttt{rd} will be delayed until the fence forwards its value. Therefore, a register \texttt{x} can be protected using \dfence \texttt{x}, \texttt{x}.
If a different register is used in \texttt{rs1}, the \dfence instruction copies the value of \texttt{rs1} into \texttt{rd} upon retirement. In this case, only instructions that use \texttt{rd} (and not \texttt{rs1}) are blocked from transiently using the value.
This encoding simplifies the hardware changes, as it enables reusing existing dependency tracking mechanisms.
We also extended the LLVM \riscv assembler and disassembler with this encoding to simplify writing code.

\subsection{Speculative Fence Implementations}
\label{sec:fenceimpl}
We implement different speculative execution barriers that comply with our specification (\cref{sec:overview}). These implementations offer the same security, but different trade-offs in terms of performance and hardware costs, which we evaluate in \cref{sec:performance-evaluation} and \cref{sec:hw-evaluation}, respectively.

\paragraph{Serializing fence}
As a baseline, we implemented \fence as a fully serializing instruction, analogous to \lfence on x86. All subsequent instructions are blocked from executing until the fence retires.
This implementation is stricter than the \riscv specification, which only specifies ordering among certain I/O and memory operations. This limited ordering would not prevent leakage through non-memory side channels.

\paragraph{Unoptimized \dfence}
We first implement a \dfence variant that does not rely on hardware speculation tracking. In this design, the value of the destination register is not forwarded until the \dfence instruction retires, regardless of whether the processor is executing speculatively.
This approach simplifies the implementation by eliminating the need for tracking speculation while still complying with our specification, at a potential performance penalty.

\paragraph{Optimized \dfence}
For the optimized \dfence variant, we reuse existing facilities in Proteus for speculation tracking.
Proteus taints instructions executing speculatively, regardless of the source of speculation (BTB, unresolved control-flow instructions, SSB, and PSF). These taints propagate over the pipeline and are recorded in the ROB. When an instruction uses the result of a tainted instruction, it also becomes tainted. Taints in the ROB are cleared upon instruction retirement or a pipeline flush.
A \dfence instruction is blocked in its reservation station before it can start executing if it follows a speculative branch or if its source register is tainted. It is only allowed to execute---and forward the register value---once all prior speculation has been resolved.
This ensures that instructions following the \dfence depending on its output value also cannot start executing (and potentially leak their operands) until the \dfence is retired.
 
\section{Evaluation}
\label{sec:evaluation}
We evaluate our \dfence implementation in multiple aspects: security, runtime overhead compared to other state-of-the-art protections, software instrumentation complexity, and hardware cost.
Our security and performance evaluations are conducted in the cycle-accurate Proteus simulator.

\subsection{Security Evaluation}\label{sec:security-evaluation}
First, we perform a security evaluation to increase our confidence in the correctness of the implementation and validate that, with the correct placement of \dfence instructions, secrets do not leak through microarchitectural channels.
To this end, we develop a template to generate vulnerable code snippets and protect them using various defense strategies. Specifically, our test cases cover:
\begin{enumerate*}
  \item four speculation strategies (SSB, PSF, and two variants of PHT),
  \item leakage through five different insecure instructions (load, store, branch, jump, division), and
  \item instrumentation using secure instructions (\fence, \dfence), and insecure strategies (no protection, or  \dfence with an incorrect register).
\end{enumerate*}
This template generates a total of 40 secure and 40 insecure programs.

Program executions on Proteus produce an output with the time-ordered value changes of the processor's signals, which can be used to detect security violations~\cite{prospect,winderix24ami,winderix24libra}. We conduct our security evaluation by monitoring two sets of signals: a \emph{conservative} set following an allowlist approach and a \emph{liberal} set following a blocklist approach.
The liberal set includes manually selected observable signals (\eg{} cache metadata, memory bus addresses, program counter), while the conservative set includes all signals except those known not to be observable (\eg{} register and cache data).

We run each program twice, using different values for the secret input. Each pair of executions with different secrets $(s_{1}, s_{2})$ produces a pair of signals sets $(\mathsf{signal}(s_{1}), \mathsf{signal}(s_{2}))$. If $\mathsf{signal}(s_{1})\neq \mathsf{signal}(s_{2})$, the program leaks; if $\mathsf{signal}(s_{1}) = \mathsf{signal}(s_{2})$, the program does not leak (for our specific inputs).
Our goal is to demonstrate that (\textbf{G1}) \emph{most} insecure programs leak on Proteus, and (\textbf{G2}) our \fence and \dfence protections effectively close the leaks for \emph{all}  programs using a secure strategy.

\paragraph{Results}
(\textbf{G1}): With both signal sets, 20/40 insecure programs were found to leak in practice. This confirms that these programs are insecure on Proteus and that our security evaluation can accurately identify insecure programs.
(\textbf{G2}:) All 40 programs were found to be secure with both signal sets, providing evidence that our protections effectively close the leaks for all secure programs.

\subsection{Performance Evaluation}\label{sec:performance-evaluation}

We evaluate the runtime overhead of \dfence by comparing it to an
unprotected baseline and standard protections for Spectre-PHT and
Spectre-STL.  As benchmark programs, we use cryptographic algorithms
distributed as part of the \texttt{libjade}
library~\cite{libjade,mldsaJasmin,mlkemJasmin}.
In addition to the existing ML-DSA (Dilithium)~\cite{Dilithium} code, we
ported the Jasmin sources of ML-KEM (Kyber)~\cite{Kyber},
Curve25519~\cite{Curve25519}, Keccak$f$-1600~\cite{SHA-3},
ChaCha20~\cite{ChaCha20}, Poly1305~\cite{Poly1305} and
Gimli~\cite{Gimli} from Arm to 32-bit \riscv.
Notably, Dilithium and Kyber are NIST standards
for post-quantum cryptography since 2024.

\paragraph{Methodology} Our tests are conducted on the \dfence
implementation built on Proteus, as described in
\Cref{sec:implementation}.  Our benchmarks execute the aforementioned cryptographic primitives on Proteus under various
configurations:
\begin{itemize}
\item \emph{Baseline}. The cryptographic
  primitives are executed without applying
  protection mechanisms or disabling speculative execution
  features.

\item \emph{\dfence}. The cryptographic primitives are protected
  exclusively using \dfence, validated by our Jasmin type system.
  These benchmarks are executed on both the unoptimized and optimized
  implementations from \cref{sec:fenceimpl}. To measure the impact of
  defending against Spectre v1.1 with \dfence, we also consider a
  configuration where return instructions are not protected with
  \dfence, named \emph{\dfence without v1.1}.

\item \emph{Regular \fence}. The cryptographic primitives are
  protected solely using regular serializing \fence instructions and
  are validated by our type system. In particular Spectre v1.1 is
  prevented by replacing \riscv \cc{ret} instructions with \cc{fence; ret}.
  We evaluate the impact protecting against Spectre v1.1 by
  also considering a variant of this configuration where return
  instructions are not protected with \fence, named \emph{regular
    \fence without v1.1}.

\item \emph{SLH and SSBD}. Jasmin's existing type system is used to
  insert dependencies between the guards of conditionals and the
  potentially leaked values, following the SLH approach. To achieve
  security comparable to the previous two configurations, we disable
  SSB and PSF to prevent leaks via these prediction
  mechanisms. Additionally, we report on the isolated impact of SLH
  and SSBD. Notably, this configuration does not prevent Spectre v1.1
  attacks that are not covered by the threat model of Jasmin's
  existing type system.
\end{itemize}

Our benchmarks were executed on a setup of Proteus with 48 ROB entries, 12 ALUs, 3 MUL/DIV units, and 3 load units with a 64-cycle memory load delay.

\begin{table}[t]
  \caption{ Geometric mean and maximum overhead measured over the evaluated benchmarks.}
  \label{fig:summary}
  \centering
  \resizebox{\linewidth}{!}{
    \begin{tabular}{lcc}
      \toprule
      \textbf{Configuration}
      &
        \textbf{Geometric mean}
      & \textbf{Maximum}
      \\
      \midrule
      {Optimized \dfence} & -0.23\% & 1.76\% \\
      {Optimized \dfence} w/o v1.1 & 0.13\% & 2.52\% \\
      {Unoptimized \dfence} & -0.19\% & 1.79\% \\
      {Unoptimized \dfence} w/o v1.1 & 0.28\% & 3.55\% \\
      {Regular \fence} & 1.67\% & 4.39\% \\
      {Regular \fence} w/o v1.1 & 1.85\% & 4.73\% \\
      {SLH and SSBD} & 13.53\% & 59.72\% \\
      {SLH only} & 2.45\% & 9.28\% \\
      {SSBD only} & 12.24\% & 56.96\% \\
      \bottomrule
    \end{tabular}
  }
\end{table}

\subsubsection{Results} The results of our performance evaluation are
summarized in \Cref{fig:summary}. \iftechreport Individual measurements
are deferred to \Cref{app:benchs}. \else Individual measurements are included in our extended paper~\cite{dfence-arxiv}. \fi Overall, the experiments show that
the overhead of \dfence consistently low. In
particular, the overhead induced by \dfence to prevent Spectre v1.1
is negative on average for both the variants.
Negative overheads have also been reported for other hardware-based
protection mechanisms that partially restrict speculative execution to
prevent leakage~\cite{SafeSpec,Cassandra,Okapi}. We hypothesize that
in our case, these negative overheads result from \dfence limiting the
incorrect propagation of values during transient execution. This
behavior may reduce the cost associated with squashing transient
instructions compared to the baseline.

Our results also show that the unoptimized implementation of \dfence
shows a good trade-off between implementation cost and performance
overhead: while its performance overhead is greater than the one of
the optimized \dfence, it is still negligible and its hardware cost is
smaller (see \Cref{sec:hw-evaluation}).

Finally, observe that disabling SSB and PSF speculation on Proteus
brings an average performance cost of over 10\%. This cost is unavoidable
when using SLH to still provide protection against Spectre-STL, which is offered
by \dfence out of the box.

\subsection{Software Instrumentation Cost}

\begin{table}[t]
  \caption{Line additions (Lines) and protection sites (Prot) required to protect our benchmarks.
  Percentages are relative to the total LOC of each benchmark.}
  \centering
  \resizebox{\linewidth}{!}{
  \begin{tabular}{lrrrrrrr}
    \toprule
    & & \multicolumn{2}{c}{\dfence}
      & \multicolumn{2}{c}{\fence}
      & \multicolumn{2}{c}{\texttt{SLH}} \\
    \cmidrule(lr){3-4} \cmidrule(lr){5-6} \cmidrule(lr){7-8}
    Benchmark
      & \multicolumn{1}{c}{LOC}
      & \multicolumn{1}{c}{Lines} & \multicolumn{1}{c}{Prot}
      & \multicolumn{1}{c}{Lines} & \multicolumn{1}{c}{Prot}
      & \multicolumn{1}{c}{Lines} & \multicolumn{1}{c}{Prot} \\
    \midrule
    ML-DSA
      & 10{,}755
      & 171 (1.6\%)   & 160 (1.5\%)
      &  93 (0.9\%)   &  91 (0.8\%)
      & 1786 (16.6\%) & 149 (1.4\%) \\
    ML-KEM
      & 2{,}665
      &  98 (3.7\%)  & 51 (1.9\%)
      &  63 (2.4\%)  & 45 (1.7\%)
      & 634 (23.8\%) & 55 (2.1\%) \\
    Curve25519
      & 443
      & 3 (0.7\%) & 3 (0.7\%)
      & 2 (0.5\%) & 1 (0.2\%)
      & 2 (0.5\%) & 1 (0.2\%) \\
    ChaCha20
      & 369
      & 19 (5.1\%) & 10 (2.7\%)
      & 10 (2.7\%) &  2 (0.5\%)
      & 14 (3.8\%) &  2 (0.5\%) \\
    Keccakf1600
      & 384
      & 43 (11.2\%) & 8 (2.1\%)
      & 45 (11.7\%) & 6 (1.6\%)
      & 76 (19.8\%) & 6 (1.6\%) \\
    Poly1305
      & 195
      & 5 (2.6\%) & 4 (2.1\%)
      & 2 (1.0\%) & 1 (0.5\%)
      & 2 (1.0\%) & 1 (0.5\%) \\
    Gimli
      & 87
      & 2 (2.3\%) & 1 (1.1\%)
      & 2 (2.3\%) & 1 (1.1\%)
      & 2 (2.3\%) & 1 (1.1\%) \\
    \bottomrule
  \end{tabular}
  }
  \label{tab:integration-effort}
\end{table}

\cref{tab:integration-effort} reports the number of lines modified to
obtain secure implementations to measure the complexity involved in
instrumenting the different countermeasures.  The effort required to
implement protections using \dfence or \fence is comparable, and is
approximately one 1 line change per protection site, with fluctuations
(see e.g. on Keccakf1600), due to orthogonal aesthetics changes in the
source files. From our measurement, it is also evident that the
integration cost for SLH is significantly higher, primarily due to
maintaining a misspeculation flag throughout the execution and
updating it at every branching instruction and passing it to and from
functions.  This flag must also remain in a register throughout the
execution. In practice, this means that the compiler must treat the
flag as a special variable that cannot be temporarily spilled to the
stack.  This constraint also has the side effect of increasing
register pressure. This is less of an issue for \riscv due to its
large register file, but it can be more problematic on architectures
like the ARM Cortex-M, which have fewer available registers.

\subsection{Hardware Cost}\label{sec:hw-evaluation}
We also perform measurements to determine the cost of implementing \dfence in hardware using the EVAL\nobreakdash-HD workflow~\cite{demeulemeester25hardwarecost}.
Implementing the baseline core with the configuration of \cref{sec:performance-evaluation} uses a chip area of \SI{0.7856}{\milli\meter\squared}.
The unoptimized \dfence implementation measures \SI{0.7872}{\milli\meter\squared}~(+0.2\%), while the optimized dfence with speculation tracking is \SI{0.8053}{\milli\meter\squared} (+2.5\%).
The critical path remains unchanged for all variants at \SI{17.13}{\nano\second}.

\section{Comparison with Blade}
\label{sec:blade}
\newcommand{\tool}[1]{$\mathsf{#1}$}
\newcommand{\soa}{state-of-the-art\xspace}

Blade~\cite{blade} introduces a \cprotect primitive designed to
serve as a fine-grained speculation barrier, with semantics similar to
our \dfence, and two software implementations. The first is a
software-level instrumentation, which ensures that array accesses
remain in bounds, putting Spectre-STL out of scope.
The second one compiles it to
\lfence, resulting in a larger performance overhead~\cite{blade}.

Our work overcomes these limitations by providing two hardware
implementations of \dfence, capable of protecting against multiple
forms of speculation while incurring low instrumentation and performance
cost. Overall, compared to Blade's \cprotect, \dfence
brings several benefits in terms of security and control over the
placement of protections, that we discuss in the following.

\paragraph{Guarantees of the type systems}

Although the type systems of {Blade} and \dfence both ensure the
correct placement of protections, they provide different security
guarantees. The type system of \dfence only accepts programs that are
speculative constant-time. In contrast, the programs accepted by
{Blade}'s type system are speculative constant-time only if they
comply to a \emph{static} notion of constant-time. More precisely,
this notion is an \emph{under-approximation} of constant-time, defined
as a typability condition under a second type system which relies on
security annotations. Therefore, while Blade's authors claim to
\emph{``eliminate speculation-based leakage in unannotated cryptographic
code''}, this is not exactly the case. More precisely, using {Blade} on
unannotated programs that do not comply with the above-mentioned
static notion of constant-time does not guarantee their security at
run time. To illustrate this, consider the following example:

\begin{lstlisting}[caption={Example of program that is typable by {Blade} but not by \dfence.},
  morekeywords={protect},escapechar=~]
if ($\expression E$) { $\vy \asgn \ar\expression{[x]}$;~\label{line:llleeeaaakkk}~}
\end{lstlisting}

In this program, we assume that the guard \cc{E} always evaluates to \cc{false} when
\cc{x} contains secret data. Thus, the program is constant-time. Since
\cc{x} is a register variable, its value is not loaded from memory,
and therefore the program is accepted by {Blade}'s type
system. Consequently, Blade does not insert any \cprotect
instruction on this program.  However, due to direct branch
speculation, the body of the conditional may be transiently executed
with \cc{x} holding secret data that leaks\footnote{ This example does
  not invalidate the soundness of {Blade}'s type system, as the
  program in consideration does not comply with {Blade}'s static
  notion of constant-time. However, it shows that to safely use
  {Blade}, one should ideally first verify that the program complies
  with this notion of constant-time via a second type system.  }. With \dfence's type system, \cc{x} would have the type $\tyhigh$ in
normal execution, as it may store secret data. Therefore, this
vulnerable program would be directly rejected.

\paragraph{Blade's SLH} The main difference between {Blade}'s SLH
implementation and other forms of SLH---like the original LLVM
version and several
variants~\cite{ultimate-slh,typing-high,exorcising-spectres}---is that
{Blade}'s SLH does not keep track of the misspeculation flag in
a dedicated register. Therefore, it cannot selectively mask values
based on the speculative state of the processor. Instead, it masks
arrays accesses based on whether the index is in bounds or
not. Consider the following program:

\begin{lstlisting}[
  caption={{Blade}'s \cc{protect} to prevent a Spectre-v1 vulnerability.},
  morekeywords={protect}, label={lst:bladev1}]
if ($\expression{E}$) {
  $\vx \asgn{}$protect($\expression{\ar[i]}$);
  $\cmemread \vy {\artwo[\vx]}$;
}
\end{lstlisting}

Using {Blade}'s variant of SLH, the \cprotect instruction is replaced with a form of non-speculative index masking:

\begin{lstlisting}[escapechar=~,
  caption={Result of the compilation of \Cref{lst:bladev1} using {Blade}'s SLH.},
  label={lst:bladev1impl}]
if ($\expr$) {
  $\expression{mask} \asgn {\expression{in\_bounds(a,\ i)}}$;~\label{line:bladecheck}~
  $\expression{mask} \asgn  \expression{mask\ ?\ 0xFFFF\ :\ 0x0000}$;~\label{line:cmov}~
  $\vx \asgn \expression{*((a\  +\ i)\  \&\  mask)}$;~\label{line:blademask}~
  $\cmemread \vy  {\expression{b[x]}}$;~\label{line:bladeslhoob}~
}
\end{lstlisting}

Line~\ref{line:bladecheck} checks whether the index
\cc{i} is in bounds for \cc{a}. Then, the \emph{non-speculative}
conditional assignment at line~\ref{line:cmov} sets the variable
\cc{mask} to \cc{0x0000} when \cc{i} is out of bounds, and to
\cc{0xFFFF} otherwise. The load address is then masked at line
\ref{line:blademask} to prevent transient out-of-bounds
accesses.

\paragraph{First limitation: the need for array bounds} {Blade}'s
SLH is only effective when precise array bounds are available. For
instance, if loose bounds are used in \Cref{lst:bladev1impl} at
line~\ref{line:bladecheck}, {Blade}'s SLH does not prevent
the speculative out-of-bounds access at line \ref{line:blademask},
leading to potential leakage of confidential data at line
\ref{line:bladeslhoob}.
In contrast to {Blade}'s SLH, \dfence protects \emph{values}
rather than \emph{accesses}, therefore it does not require array
bounds information. As an example, \Cref{lst:bladev1} can be protected
with \dfence as follows:

\begin{lstlisting}[escapechar=~,
  caption={Program in \Cref{lst:bladev1} protected with \dfence. },
  morekeywords={dfence},
  label={lst:dfencevsblade}]
if ($\expression{E}$) {
  $\cmemread \vx  {\ar\expression{[i]}}$;~\label{line:ladedfenceload1}~
  $\text{\dfencex}$;~\label{line:ladedfence}~
  $\cmemread \vy {\artwo\expression{[x]}}$;~\label{line:ladedfenceload}~
}
\end{lstlisting}
Here, the \dfence instruction at line~\ref{line:ladedfence} prevents
the transient execution of the load at line~\ref{line:ladedfenceload},
thereby preventing the leakage of secret data originating from
out-of-bounds memory loads at line~\ref{line:ladedfenceload1}.

\paragraph{Second limitation: specificity to out-of-bounds accesses}
As {Blade}'s SLH performs index masking, it can only prevent
attacks that exploit out-of-bounds accesses.  To see this, assume
that the condition \cc{E} of \Cref{lst:bladev1impl} always evaluates
to false and that the array \cc{a} contains secret data. Under these
assumptions, the program in \Cref{lst:bladev1impl} trivially does not leak
secret data in normal execution---i.e., it is constant-time---but it
can leak secret data during transient execution. More precisely,
it is sufficient to trigger the transient execution of the
conditional's body with any value of \cc{i} within the bounds of
\cc{a} to cause the leakage of confidential information at
line~\ref{line:bladeslhoob}.

To see how \dfence can prevent this issue, assume the same scenario in
\Cref{lst:dfencevsblade}, i.e., \cc{E} always evaluates to false and
that \cc{a} contains secret data. Even if an attacker triggers the
transient execution of the conditional's body, the leaking instruction
at line~\ref{line:ladedfenceload} would not be executed because
\dfence would withhold the value of \cc{x}.

\paragraph{Protection against Spectre v1.1} Blade can mitigate the
Spectre v1.1 data variant~\cite{blade}, which leaks confidential data
that is stored during transient execution by forwarding it to an
aliasing load that would otherwise access only public data during
normal execution.  However, the variant of Spectre v1.1 illustrated in
\Cref{fig:v1.1}, corresponding to the original formulation of
v1.1~\cite{speculative-overflows}, is not part of Blade's threat model
because Blade's type system does not model indirect jumps or return
instructions, and consequently it imposes no restrictions on jump
targets. With \dfence, we can prevent both variants with a negligible
overhead (\Cref{sec:evaluation}). Using Blade, the geometric mean
overhead for preventing Spectre v1.1 is more than 15\%~\cite{blade}.

\section{Related Work}
\label{sec:relatedwork}
We focus the related work on software and hardware defenses that are suitable at least for Spectre-PHT or -STL.

\paragraph{Speculation barriers in hardware}
Commercial ISAs include speculation barrier instructions designed for various purposes~\cite{x86-manual,armv8,riscv-user}, primarily to address concurrency issues and ensure program correctness.
Intel and AMD recommend using fences---alongside other techniques such as retpolines---to mitigate speculative vulnerabilities~\cite{amd-speculation,intel-bcb,intel-lvi,intel-compilers}.
In contrast to \dfence, which discerns what to block, traditional fences block all forms of speculation, introducing
a performance overhead. Another important caveat with fence-based approaches is that the programmer is often responsible for determining where to insert fences, possibly leading to unnecessary slowdowns or security vulnerabilities.
Our type system (\cref{sec:ts}) is of independent interest from \dfence, as it can be used as an automatic verification tool for the smart insertion of all speculation barriers to effectively protect against Spectre-PHT and STL (as we have used it for our evaluation in \cref{sec:evaluation}).

Building on speculation barriers, multiple academic proposals emerged, introducing fence-like instructions with stronger security guarantees (including this work). Context-sensitive fences~\cite{context-sensitive-fence} are inserted dynamically by the CPU using decoder-level dynamic information flow tracking, reducing the performance overhead compared to \lfence{}s. In contrast to their work, our approach requires fewer hardware modifications.
Temporal fences (\texttt{fence.t})~\cite{fencet} are instructions that flush the microarchitectural state, effectively preventing secret leakage with only a 2\% performance overhead. However, this protection does not address end-to-end timing and same-address-space attacks.

Concurrent to our work, Herinomena et al.~\cite{fence.spec} designed
 and implemented a \riscv instruction similar to \dfence. Their
 proposed instruction, $\kwd{fence.spec}\ \mathtt{rd}, \mathtt{rs1}$,
 establishes a dependency between registers \texttt{rd} and
 \texttt{rs1}, and ``\emph{completes execution only in a
   non-speculative state, i.e. when the core is certain that it will
   commit}''.
To mitigate speculative execution attacks, the authors
 developed a compiler pass that can insert $\kwd{fence.spec}$
 instructions before and/or after every load and/or store.
Due to the coarse-grained nature of these insertion policies,
their performance evaluation reports overheads exceeding 100\%
 in the most conservative policies.  Based on these results, the
 authors conclude that ``\emph{there is no way to use them
   [$\kwd{fence.spec}$ instructions] correctly and efficiently, which
   significantly hinders their practical adoption}''. Our results
 challenge this conclusion by demonstrating that, by using a more
 fine-grained protection strategy, our \dfence implementation can
 prevent speculative attacks based on PHT, STL, and PSF speculation
 with negligible overhead.

\DD{ Intel patent: \url{https://pubchem.ncbi.nlm.nih.gov/patent/US-2023350687-A1}
}

\paragraph{Software-based defenses}

The state-of-the-art software-only mitigations for Spectre-PHT are index-masking and SLH, both recommended by CPU vendors~\cite{arm-spectrev1}.
Academic works have observed that prior SLH techniques were insecure, and hence, proposed extensions of SLH---backed-up by formal guarantees---that still offer better performance than \lfence{}s~\cite{typing-high,ultimate-slh,DBLP:conf/csfw/BaumannBDHH25}. In particular, Shivakumar et al.~\cite{typing-high} define a type
system to verify that protections based on an
optimized form of SLH are correctly applied at the source level.
Their type system only covers PHT speculation, in which
well-typed programs must explicitly track the
value of the misspeculation predicate in a dedicated register, with
the type system ensuring that this register is updated correctly after each
branching instruction.
In comparison, our type system does not impose this
requirement, as the \dfence instruction operates independently of
any software-level misspeculation predicate, and we also cover
SSB and PSF speculation and indirect jumps.
Olmos et al.~\cite{DBLP:journals/iacr/OlmosBBGL24} focus on ensuring that compilers preserve the speculative constant-time property, tackling a critical aspect of compiler interference with security properties.
Other works have also evaluated the performance impact of deployed mitigations on real-world operating systems~\cite{performance-cost}.
More recently, a provably secure software-only approach was proposed that mitigates Spectre-v1, -v2 and -v4 variants at a 2-7\% performance overhead on commercial hardware~\cite{protecting-rsb}.

\paragraph{Hardware-based defenses}
Several works have proposed novel processor models~\cite{spt,stt,prospect,speclfb,nda} to defend against transient execution attacks.
These works can significantly reduce the performance impact (sometimes even leading to {performance gains}~\cite{sequential-contract}) and often achieve stronger security guarantees than fence-insertion techniques~\cite{prospect}.
In particular, from the works using taint tracking, ProSpeCT~\cite{prospect} has been implemented in hardware.
In comparison to our proposal, ProSpeCT has a larger hardware resource overhead, recently estimated at 20.87\% on an older Proteus baseline~\cite{demeulemeester25hardwarecost}.
Additionally, our approach necessitates fewer software changes: in contrast to inserting \dfence instructions, the entire software stack must track secrets stored in secret memory areas with ProSpeCT.
\begin{diff}
  STT~\cite{stt} and SPT~\cite{spt} are other secure speculation mechanisms based on taint tracking that prevent the transient leakage of values that are considered \emph{unsafe}. Although these mechanisms do not impose software modifications, they have a much higher runtime overhead and more expensive hardware changes than \dfence.

  Moreover, while taint-tracking mechanisms implement the protection logic in hardware,  \dfence fully delegates it to software. This delegation also provides more flexibility: in our design, the software can use \dfence to implement different protection mechanisms (e.g., SCT or speculative sandboxing). This is typically not possible with hardware-level taint tracking.
\end{diff}

Other recent work has proposed SpecLFB~\cite{speclfb}, implemented in the BOOM \riscv processor with a hardware overhead of only +0.77\% LUTs and +0.31\% FFs.
SpecLFB is limited to cache side-channel attacks, whereas \dfence can also prevent other side channels such as operand-dependent timing for division or multiplication, or port contention.
OISA~\cite{oisa} is a design also based on BOOM that facilitates writing constant-time (data-oblivious) code, also enforced during speculation, by introducing new leakage-free instructions to the ISA.
While this design enforces strong security guarantees, it introduces significant changes to software (new ISA) and hardware, such as using oblivious RAM~\cite{oram}, and slows down execution by a factor of 3-40 depending on the configuration, making it impractical in many applications.

\section{Discussion}
\label{sec:discussion}
In this section, we discuss extensions of \dfence to other Spectre
variants, other possible improvements, and generalizations of this
work, such as the application of our type system to speculative
sandboxing and automatic insertion mechanisms.

\subsection{Extension to other Spectre Variants}\label{sec:other-spectre-variants}
\paragraph{Load address prediction}
Load address prediction (LAP) has recently been documented in Apple processors and successfully exploited to leak secrets~\cite{slap}. The processor can speculate on the address of a load and incorrectly forward secret data (including non-architecturally accessed data) to subsequent instructions.
This speculation mechanism is similar to PSF and SSB in the sense that a load instruction can transiently forward stale or secret data. Consequently, \dfence and our type system can address Spectre-LAP without any software modifications. On the hardware side, the only required modification is to extend our speculation tracking to include this new speculation source---which should be a simple extension of PSF and SSB tracking.

\paragraph{Control-flow prediction.}
\dfence requires complementary defenses to secure programs in the presence of \emph{arbitrary} control flow speculation, e.g., in the presence of BTB or RSB predictors. In these cases, even if all unsafe operands are protected by a \dfence, the attacker can abuse speculative execution to redirect a victim's control flow to an arbitrary program location, bypassing the \dfence instruction and leaking secrets.

A first solution to handle these forms of speculative execution is to replace indirect jumps and returns with conditional direct jumps. This solution has been shown to incur low overhead on cryptographic benchmarks~\cite{protecting-rsb}.
A second solution is to complement \dfence with existing hardware defenses for (speculative) control-flow integrity~\cite{DBLP:conf/sp/KoruyehSKSA20,IntelHWFeaturesSpeculation}. Intel's Control-flow Enforcement Technology (CET) can restrict the potential targets of control-flow speculation to a set of well-defined landing pads.
On a processor with such a defense, it would be possible to also handle control-flow speculation like Spectre-BTB or Spectre-RSB with \dfence.
In hardware, the speculation-tracking mechanism would need to capture these new speculation sources.
In software, the semantics and type system would need to capture these extra speculative paths and ensure that unsafe operands are protected along them.

\paragraph{Value-based PSF and SSB resolution}
In this paper, we assumed that PSF and SSB are purely based on address speculation and not on value speculation, which is consistent with the Proteus configuration we used.
In theory, PSF and SSB could resolve speculation based on load \emph{values}: only roll back the execution if the predicted load value does not match the actual value, regardless of the memory addresses. This creates a resolution-based timing channel~\cite{stt} dependent on whether the load value and forwarded store values are equal.
Such a value-based resolution mechanism would turn both loads and store values into unsafe operands, breaking the common assumption that load and store values do not leak, and requiring orthogonal defenses.
We are not aware of any currently deployed processor featuring such a value-based resolution mechanism for SSB and PSF speculation.

\paragraph{Other value prediction.}
Value prediction~\cite{DBLP:conf/micro/LipastiS96,gabbaySpeculativeExecutionBased,DBLP:conf/asplos/LipastiWS96} speculates on the value of instructions' operands. For example, Intel CPUs can speculate that the upper 64 bits of a 128-bit dividend are zero.
Another example is load value prediction, which was recently exploited by FLOP~\cite{flop} attacks.
Whether \dfence alone can address value-based speculation depends on which instruction is impacted.
These predictions introduce resolution-based channels~\cite{stt} dependent on whether the predicted value equals the actual value, making predicted operands unsafe. Thus, we need to adapt our leakage model to forbid passing secret values to unsafe operands. Adapting software to comply with the new leakage model might not be feasible---\eg{} if load/store values become unsafe---in which case orthogonal mitigations are needed. For divisions, however, it is possible to consider them as unsafe (many cryptographic implementations already do) and adapt our speculation model accordingly.

\subsection{Possible Improvements}\label{sec:improvements}
\paragraph{Type system}
Our type system could be refined to reduce the number of protected operands.
First, the type system could be equipped with a speculation window to avoid unnecessarily protecting registers outside this window.
Second, depending on the target CPU, we could model memory disambiguation speculation more precisely to identify cases where loads can only access public values speculatively, thereby reducing the number of registers to protect. For example, in the code $\fence; \cmemread \vx {\ar[\vy]}$, the write buffer is empty after the \fence, so if $\vy$ is a public in-bounds address, $\vx$ can only contain public data; yet our current type system conservatively labels $\vx$ as secret.

\paragraph{Hardware}
Any conservative implementation of \dfence \wrt{} its specification (\cref{sec:overview}) is secure, such as a serializing \fence or the unoptimized \dfence from \cref{sec:implementation}.
However, there is a performance trade-off associated with different implementations: more conservative approaches are likely cheaper in hardware, but may unnecessarily decrease instruction throughput.
On the other hand, a more optimized implementation of \dfencex could allow $\vx$ to be forwarded to \emph{safe} operands while blocking forwarding to \emph{unsafe} operands. This approach has the advantage of enabling further speculative out-of-order execution of instructions past the fence. However, it increases hardware complexity, as it must track dependencies between protected registers and the instructions using them. This requires propagating a protection taint and ensuring that unsafe instructions do not speculatively use tainted registers, akin to speculative taint-tracking mechanisms~\cite{stt}. 

\paragraph{End-to-end verification}
There is currently a gap between our Jasmin-level security contract and its hardware implementation. To close this gap and achieve end-to-end security, future work could focus on
(i)
proving that the Jasmin compiler preserves the security contract (SCT), building on recent theoretical foundations and
(ii)
verifying that a hardware implementation satisfies the contract.

\subsection{Generalization}\label{sec:generalization}
\paragraph{Application to speculative sandboxing}
Even though we focused our formalization on speculative constant-time, \dfence can also be used to reason about speculative sandboxing.
We detail here how to apply our type system to this model.
Speculative sandboxing considers architecturally accessed values (or values inside a sandbox) as public and transiently accessed values (or values outside a sandbox) as secret.
Therefore, applying our type system with the memory annotated as such trivially ensures speculative sandboxing.
However, we remark that applying our type system without any secret annotation (\ie{} with the entire memory marked as public) also guarantees speculative sandboxing.
Indeed, assuming the program is architecturally sandboxed, it can only access secrets during speculative execution, via a transient load.
Since our type system types all transient load values as secret, a program can only type-check if transient load values are protected before speculatively flowing to unsafe instructions, which guarantees speculative sandboxing.

\paragraph{Automatic insertion of \dfence} The type system presented
in \Cref{sec:ts} can aid developers in detecting security violations
and deciding where to \emph{manually} insert protections. Building on
this, it is possible to implement an \emph{automated} mechanism that
adopts heuristics to automatically insert \dfence instructions when
violations are detected. For instance, if
$\sjudg[] \Gamma \expr (\tylow,\tyhigh)$ violates the constraint
$\sjudg[] \Gamma \expr \typair\tylow$, such a heuristic can apply
\dfence protections to all variables within $\expr$ with type $(\tylow,\tyhigh)$
right before the violation site. We conducted an initial exploration
where we implemented this and similar heuristics in Jasmin. We
were able to automatically protect Curve25519, Keccak$f$-1600,
ChaCha20, Poly1305 and Gimli. We believe that these preliminary
experiments are promising, and investigating to what extent such an
approach could reduce---or even eliminate---developer effort is
interesting future work.

\paragraph{Other languages}
As we already observed in \Cref{sec:evallim}, while our paper focuses on the Jasmin language, \dfence is not specific to Jasmin and can generally replace SLH.
Moreover, although the implementation of our type system is specific to Jasmin, its core ideas can be adapted to other languages.
For safe languages, information about array bounds and possible jump targets, used by our type system, can be gathered at compile time.
For unsafe languages, adaptation is more difficult, but \dfence can be coupled with control-flow integrity mechanisms, like Intel CET, that restrict jump targets to a set of well-defined landing pads.

\section{Conclusion}

In this paper, we introduced \dfence, a novel instruction generalizing SLH to mitigate both Spectre-PHT and Spectre-STL. We demonstrated that \dfence requires minimal hardware support, simplifies software instrumentation, and provides strong security guarantees. Our formally proven type system ensures correct insertion of \dfence instructions, and our evaluation using the Proteus \riscv core shows that \dfence offers effective security with negligible performance and hardware impact.

\iffalse
In this paper, we propose a hardware defense and reason about how to leverage this defense in software at the source-level, using the Jasmin programming language. Yet, there is a gap between our Jasmin-level security contract and our hardware implementation.\tamara{extend a bit this conclusion}
To close this gap and get closer to end-to-end security, we need to:
\begin{enumerate*}
  \item show that the Jasmin compiler preserves SCT, and
  \item show that our hardware implementation satisfies our security contract.
\end{enumerate*}
We tested these two hypotheses with security tests, but proving them would require solving open research questions.
First, the theoretical foundations for preservation of SCT have been laid very recently~\cite{SNIP,DBLP:journals/iacr/OlmosBBGL24}, and proofs of preservation for practical compilers are an open research topic.
Jasmin~\cite{jasmin,the-last-mile} and CompCert~\cite{Compcert,Compcert1} have been proven to
preserve (non-speculative) constant-time~\cite{CompcertCT,DBLP:conf/ccs/BartheGLP21} and will be obvious targets for practical SCT preservation.
Second, verification of hardware implementations \wrt{} security contracts is also a recent research topic and current technique only scale to in-order cores or very small out-of-order cores~\cite{shadow-logic}.
Proteus is the next step in terms of complexity and our implementation is therefore an obvious target for such research.
\fi
 
\begin{acks}
This research was supported by the Agence Nationale de la Recherche (French National Research Agency) as part of the France 2030 programme -- ANR-22-PECY-0006, the Internal Funds KU Leuven, and the Cybersecurity Research Program Flanders.
\end{acks}

\section*{Ethical considerations}

This research involved only the authors and did not rely on user data or interactions. End-users, companies, and public institutions may benefit from this research. The main ethical principle guiding our work is beneficence, since the results aim to improve computer security without creating risks of harm or rights violations. 

\paragraph{Generative AI usage} In this work, generative AI has been
used as a support for writing and proofreading parts of the paper.

\bibliographystyle{ACM-Reference-Format}
\bibliography{references_clean.bib}

@string{acm = "ACM"}

@string{acmpub = "ACM"}

@string{NY = "New York, NY, USA"}

@string{ieee = "IEEE"}

@string{losalamitos = "Los Alamitos, CA, USA"}

@string{springer = "Sprin\-ger"}

@string{usenixpub = "USENIX Association"}

@string{usenix = "USENIX Association"}

@string{intel = "Intel Corporation"}

@string{intelpub = "Intel"}

@string{arm = "Arm Limited"}

@string{amd = "Advanced Micro Devices"}

@string{amdpub = "AMD"}

@string{riscv = "RISC-V International"}

@string{riscvpub = "RISC-V"}

@string{pacmpl = "Proceedings of the ACM on Programming Languages"}

@string{cacm = "Communications of the ACM"}

@string{iacr = "IACR"}

@inproceedings{davoli2026dfence,
  author    = {Davoli, Davide and Bognar, Marton and Daniel, Lesly-Ann and Gr{\'e}goire, Benjamin and Piessens, Frank and Rezk, Tamara},
  title     = {dfence: Fine-Grained Speculation Barriers for Efficient and Effective Hardware-Software Protection in the Spectre Era},
  booktitle = {CCS},
  year      = {2026},
  month     = {11},
  pages     = {35},
  publisher = {ACM},
  doi       = {10.1145/3830454.3832748},
  url       = {https://doi.org/10.1145/3830454.3832748}
}

@article{dfence-arxiv,
  author       = {Davoli, Davide and Bognar, Marton and Daniel, Lesly-Ann and Gr{\'e}goire, Benjamin and Piessens, Frank and Rezk, Tamara},
  title        = {dfence: Fine-Grained Speculation Barriers for Efficient and Effective Hardware-Software Protection in the Spectre Era (Extended Version)},
  journal      = {CoRR},
  year         = {2026},
  eprinttype   = {arXiv},
}

@techreport{amd-speculation,
  author       = amdpub,
  title        = {Software Techniques for Managing Speculation on {AMD} Processors},
  year         = {2023},
  number       = {Revision 5.09.23},
  type = {White Paper},
  url = "https://www.amd.com/content/dam/amd/en/documents/processor-tech-docs/programmer-references/software-techniques-for-managing-speculation.pdf"
}

@techreport{arm-spectrev1,
  author       = arm,
  title        = {Addressing Spectre Variant 1 (CVE-2017-5753) in Software},
  year         = {2018},
  type         = {White Paper},
  number = {Version 1.0},
  url = {https://developer.arm.com/documentation/102820/latest/}
}

@online{armDITDataIndependent,
  title = {{{DIT}}, {{Data Independent Timing}}},
  author = arm,
  url = {https://developer.arm.com/documentation/ddi0601/2020-12/AArch64-Registers/DIT--Data-Independent-Timing},
  year = {2020},
  urldate = {2025-01-22},
  organization = arm,
}

@manual{armv8,
  author       = arm,
  title        = {Programmer’s Guide for ARMv8-A},
  year         = {2015},
  month         = {3}
}

@misc{bernstein2005cache,
  title={Cache-timing attacks on AES},
  author={Bernstein, Daniel J.},
  url={https://cr.yp.to/antiforgery/cachetiming-20050414.pdf},
  urldate={2025-07-05},
  year={2005}
}

@article{blade,
  author       = {Marco Vas\-sena and
                  Craig Disselkoen and
                  Klaus von Gleissenthall and
                  Sunjay Cauligi and
                  Rami G{\"{o}}khan Kici and
                  Ranjit Jhala and
                  Dean M. Tullsen and
                  Deian Stefan},
  title        = {Automatically eliminating speculative leaks from cryptographic code
                  with blade},
  journal      = pacmpl,
  publisher    = acmpub,
  address      = NY,
  volume       = {5},
  number       = {{POPL}},
  year         = {2021},
  doi          = {10.1145/3434330},
  bibsource    = {dblp computer science bibliography, https://dblp.org}
}

@misc{bognar2023proteus,
  author       = {Bognar, Marton and Noorman, Job and Piessens, Frank},
  title        = {Proteus: An Extensible RISC-V Core for Hardware Extensions},
  howpublished = {Presented at "RISC-V Summit Europe"},
  year         = {2023},
  month        = jun,
  url          = {https://lirias.kuleuven.be/4088721},
  urldate      = {2025-07-05},
  address      = {Barcelona, Spain}
}

@inproceedings{Cassandra,
   title={Cassandra: Efficient Enforcement of Sequential Execution for Cryptographic Programs},
   DOI={10.1145/3695053.3731048},
   booktitle={ISCA},
   publisher=acm,
   author={Hajiabadi, Ali and Carlson, Trevor E.},
   year={2025},
   }

@inproceedings{ChaCha20,
  title={ChaCha, a variant of Salsa20},
  author={Bernstein, Daniel J.},
  booktitle={SASC},
  publisher= {ECRYPT},
  url={https://www.ecrypt.eu.org/stvl/sasc2008/SASCRecord.zip},
  urldate = {2011-05-27},
  volume={8},
  number={1},
  year={2008},
}

@article{chowdhuryy2021leaking,
  author       = {Md Hafizul Islam Chowdhuryy and
                  Fan Yao},
  title        = {Leaking Secrets Through Modern Branch Predictors in the Speculative
                  World},
  journal      = {IEEE Transactions on Computers},
  publisher    = ieee,
  address      = losalamitos,
  volume       = {71},
  number       = {9},
  year         = {2022},
  url          = {https://doi.org/10.1109/TC.2021.3122830},
  bibsource    = {dblp computer science bibliography, https://dblp.org}
}

@inproceedings{context-sensitive-fence,
  author       = {Mohammadkazem Taram and
                  Ashish Venkat and
                  Dean M. Tullsen},
  title        = {Context-Sensitive Fencing: Securing Speculative Execution via Microcode
                  Customization},
  booktitle    = {ASPLOS},
  publisher    = acm,
  year         = {2019},
  doi          = {10.1145/3297858.3304060}
}

@InProceedings{Curve25519,
author="Bernstein, Daniel J.",
title="Curve25519: New Diffie-Hellman Speed Records",
booktitle="PKC",
year="2006",
publisher=springer,
doi={10.1007/11745853_14}
}

@inproceedings{DBLP:conf/asplos/LipastiWS96,
  title = {Value Locality and Load Value Prediction},
  booktitle    = {ASPLOS},
  author = {Lipasti, Mikko H. and Wilkerson, Christopher B. and Shen, John Paul},
  year = {1996},
  doi          = {10.1145/237090.237173},
  publisher = acm,
  }

@inproceedings{DBLP:conf/ccs/BhattacharyyaSN19,
  author       = {Atri Bhattacharyya and
                  Alexandra Sandulescu and
                  Matthias Neugschwandtner and
                  Alessandro Sorniotti and
                  Babak Falsafi and
                  Mathias Payer and
                  Anil Kurmus},
  title        = {SMoTherSpectre: Exploiting Speculative Execution through Port Contention},
  booktitle    = {CCS},
  publisher    = acm,
  year         = {2019},
  doi          = {10.1145/3319535.3363194},
  bibsource    = {dblp computer science bibliography, https://dblp.org}
}

@inproceedings{DBLP:conf/ccs/FustosBY20,
  author       = {Jacob Fustos and
                  Michael Garrett Bechtel and
                  Heechul Yun},
  title        = {SpectreRewind: Leaking Secrets to Past Instructions},
  booktitle    = {ASHES@CCS},
  publisher    = acm,
  year         = {2020},
  doi          = {10.1145/3411504.3421216},
  bibsource    = {dblp computer science bibliography, https://dblp.org}
}

@inproceedings{DBLP:conf/ccs/LiuCCR22,
  author       = {Chen Liu and
                  Abhishek Chakraborty and
                  Nikhil Chawla and
                  Neer Roggel},
  title        = {Frequency Throttling Side-Channel Attack},
  booktitle    = {CCS},
  publisher    = acm,
  year         = {2022},
  doi          = {10.1145/3548606.3560682},
  bibsource    = {dblp computer science bibliography, https://dblp.org}
}

@inproceedings{DBLP:conf/ccs/MaisuradzeR18,
  author       = {Giorgi Maisuradze and
                  Christian Rossow},
  title        = {ret2spec: Speculative Execution Using Return Stack Buffers},
  booktitle    = {CCS},
  publisher    = acm,
  year         = {2018},
  doi          = {10.1145/3243734.3243761},
  bibsource    = {dblp computer science bibliography, https://dblp.org}
}

@InProceedings{DBLP:conf/crypto/Kocher96,
author="Kocher, Paul C.",
title="Timing Attacks on Implementations of Diffie-Hellman, RSA, DSS, and Other Systems",
  doi          = {10.1007/3-540-68697-5_9},
  booktitle    = "CRYPTO",
year="1996",
publisher=springer,
}

@inproceedings{DBLP:conf/esorics/0001SLMG19,
  author       = {Michael Schwarz and
                  Martin Schwarzl and
                  Moritz Lipp and
                  Jon Masters and
                  Daniel Gruss},
  title        = {NetSpectre: Read Arbitrary Memory over Network},
  booktitle    = {ESORICS},
  publisher    = springer,
  year         = {2019},
  doi          = {10.1007/978-3-030-29959-0_14},
  bibsource    = {dblp computer science bibliography, https://dblp.org}
}

@inproceedings{DBLP:conf/icisc/GrossschadlOPT09,
  author       = {Johann Gro{\ss}sch{\"{a}}dl and
                  Elisabeth Oswald and
                  Dan Page and
                  Michael Tunstall},
  title        = {Side-Channel Analysis of Cryptographic Software via Early-Terminating
                  Multiplications},
  booktitle    = {ICISC},
  publisher    = springer,
  year         = {2009},
  doi          = {10.1007/978-3-642-14423-3_13},
  bibsource    = {dblp computer science bibliography, https://dblp.org}
}

@inproceedings{DBLP:conf/latincrypt/BernsteinLS12,
  author       = {Daniel J. Bernstein and
                  Tanja Lange and
                  Peter Schwabe},
  title        = {The Security Impact of a New Cryptographic Library},
  year    = {2012},
  booktitle    = {LATINCRYPT},
  publisher    = springer,
  doi          = {10.1007/978-3-642-33481-8_9},
  bibsource    = {dblp computer science bibliography, https://dblp.org}
}

@inproceedings{DBLP:conf/mark2/Blakley79,
  author       = {G. R. Blakley},
  title        = {Safeguarding cryptographic keys},
  booktitle    = {MARK},
  publisher    = ieee,
  year         = {1979},
  doi          = {10.1109/MARK.1979.8817296},
  bibsource    = {dblp computer science bibliography, https://dblp.org}
}

@inproceedings{DBLP:conf/micro/LipastiS96,
  title = {Exceeding the Dataflow Limit via Value Prediction},
  booktitle    = {MICRO},
  author = {Lipasti, Mikko H. and Shen, John Paul},
  year = {1996},
  doi = {10.1109/MICRO.1996.566464},
  publisher= ieee,
  }

@inproceedings{DBLP:conf/pldi/CauligiDGTSRB20,
  author       = {Sunjay Cauligi and
                  Craig Disselkoen and
                  Klaus von Gleissenthall and
                  Dean M. Tullsen and
                  Deian Stefan and
                  Tamara Rezk and
                  Gilles Barthe},
  title        = {Constant-time foundations for the new spectre era},
  booktitle    = {PLDI},
  publisher    = acm,
  year         = {2020},
  doi          = {10.1145/3385412.3385970},
  bibsource    = {dblp computer science bibliography, https://dblp.org}
}

@inproceedings{DBLP:conf/sp/AldayaBHGT19,
  author       = {Alejandro Cabrera Aldaya and
                  Billy Bob Brumley and
                  Sohaib ul Hassan and
                  Cesar Pereida Garc{\'{\i}}a and
                  Nicola Tuveri},
  title        = {Port Contention for Fun and Profit},
  booktitle    = {IEEE S\&P},
  publisher    = ieee,
  year         = {2019},
  doi          = {10.1109/SP.2019.00066},
  bibsource    = {dblp computer science bibliography, https://dblp.org}
}

@inproceedings{DBLP:conf/sp/CauligiDMBS22,
  author       = {Sunjay Cauligi and
                  Craig Disselkoen and
                  Daniel Moghimi and
                  Gilles Barthe and
                  Deian Stefan},
  title        = {SoK: Practical Foundations for Software Spectre Defenses},
  booktitle    = {IEEE S\&P},
  publisher    = ieee,
  year         = {2022},
  doi          = {10.1109/SP46214.2022.9833707},
  bibsource    = {dblp computer science bibliography, https://dblp.org}
}

@inproceedings{DBLP:conf/sp/GuarnieriKRV21,
  author       = {Marco Guarnieri and
                  Boris K{\"{o}}pf and
                  Jan Reineke and
                  Pepe Vila},
  title        = {Hardware-Software Contracts for Secure Speculation},
  booktitle    = {IEEE S\&P},
  publisher    = ieee,
  year         = {2021},
  doi          = {10.1109/SP40001.2021.00036},
  bibsource    = {dblp computer science bibliography, https://dblp.org}
}

@inproceedings{DBLP:conf/sp/KoruyehSKSA20,
  author       = {Esmaeil Mohammadian Koruyeh and
                  Shirin Haji Amin Shirazi and
                  Khaled N. Khasawneh and
                  Chengyu Song and
                  Nael B. Abu{-}Ghazaleh},
  title        = {SpecCFI: Mitigating Spectre Attacks using {CFI} Informed Speculation},
  booktitle    = {IEEE S\&P},
  publisher    = ieee,
  year         = {2020},
  doi          = {10.1109/SP40000.2020.00033},
  bibsource    = {dblp computer science bibliography, https://dblp.org}
}

@inproceedings{DBLP:conf/sp/OleksenkoGKS23,
  author       = {Oleksii Oleksenko and
                  Marco Guarnieri and
                  Boris K{\"{o}}pf and
                  Mark Silberstein},
  title        = {Hide and Seek with Spectres: Efficient discovery of speculative information
                  leaks with random testing},
  booktitle    = {IEEE S\&P},
  publisher    = ieee,
  year         = {2023},
  doi          = {10.1109/SP46215.2023.10179391},
  bibsource    = {dblp computer science bibliography, https://dblp.org}
}

@inproceedings{DBLP:conf/sp/VicarteFPG0FK22,
  author       = {Jose Rodrigo Sanchez Vicarte and
                  Michael Flanders and
                  Riccardo Paccagnella and
                  Grant Garrett{-}Grossman and
                  Adam Morrison and
                  Christopher W. Fletcher and
                  David Kohlbrenner},
  title        = {Augury: Using Data Memory-Dependent Prefetchers to Leak Data at Rest},
  booktitle    = {IEEE S\&P},
  publisher    = ieee,
  year         = {2022},
  doi          = {10.1109/SP46214.2022.9833570},
  bibsource    = {dblp computer science bibliography, https://dblp.org}
}

@inproceedings{DBLP:conf/uss/AlmeidaBBDE16,
  title = {Verifying Constant-Time Implementations},
  author = {Almeida, José Bacelar and Barbosa, Manuel and Barthe, Gilles and Dupressoir, François and Emmi, Michael},
  year = {2016},
  booktitle    = {USENIX Security},
  url          = {https://www.usenix.org/conference/usenixsecurity16/technical-sessions/presentation/almeida},
  urldate      = {2025-07-04},
  publisher    = usenix,
  }

@inproceedings{DBLP:conf/uss/ChenWSFKPG24,
  author       = {Boru Chen and
                  Yingchen Wang and
                  Pradyumna Shome and
                  Christopher W. Fletcher and
                  David Kohlbrenner and
                  Riccardo Paccagnella and
                  Daniel Genkin},
  title        = {GoFetch: Breaking Con\-stant-Time Cryptographic Implementations Using
                  Data Memory-Dependent Prefetchers},
  url          = {https://www.usenix.org/conference/usenixsecurity24/presentation/chen-boru},
  urldate      = {2025-07-05},
  booktitle    = {USENIX Security},
  publisher    = usenix,
  year         = {2024}
}

@inproceedings{DBLP:conf/uss/RagabBBG21,
  title = {Rage against the Machine Clear: A Systematic Analysis of Machine Clears and Their Implications for Transient Execution Attacks},
  booktitle = {USENIX Security},
  author = {Ragab, Hany and Barberis, Enrico and Bos, Herbert and Giuffrida, Cristiano},
  year = {2021},
  url  = {https://www.usenix.org/conference/usenixsecurity21/presentation/ragab},
  urldate  = {2025-07-05},
  publisher = usenixpub,
  }

@inproceedings{DBLP:conf/uss/WangPHSFK22,
  title = {Hertzbleed: {{Turning}} Power Side-Chan\-nel Attacks into Remote Timing Attacks on X86},
  author = {Wang, Yingchen and Paccagnella, Riccardo and He, Elizabeth Tang and Shacham, Hovav and Fletcher, Christopher W. and Kohlbrenner, David},
  year = {2022},
  url = {https://www.usenix.org/conference/usenixsecurity22/presentation/wang-yingchen},
  urldate   ={2025-07-05},
  booktitle    = {USENIX Security},
  publisher    = usenix,
  }

@inproceedings{DBLP:conf/woot/KoruyehKSA18,
  ids = {koruyehSpectreReturnsSpeculation2018},
  title = {Spectre Returns! {{Speculation}} Attacks Using the Return Stack Buffer},
  booktitle    = {USENIX WOOT},
  author = {Koruyeh, Esmaeil Mohammadian and Khasawneh, Khaled N. and Song, Chengyu and Abu-Ghazaleh, Nael B.},
  url   = {https://www.usenix.org/conference/woot18/presentation/koruyeh},
  urldate  = {2025-07-05},
  year = {2018},
  publisher = usenix,
  }

@article{DBLP:journals/cacm/Shamir79,
  title = {How to Share a Secret},
    doi          = {10.1145/359168.359176},
  author = {Shamir, Adi},
  year = {1979},
  journal = cacm,
  publisher = acm,
  address = NY,
  volume = {22},
  number = {11}
}

@misc{DBLP:journals/corr/abs-2309-03376,
      title={This is How You Lose the Transient Execution War}, 
      author={Allison Randal},
      year={2023},
      eprint={2309.03376},
      archivePrefix={arXiv},
      primaryClass={cs.CR},
}

@article{DBLP:journals/iacr/OlmosBBGL24,
  author       = {Arranz Olmos, Santiago and
                  Gilles Barthe and
                  Lionel Blatter and
                  Benjamin Gr{\'{e}}goire and
                  Vincent Laporte},
  title        = {Preservation of Speculative Constant-Time by Compilation},
  journal      = pacmpl,
  publisher    = acmpub,
  address      = NY,
  volume       = {9},
  number       = {{POPL}},
  year         = {2025},
  doi          = {10.1145/3704880},
  bibsource    = {dblp computer science bibliography, https://dblp.org}
}

@inproceedings{demeulemeester25hardwarecost,
  title     = {Hardware Cost Evaluation in Systems Security},
  author    = {De Meulemeester, Jesse and Norga, Quinten and Piessens, Frank and Verbauwhede, Ingrid and Bognar, Marton},
  year      = 2025,
  booktitle = {ACM REP},
  doi       = {10.1145/3736731.3746155},
}

@article{Dilithium,
  author       = {L{\'{e}}o Ducas and
                  Eike Kiltz and
                  Tancr{\`{e}}de Lepoint and
                  Vadim Lyubashevsky and
                  Peter Schwabe and
                  Gregor Seiler and
                  Damien Stehl{\'{e}}},
  title        = {CRYSTALS-Dilithium: {A} Lattice-Based Digital Signature Scheme},
  journal      = {IACR Transactions on Cryptographic Hardware and Embedded Systems},
  publisher    = {Ruhr-Universität Bochum},
  address      = {Bochum, Germany},  
  volume       = {2018},
  number       = {1},
  year         = {2018},
  doi          = {10.13154/TCHES.V2018.I1.238-268},
  bibsource    = {dblp computer science bibliography, https://dblp.org}
}

@misc{doweck2006intel,
  author={Doweck, Jack},
  title={Intel Smart Memory Access and the Energy-Efficient Performance of the Intel Core Microarchitecture},
  urldate     ={2023-10-06},
  year={2006},
  url = {https://web.archive.org/web/20231026102745/https://www.all-electronics.de/wp-content/uploads/migrated/document/196371/413ei0507-intel-sma.pdf}
}

@inproceedings{Eclipse,
author = {Christou, Neophytos and Gaidis, Alexander J. and Atlidakis, Vaggelis and Kemerlis, Vasileios P.},
title = {Eclipse: Preventing Speculative Memory-error Abuse with Artificial Data Dependencies},
year = {2024},
publisher = acm,
doi = {10.1145/3658644.3690201},
booktitle = {CCS}
}

@inproceedings{exorcising-spectres,
  author       = {Marco Patrignani and
                  Marco Guarnieri},
  title        = {Exorcising Spectres with Secure Compilers},
  booktitle    = {CCS},
  publisher    = acm,
  year         = {2021},
  doi          = {10.1145/3460120.3484534},
  bibsource    = {dblp computer science bibliography, https://dblp.org}
}

@online{FastStoreForwarding,
  title = {Fast {{Store Forwarding Predictor}}},
  author = intelpub,
  url = {https://www.intel.com/content/www/us/en/developer/articles/technical/software-security-guidance/technical-documentation/fast-store-forwarding-predictor.html},
  year = {2022},
  urldate = {2025-01-21},
  langid = {english},
  organization = {Intel}
}

@inproceedings{fence.spec,
  TITLE = {{Exploring speculation barriers for RISC-V selective speculation}},
  AUTHOR = {Andrianatrehina, Herinomena and Lashermes, Ronan and Paturel, Joseph and Rokicki, Simon and Rubiano, Thomas},
  URL = {https://hal.science/hal-05061555},
  urldate = {2025-07-04},
  BOOKTITLE = {{ARES}},
  publisher  = acm,
  YEAR = {2025},
  MONTH = Aug,
  HAL_ID = {hal-05061555},
  HAL_VERSION = {v1},
}

@inproceedings{fencet,
  author       = {Nils Wistoff and
                  Moritz Schneider and
                  Frank K. G{\"{u}}rkaynak and
                  Luca Benini and
                  Gernot Heiser},
  title        = {Microarchitectural Timing Channels and their Prevention on an Open-Source
                  64-bit {RISC-V} Core},
  booktitle    = {DATE},
  publisher    = ieee,
  year         = {2021},
  doi          = {10.23919/DATE51398.2021.9474214}
}

@inproceedings{flop,
  title = {{{FLOP}}: {{Breaking}} the Apple {{M3 CPU}} via False Load Output Predictions},
  url          = {https://www.usenix.org/conference/usenixsecurity24/presentation/chen-boru},
  author = {Kim, Jason and Chuang, Jalen and Genkin, Daniel and Yarom, Yuval},
  urldate      = {2025-07-05},
  booktitle    = {USENIX Security},
  publisher    = usenix,
  year         = {2025}
}

@techreport{gabbaySpeculativeExecutionBased,
  title = {Speculative {{Execution}} Based on {{Value Prediction}}},
  author = {Gabbay, Freddy},
  year   ={1996},
  month   ={11},
  type = {Technical Report},
  number ={1080},
  institution = {Technion --- Israel Institute of Technology, Electrical Engineering Department}
}

@InProceedings{Gimli,
author="Bernstein, Daniel J.
and K{\"o}lbl, Stefan
and Lucks, Stefan
and Massolino, Pedro Maat Costa
and Mendel, Florian
and Nawaz, Kashif
and Schneider, Tobias
and Schwabe, Peter
and Standaert, Fran{\c{c}}ois-Xavier
and Todo, Yosuke
and Viguier, Beno{\^i}t",
title="Gimli : A Cross-Platform Permutation",
booktitle="CHES",
year="2017",
publisher=springer,
  doi          = {10.1007/978-3-319-66787-4_15},
  }

@online{intel-bcb,
  author       = intelpub,
  title        = {Bounds Check Bypass / CVE-2017-5753 / INTEL-SA-00088},
  year         = {2018},
  urldate = {2025-07-04},
  url = {https://www.intel.com/content/www/us/en/developer/articles/technical/software-security-guidance/advisory-guidance/bounds-check-bypass.html}
}

@online{intel-compilers,
  author       = intelpub,
  title        = {Using Intel® Compilers to Mitigate Speculative Execution Side-Channel Issues},
  year         = {2018},
  urldate = {2025-07-04},
  url = "https://www.intel.com/content/www/us/en/developer/articles/troubleshooting/using-intel-compilers-to-mitigate-speculative-execution-side-channel-issues.html"
}

@online{intel-lvi,
  author       = intelpub,
  title        = {An Optimized Mitigation Approach for Load Value Injection},
  year         = {2020},
  urldate = {2025-07-04},
  url = {https://www.intel.com/content/www/us/en/developer/articles/technical/software-security-guidance/best-practices/optimized-mitigation-approach-load-value-injection.html}
}

@online{intelDataOperandIndependent,
  title = {Data {{Operand Independent Timing ISA Guidance}}},
  author = intelpub,
  url = {https://www.intel.com/content/www/us/en/developer/articles/technical/software-security-guidance/best-practices/data-operand-independent-timing-isa-guidance.html},
  urldate = {2025-01-20},
  year = {2025},
  langid = {english},
  organization = intel,
}

@online{IntelHWFeaturesSpeculation,
  title = {Hardware {{Features}} and {{Behaviors Related}} to {{Speculative Execution}}},
  url = {https://www.intel.com/content/www/us/en/developer/articles/technical/software-security-guidance/technical-documentation/hardware-behavior-related-to-speculative-execution.html},
  year = 2024,
  urldate = {2024-12-29},
  organization = intelpub
}

@inproceedings{jasmin,
author = {Almeida, Jos\'{e} Bacelar and Barbosa, Manuel and Barthe, Gilles and Blot, Arthur and Gr\'{e}goire, Benjamin and Laporte, Vincent and Oliveira, Tia\-go and Pacheco, Hugo and Schmidt, Benedikt and Strub, Pierre-Yves},
title = {Jasmin: High-Assurance and High-Speed Cryptography},
year = {2017},
publisher = acm,
doi = {10.1145/3133956.3134078},
booktitle = {CCS}
}

@INPROCEEDINGS{Kyber,
  author={Bos, Joppe and Ducas, Leo and Kiltz, Eike and Lepoint, T and Lyubashevsky, Vadim and Schanck, John M. and Schwabe, Peter and Seiler, Gregor and Stehle, Damien},
  booktitle={IEEE EuroS\&P},
  title={CRYSTALS - Kyber: A CCA-Secure Module-Lattice-Based KEM},
  publisher = ieee,
  year={2018},
  doi={10.1109/EuroSP.2018.00032}
  }

@software{libjade,
  title = {libjade},
  url = {https://github.com/formosa-crypto/libjade},
  year={2025},
  author={Formosa Crypto},
  urldate = {2025-04-11}
}

@software{mldsaJasmin,
  title = {formosa-mlkem},
  url = {https://github.com/formosa-crypto/formosa-mldsa/tree/armv7m-riscv32},
  year={2025},
  author={Formosa Crypto},
  urldate = {2025-04-11}
}

@software{mlkemJasmin,
  title = {formosa-mldsa},
  url = {https://github.com/formosa-crypto/formosa-mlkem},
  year={2025},
  author={Formosa Crypto},
  urldate = {2025-04-11}
}

@inproceedings{nda,
author = {Weisse, Ofir and Neal, Ian and Loughlin, Kevin and Wenisch, Thomas F. and Kasikci, Baris},
title = {NDA: Preventing Speculative Execution Attacks at Their Source},
year = {2019},
publisher = acm,
doi = {10.1145/3352460.3358306},
booktitle = {MICRO},
}

@misc{Okapi,
      title={Okapi: Efficiently Safeguarding Speculative Data Accesses in Sandboxed Environments}, 
      author={Philipp Schmitz and Tobias Jauch and Alex Wezel and Mohammad R. Fadiheh and Thore Tiemann and Jonah Heller and Thomas Eisenbarth and Dominik Stoffel and Wolfgang Kunz},
      year={2024},
      eprint={2312.08156},
      archivePrefix={arXiv},
      primaryClass={cs.CR},
}

@inproceedings{performance-cost,
  author       = {Lucy Bowen and
                  Chris Lupo},
  title        = {The Performance Cost of Software-based Security Mitigations},
  booktitle    = {ICPE},
  publisher    = acm,
  year         = {2020},
  doi          = {10.1145/3358960.3379139}
}

@inproceedings{Poly1305,
  author       = {Daniel J. Bernstein},
  title        = {The Poly1305-AES Message-Authentication Code},
  booktitle    = {FSE},
  publisher    = springer,
  year         = {2005},
  doi          = {10.1007/11502760_3},
  bibsource    = {dblp computer science bibliography, https://dblp.org}
}

@inproceedings{prospect,
  author       = {Lesly{-}Ann Daniel and
                  Marton Bognar and
                  Job Noorman and
                  S{\'{e}}bastien Bardin and
                  Tamara Rezk and
                  Frank Piessens},
  title        = {ProSpeCT: Provably Secure Speculation for the Constant-Time Policy},
  booktitle    = {USENIX Security},
  publisher    = usenix,
  year         = {2023},
  url          = {https://www.usenix.org/conference/usenixsecurity23/presentation/daniel},
  urldate      = {2025-07-05}
}

@inproceedings{protecting-rsb,
author = {Arranz Olmos, Santiago and Barthe, Gilles and Chuengsatiansup, Chitchanok and Gr\'{e}goire, Benjamin and Laporte, Vincent and Oliveira, Tiago and Schwabe, Peter and Yarom, Yuval and Zhang, Zhi\-yuan},
title = {Protecting Cryptographic Code Against Spectre-RSB: (and, in Fact, All Known Spectre Variants)},
year = {2025},
publisher = acm,
doi = {10.1145/3676641.3716015},
volume = {2},
booktitle = {ASPLOS},
}

@manual{riscv-user,
  author       = riscv,
  title        = {The {RISC-V} Instruction Set Manual, Volume I: User-Level {ISA} version 20240411},
  year         = {2024},
  month         = {4},
  organization = riscvpub
}

@INPROCEEDINGS{SafeSpec,
  author={Khasawneh, Khaled N. and Koruyeh, Esmaeil Mohammadian and Song, Chengyu and Evtyushkin, Dmitry and Ponomarev, Dmitry and Abu-Ghazaleh, Nael},
  booktitle={DAC},
  title={SafeSpec: Banishing the Spectre of a Meltdown with Leakage-Free Speculation}, 
  year={2019},
  doi          = {10.1145/3316781.3317903},
  publisher=acm,
  }

@online{SecurityAnalysisAMD2021,
  title = {Security {{Analysis}} of {{AMD Predictive Store Forwarding}}},
  author = amdpub,
  year = {2021},
  url = {https://www.amd.com/system/files/documents/security-analysis-predictive-store-forwarding.pdf},
  urldate = {2022-02-08}
}

@techreport{SHA-3,
  author = {Guido Bertoni and Joan Daemens and Micha\"el Peeters and Gilles Van Assche},
  title = {The Keccak Reference},
  number       = {Version 3.0},
  year = {2011},
  type = {Reference Specification},
  url = {https://keccak.team/files/Keccak-reference-3.0.pdf},
  language = {en},
}

@inproceedings{slap,
  author       = {Jason Kim and
                  Daniel Genkin and
                  Yuval Yarom},
  title        = {{SLAP:} Data Speculation Attacks via Load Address Prediction on Apple
                  Silicon},
  booktitle    = {IEEE S\&P},
  publisher    = ieee,
  year         = {2025},
  doi          = {10.1109/SP61157.2025.00098},
  bibsource    = {dblp computer science bibliography, https://dblp.org}
}

@inproceedings{speclfb,
  author       = {Xiaoyu Cheng and
                  Fei Tong and
                  Hongyu Wang and
                  Zhe Zhou and
                  Fang Jiang and
                  Yuxing Mao},
  title        = {SpecLFB: Eliminating Cache Side Channels in Speculative Executions},
  booktitle    = {USENIX Security},
  publisher    = usenix,
  year         = {2024},
  url          = {https://www.usenix.org/conference/usenixsecurity24/presentation/cheng-xiaoyu},
  urldate      = {2025-07-05},
  bibsource    = {dblp computer science bibliography, https://dblp.org}
}

@inproceedings{spectre,
  author       = {Paul Kocher and
                  Jann Horn and
                  Anders Fogh and
                  Daniel Genkin and
                  Daniel Gruss and
                  Werner Haas and
                  Mike Hamburg and
                  Moritz Lipp and
                  Stefan Mangard and
                  Thomas Prescher and
                  Michael Schwarz and
                  Yuval Yarom},
  title        = {Spectre Attacks: Exploiting Speculative Execution},
  booktitle    = {IEEE S\&P},
  publisher    = ieee,
  year         = {2019},
  doi          = {10.1109/SP.2019.00002}
}

@online{spectrev4,
  title = {Speculative Execution, Variant 4: Speculative Store Bypass},
  author = {Horn, Jann},
  year = {2018},
  url = {https://bugs.chromium.org/p/project-zero/issues/detail?id=1528},
  urldate = {2020-10-12}
}

@misc{speculative-overflows,
      title={Speculative Buffer Overflows: Attacks and Defenses}, 
      author={Vladimir Kiriansky and Carl Waldspurger},
      year={2018},
      eprint={1807.03757},
      archivePrefix={arXiv},
      primaryClass={cs.CR},
}

@online{SpeculativeLoadHardening,
  title={Speculative Load Hardening},
  author = {Carruth, Chandler},
  year = 2018,
  url = {https://llvm.org/docs/SpeculativeLoadHardening.html},
  urldate = {2024-06-07}
}

@online{SpeculativeStoreBypass,
  title = {Speculative {{Store Bypass}} / {{CVE-2018-3639}} / {{INTEL-SA-00115}}},
  url = {https://www.intel.com/content/www/us/en/developer/articles/technical/software-security-guidance/advisory-guidance/speculative-store-bypass.html},
  year = {2018},
  urldate = {2025-01-21},
  author = intelpub
}

@inproceedings{spt,
author = {Choudhary, Rutvik and Yu, Jiyong and Fletcher, Christopher and Morrison, Adam},
title = {Speculative Privacy Tracking (SPT): Leaking Information From Speculative Execution Without Compromising Privacy},
booktitle = {MICRO},
year = {2021},
publisher = acm,
doi = {10.1145/3466752.3480068},
}

@inproceedings{stt,
author = {Yu, Jiyong and Yan, Mengjia and Khyzha, Artem and Morrison, Adam and Torrellas, Josep and Fletcher, Chris\-topher W.},
title = {Speculative Taint Tracking (STT): A Comprehensive Protection for Speculatively Accessed Data},
year = {2019},
publisher = acm,
doi = {10.1145/3352460.3358274},
booktitle = {MICRO},
}

@inproceedings{typing-high,
  author       = {Ammanaghatta Shivakumar, Basavesh and
                  Gilles Barthe and
                  Benjamin Gr{\'{e}}goire and
                  Vincent Laporte and
                  Tiago Oliveira and
                  Swarn Priya and
                  Peter Schwabe and
                  Lucas Tabary{-}Maujean},
  title        = {Typing High-Speed Cryptography against Spectre v1},
  booktitle    = {IEEE S\&P},
  publisher    = ieee,
  year         = {2023},
  doi          = {10.1109/SP46215.2023.10179418}
}

@inproceedings{ultimate-slh,
  author       = {Zhiyuan Zhang and
                  Gilles Barthe and
                  Chitchanok Chuengsatiansup and
                  Peter Schwabe and
                  Yuval Yarom},
  title        = {Ultimate {SLH:} Taking Speculative Load Hardening to the Next Level},
  booktitle    = {USENIX Security},
  publisher    = usenix,
  year         = {2023},
  url          = {https://www.usenix.org/conference/usenixsecurity23/presentation/zhang-zhiyuan-slh},
  urldate      = {2025-07-05}
}

@inproceedings{winderix24ami,
  title     = {Architectural Mimicry: Innovative Instructions to Efficiently Address
               Control-Flow Leakage in Data-Oblivious Programs},
  author    = {Hans Winderix and
               Marton Bognar and
               Job Noorman and
               Lesly{-}Ann Daniel and
               Frank Pies\-sens},
  year      = {2024},
  booktitle = {IEEE S\&P},
  doi       = {10.1109/SP54263.2024.00047},
  publisher = ieee,
  }

@inproceedings{winderix24libra,
  title     = {Libra: Architectural Support For Principled, Secure And Efficient
               Balanced Execution On High-End Processors},
  author    = {Hans Winderix and
               Marton Bognar and
               Lesly{-}Ann Daniel and
               Frank Piessens},
  year      = {2024},
  publisher = acm,
  booktitle = {CCS},
  doi       = {10.1145/3658644.3690319}
}

@manual{x86-manual,
  title={Intel \textregistered 64 and IA-32 Architectures Software Developer's Manual},
  year={2025},
  month={6},
  author=intelpub,
}

@inproceedings{DBLP:conf/csfw/BaumannBDHH25,
  author       = {Jonathan Baumann and
                  Roberto Blanco and
                  L{\'{e}}on Ducruet and
                  Sebastian Harwig and
                  Catalin Hritcu},
  title        = {{FSLH}: Flexible Mechanized Speculative Load Hardening},
  booktitle    = {IEEE CSF},
  publisher    = ieee,
  year         = {2025},
  pages        = {569--584},
  doi          = {10.1109/CSF64896.2025.00023},
  bibsource    = {dblp computer science bibliography, https://dblp.org}
}

@inproceedings{DBLP:conf/uss/CanellaB0LBOPEG19,
  author       = {Claudio Canella and
                  Jo Van Bulck and
                  Michael Schwarz and
                  Moritz Lipp and
                  Benjamin von Berg and
                  Philipp Ortner and
                  Frank Piessens and
                  Dmitry Evtyushkin and
                  Daniel Gruss},
  title        = {A Systematic Evaluation of Transient Execution Attacks and Defenses},
  booktitle    = {USENIX Security},
  publisher    = usenix,
  year         = {2019},
  pages        = {249--266},
  url          = {https://www.usenix.org/conference/usenixsecurity19/presentation/canella},
  urldate      = {2025-07-05}
}

@inproceedings{oisa,
  author       = {Jiyong Yu and
                  Lucas Hsiung and
                  Mohamad El Hajj and
                  Christopher W. Fletcher},
  title        = {Data Oblivious {ISA} Extensions for Side Channel-Resistant and High Performance Computing},
  booktitle    = {NDSS},
  year         = {2019},
  url          = {https://www.ndss-symposium.org/ndss-paper/data-oblivious-isa-extensions-for-side-channel-resistant-and-high-performance-computing/}
}

@inproceedings{oram,
  author       = {Emil Stefanov and
                  Marten van Dijk and
                  Elaine Shi and
                  Christopher W. Fletcher and
                  Ling Ren and
                  Xiangyao Yu and
                  Srinivas Devadas},
  title        = {Path {ORAM}: An Extremely Simple Oblivious {RAM} Protocol},
  booktitle    = {CCS},
  publisher    = acm,
  year         = {2013},
  pages        = {299--310},
  doi          = {10.1145/2508859.2516660}
}

@article{rct,
  author       = {Matthew Kolosick and
                  Basavesh Ammanaghatta Shivakumar and
                  Sunjay Cauligi and
                  Marco Patrignani and
                  Marco Vassena and
                  Ranjit Jhala and
                  Deian Stefan},
  title        = {Robust Constant-Time Cryptography},
  journal      = pacmpl,
  publisher    = acmpub,
  address      = NY,
  volume       = {9},
  number       = {{PLDI}},
  pages        = {1491--1515},
  year         = {2025},
  doi          = {10.1145/3729310},
  bibsource    = {dblp computer science bibliography, https://dblp.org}
}

@inproceedings{renSeeDeadMu2021,
  author       = {Xida Ren and
                  Logan Moody and
                  Mohammadkazem Taram and
                  Matthew Jordan and
                  Dean M. Tullsen and
                  Ashish Venkat},
  title        = {I See Dead $\mu$ops: Leaking Secrets via Intel/AMD Micro-Op Caches},
  booktitle    = {ISCA},
  publisher    = ieee,
  year         = {2021},
  doi          = {10.1109/ISCA52012.2021.00036}
}

@misc{sequential-contract,
      title={Providing High-Performance Execution with a Sequential Contract for Cryptographic Programs},
      author={Ali Hajiabadi and Trevor E. Carlson},
      year={2024},
      eprint={2406.04290},
      archivePrefix={arXiv},
      primaryClass={cs.CR},
}

\appendix

\section{Open Science}

We believe that our prototype is an important contribution, which we release as open software.
At \zenodourl, we archive the following:

\begin{itemize}
\item The prototype implementation of \dfence in Proteus;
\item The Jasmin implementation of the benchmark suite;
\item A version of the Jasmin compiler with support for \dfence instructions;
\item A Docker environment for replicating the evaluation.
\end{itemize}

We will also incorporate most components in the upstream Proteus ecosystem at \proteusurl.

\iftechreport
\section{Speculative Semantics}\label{sec:source-language}

In this section we define a small-step operational semantics for our
programming language.

\renewcommand{\state}{s}

\paragraph{States} The states of our programs are triples
$(\reg, \mbuf, \mem)$, ranged over by the metavariable $\state$.
Here, $\reg: \Regs \to \Vals_\bot$ is a \emph{register map},
$\mem: \Arr \times \Nat \to \Vals$ is a \emph{memory}, and $\mbuf$ is
a \emph{store buffer}. Register maps associate each register variable
with either a value $\val{} \in \Vals$ or $\bot$, which indicates that
the current value of the variable is not available.  Memories map pairs
$(\ar, \nat)$ of an array name $\ar\in \Arr$ and an offset
$\nat \in \Nat$ to a value.  Finally, store buffers are modeled as
sequences of delayed store operations, defined as follows:
\begin{align*}
  \mbuf &\bnfdef \varepsilon \bnfmid [(\ar, \nat) \mapsto \val{}] :\mbuf.
\end{align*}
Here, $\bitem {(\ar, \nat)} {\val{}}$ denotes a pending write of value
$\val{}$ at offset $\nat$ in array $\ar$. In the following, we use
$\supdate{\reg}{x}{\val{}}$ to denote the register file that is
pointwise identical to $\reg$, except for $\vx$, which is updated with
$\val{}\in \Vals_\bot$. Similarly, we write
$\supdate{\mem}{(\ar, \nat)}{\val{}}$ to denote the memory where the
$\nat$-th entry of array $\ar$ is updated with value $\val{}$.  With
slight abuse of notations, we will write $\supdate{\state}{x}{\val{}}$
for $(\supdate{\reg}{x}{\val{}}, \mbuf, \mem)$, and
$\bitem{{(\ar, \nat)}}{\val{}}\cons \state$ for
$(\reg, \bitem {{(\ar, \nat)}} {\val{}}\cons \mbuf, \mem)$.

\paragraph{Semantics of expressions} We associate every $n$-ary
operator $\op$ with an interpretation
$\widehat \op:\Vals^n \to \Vals$.  The semantics of expressions is
defined by:
\small
\begin{align*}
  \sem{\val{}}_{\reg} &\mydef \val{}\\
  \sem\vx_{\reg} &\mydef \reg(\vx) \\
  \sem{\op (\expr_1,\dots,\expr_n)}_{\reg} &\mydef
  \begin{cases}
    \widehat \op (\sem{\expr_1}_{\reg}, \ldots, \sem{\expr_n}_{\reg} ) &\text{if } \forall 1 \le i \le n. \sem {\expr_i}_\reg \neq \bot,\\
    \bot & \text{otherwise.}
  \end{cases}
\end{align*}
\normalsize We write $\sem{\expr}_{\reg} = \val{}$ to indicate that
the expression $\expr$ evaluates to a concrete value $\val{}$, meaning
it does not return $\bot$. For simplicity, given a state
$\state = (\reg, \mbuf, \mem)$ we will write $\sem{\expr}_{\state}$
for $\sem{\expr}_{\reg}$.

\paragraph{Observations and directives} In our semantics transitions
are labeled with \emph{observations} $\obs \in \Obs$ that model the
side-channel leaks performed by the execution step, and directives
$\dir \in \Dir$, that model the attacker-controlled predictions made
by the processor. The set of observations is defined by the following
grammar:
$$
\begin{array}{lcl}
  \Obs \ni \obs & \bnfdef &
  \onone
  \mid \obranch \bool
  \mid \omem {(\ar, \nat)},
\end{array}
$$
where $\bool \in \{\ctrue, \cfalse\}$, $\ar \in \Ar$ and
$\nat\in \Nat$. The observation $\onone$ indicates that a transition
does not leak information, the observation $\obranch \bool$ captures
the leakage caused by conditional branches, and the observation
$\omem {(\ar, \nat)}$ is a pair of an array name and offset that
represents the leakage generated by memory accesses.

The set of directives is described by the following grammar:
\[
  \Dir \ni \dir \bnfdef  \dstep
  \mid \dbranch \bool
  \mid \dload{i}
  \mid \doob {(\ar, \nat)},
\]
where $\bool \in \{\ctrue, \cfalse\}$, $\ar \in \Ar$ and
$i, \nat\in \Nat$. The directive $\dstep$ is used for instructions
that are not subject to speculative predictions. In a conditional
branch, the directive $\dbranch \bool$ instructs the processor to
predict that the value of guard evaluates to $\bool$.  In the presence
of a memory load, the directive $\dload{i}$ controls load value
prediction by instructing the processor to load the $i$-th entry from
the load buffer. Finally, directive $\doob {(\ar, \nat)}$ is used for
speculative out-of-bounds memory accesses, and causes the speculative
load of the $\nat$-th entry of array $\ar$.

\paragraph{Operational Semantics} The \emph{configurations} of our
operational semantics are ranged over by the metavariable $\confone$ and
consist of triples \( {\cmd,\state,\bms} \), where $\cmd$ is a
program, $\state$ is a state, and $\bms$ is the mis-speculation flag:
a Boolean flag that is equal to $\top$ if the processor is executing
transiently.
Transitions take the form
\[
  \sstep{\cmd,\state,\bms}{\cmd',\state',\bms'}{\dir}{\obs},
\]
where ${\cmd,\state,\bms}$ and ${\cmd',\state',\bms'}$ are
configurations, $\dir$ is a directive and $\obs$ is an observation.
To model indirect jumps, transitions are parameterized in a
\emph{continuation map}~$\jmpcont: \Loc \to \Cmd$, which resolves indirect jumps by
associating each program location~$\loc \in \Loc$ with the
corresponding continuation~$\jmpcont(\loc)$, that we assume to be
fixed in the following.

\newcommand{\fn}{\mathtt{f}}

\begin{figure*}[t]
  \centering
  \columnwidth=\linewidth
  \small
    \begin{gather*}
      \Infer[SIE][SAsgn][\textsc{SAsgn}]
      {\sstep {\vx \asgn{\expr}\sep\cmd, \st, \bms}
              {\cmd, \supdate{\st}{\vx}{\val{}}, \bms}
              {\dstep}{\onone}}
              {\sem\expr_{\st} = \val{}}
      \ \
      \Infer[SIE][SLoad-Step][\textsc{SLoad-Step}]
      {\sstep
        {\cmemread{\vx}{\ar[\expr]}\sep\cmd, \st, \bms}
        {\cmd, \supdate{\st}{\vx}{\val{}}, \bms}
        {\dstep} {\omem (\ar, n)}}
      {
        \sem\expr_{\st} = \nat \quad
        \nat \in [0,\size \ar) \quad
        \buflookup {\bm\mbuf\mem} {\ar, \nat}  = \val{}
      }
      \ \
      \Infer[SIE][SStore-Step][\textsc{SStore-Step}]
      {\sstep
        {\cmemasgn {\ar[\expr]} \exprtwo\sep\cmd, \st, \bms}
        {\cmd, \bitem{(\ar, \nat)}{\val{}} \cons \st, \bms}
        {\dstep} {\omem {(\ar, \nat)}}
      }
      {
        \sem\expr_{\st} = \nat \quad
        \nat \in [0,\size \ar) \quad
        \sem\exprtwo_{\st} = \val{}
      }
      \\
      \Infer[SIE][SLoad-OOB][\textsc{SLoad-OOB}]
      {\sstep
        {\cmemread{\vx}{\ar[\expr]}\sep\cmd, \st, \top}
        {\cmd, \supdate{\st}{\vx}{\val{}}, \top}
        {\doob {(\artwo, i)}}{{\omem (\ar, n)}}
      }
      {
        \sem\expr_{\st} = \nat \quad
        \nat \not\in [0,\size \ar) \quad  i \in [0,\size \artwo) \quad
        {\mem}(\artwo,  i) = \val{}
      }
      \qquad
      \Infer[SIE][SStore-OOB][\textsc{SStore-OOB}]
      {\sstep
          {\cmemasgn {\ar[\expr]}{\exprtwo}\sep\cmd, \st, \top}
          {\cmd, \bitem{(\artwo, i)}{\val{}}\cons \st, \top}
          {\doob {(\artwo, i)}} {\omem {(\ar, n)}}
      }
      {
        \sem\expr_{\st} = \nat \quad
        \nat \not\in [0,\size \ar) \quad  i \in [0,\size \artwo) \quad
         \sem\exprtwo_{\st} = \val{}
      }
      \\
\Infer[SIE][SLoad-PSF][\textsc{SLoad-PSF}]
      {\sstep
        {\cmemread \vx  {\ar[\expr]}\sep\cmd, \st, \bms}
        {\cmd, \supdate{\st}{\vx}{\val{}}, \top}
        {\dload{i}} {\omem {(\ar, n)}}}
      {\sem\expr_{\st} = n \quad
        \bufreadi{\mbuf}{i} = \val{}
      }
      \qquad
\Infer[SIE][SLoad-SSB][\textsc{SLoad-SSB}]
      {\sstep
        {\cmemread{\vx}{\ar[\expr]}\sep\cmd, \st, \bms}
        {\cmd, \supdate{\st}{\vx}{\val{}}, \top}
        {\dload{\infty}} {\omem (\ar, n)}}
      {
        \sem\expr_{\st} = \nat \quad
        \nat \in [0,\size \ar) \quad
        {\mem}(\ar, \nat)  = \val{}
      }\\
\Infer[SIE][SWhile]{
      \sstep{W \sep\cmdtwo, \st, \bms}
            {\cmdtwo_{\bool}\sep\cmdtwo, s, \bms\lor \bool \neq{\sem \expr_{\st}}}
            {\dbranch{\bool}}{\obranch{\sem \expr_\st}}
    }
    {
      W = \cwhile \expr {\cmd} \quad
      \cmdtwo_\ctrue = \cmd \sep W \sep\cmdtwo\quad
      \cmdtwo_\cfalse = \cmdtwo
    }
    \qquad
\Infer[SIE][SIf]{
      \sstep{I\sep\cmdtwo, \st, \bms}
            {\cmdtwo_{\bool}\sep\cmdtwo, \st, \bms\lor \bool\neq{\sem \expr_{\st}}}
            {\dbranch{\bool}} {\obranch{\sem \expr_\st}}
    }
    {
      I = \cif \expr {\cmdtwo_\ctrue} {\cmdtwo_\cfalse}
    }
    \\
    \Infer[SIE][Jmp]{
      \sstep{\kwd{jmp}\ \expr, \st, \bms}
            {\phi(\sem \expr_\st), \st, \bms}
            {\dstep} {\ojmp{\sem \expr_\st}}
    }
    {
    }
    \quad
\Infer[SIE][SFence]{
      \sstep{\fence\sep\cmd,(\reg,\mbuf,\mem),\bot}
            {\cmd,(\reg, \epsilon, \overline{\bm\mbuf\mem}), \bot}
            {\dstep}{\onone}
    }{~}
    \quad
    \Infer[SIE][SDfence]{
      \sstep{\dfencex\sep\cmd, \st, \bms}
            {\cmd, \supdate{\st}{\vx}{z}, \bms}
            {\dstep}{\onone}
    }
    {
      z = \mathrm{if}~\bms \mathrm{~then~} \bot \mathrm{~else~} r(x)
    }
    \end{gather*}
    \caption{Speculative semantics.}
  \label{fig:specsem}
\end{figure*}
 
\Cref{fig:specsem} presents the semantic rules for our language. Rule
\ref{SIE:SAsgn} evaluates assignments in the standard way. The rules
\ref{SIE:SLoad-Step} and \ref{SIE:SStore-Step} execute load and store
operations, respectively, without performing any speculation.  The
index $\expr$ must evaluate to an integer $\nat \in \Nat$ within the
bounds of $\ar$.  The function $\buflookup {\bm\mbuf\mem} {\ar, \nat}$
checks (in order) whether there is a pending store at address
$(\ar,\nat)$ in the store buffer $\mbuf$. If so, it returns the
associated value; otherwise, it retrieves the corresponding value from
$\mem$.  For a store operation, the pending store is added to the
store buffer.

Stores and loads can also go wrong in several ways.
First, stores and loads can be out-of-bounds (this is only possible on a misspeculated path, since we assume architectural memory safety). To avoid having to reason about the memory layout, we assume a strong attacker model, where the attacker can then choose the target address to be any other memory location (Rules~\ref{SIE:SLoad-OOB} and~\ref{SIE:SStore-OOB}).
Second, due to PSF, {\em any} load (out-of-bounds or not) can also read any entry from the store buffer (Rule \ref{SIE:SLoad-PSF}). The operator $\bufreadi{\mbuf}{i}$ searches for the
$i$-th pending store in the buffer $\mbuf$ and returns its associated value. If no such store exists, the rule is not applicable.
Third, even an in-bounds  load can read a {\em stale} value from memory due to SSB (Rule \ref{SIE:SLoad-SSB}).
Both \ref{SIE:SLoad-PSF} and \ref{SIE:SLoad-SSB} set the misspeculation flag to $\top$.

In rules \ref{SIE:SWhile} and \ref{SIE:SIf}, the $\dbranch \bool$
directive forces a branch instruction to be evaluated as if the guard resolved to
$\bool$, modeling PHT speculation. If the guard does not actually evaluate to
$\bool$, the misspeculation flag is set.

When an indirect jump $\kwd{jmp}\ \expr$ executes, Rule~\ref{SIE:Jmp}
evaluates the jump target $\expr$ to a program location $\loc$, and
the continuation is set to $\jmpcont(\loc)$.

Finally, Rules~\ref{SIE:SFence} and \ref{SIE:SDfence} address
\fence and \dfence, respectively.  For \fence, the
misspeculation flag must be false.  The store buffer is emptied and
its contents are pushed to memory, denoted $\overline{\bm\mbuf\mem}$
and defined as follows:
\begin{align*}
  \overline {\bm{\nil}\mem} &\mydef \mem &  \overline {({[(\ar, \nat) \mapsto \val{}] : \mbuf}, \mem )} &\mydef \supdate{\overline {\bm{\mbuf}\mem}}{(\ar, \nat)} {\val{}}.
\end{align*}
If the misspeculation flag is true, the semantics become stuck.\footnote{At the assembly level, the semantics do not get stuck; instead, the processor backtracks to a point where no misspeculation has occurred. However, it is unnecessary to model this behavior in order to capture the notion of speculative constant-time.}
In contrast to \fence that blocks all executions, \dfence blocks {\it only} executions depending on the value of $\vx$.
If the misspeculation flag is set, the value of $\vx$
is set to $\bot$, rendering it unavailable for the remainder of the execution;
otherwise, its value remains unchanged.

On top of our semantics, we can define \emph{speculative constant time}:

\begin{definition}[Speculative constant-time] Let $\mathrel{\mathcal R}$ be an
  equivalence relation on states.  A program $\cmd$ is speculative
  constant-time with respect to $\mathrel{\mathcal R}$ (written $\mathrel{\mathcal R}$-SCT)
  iff for every set of directives $\Ds$, misspeculation flag $\bms$,
  every pair of initial states $\st_1$ and $\st_2$ such that
  $\st_1 \mathrel{\mathcal R} ~\st_2$, then for every execution
  $\sstep{\cmd, \st_1, \bms}{\cmd'_1 , \st'_1, \bms_1}{\Ds}{\Os_1}$ and
  $\sstep{\cmd, \st_2, \bms}{\cmd'_2 , \st'_2, \bms_2}{\Ds}{\Os_2}$ we have
  $\Os_1 = \Os_2$.
\end{definition}
Note that both executions follow the same set of directives and begin with the same misspeculation flag.
The relation $\mathrel{\mathcal R}$ is used to relate the two initial states. Typically, it imposes that
the two states must be identical in the portions of the memory and register map containing public data.

\section{Proof of soundness}
\label{sec:appsoundness}

\newcommand{\idx}{i}
\newcommand{\err}{\mathsf{err}}

With an eye on the proof, it is convenient to define a notion of state
as a tuple $\reg, \mbuf, \mem, \bms, \idx$, where $\reg$ is a register
file, $\mbuf$ is a store buffer, $\mem$ is a memory, $\bms$ is the
misspeculation flag and $\idx$ is an index. We also consider the
configuration $\err$ as a state. The set of states is ranged over by
metavariables $\st, \sttwo$.

The proof of soundness relies on \Cref{lemma:mainlemma} below, which can be
seen as a form of \emph{subject reduction} result, stating that
reduction sequences preserve a relation $\confeq{}$ between the
states. Specifically, the relation $\confeq{}$ is parameterized on an
environment $\env$ and is the one that is closed under the rules in
\Cref{fig:confeq}.

\newenvironment{rules}[1][]
{\ifx\\#1\\\else\par\raggedright\emph{\textbf{#1}}\fi
  \[\begin{array}{c}}
{\end{array}\]\vspace{.5mm}}

\begin{figure*}[t]
  \centering
    \columnwidth=\linewidth
    \begin{rules}[Rules for ordinary configurations:]
      \inferrev{
        \reg_1, \mbuf_1, \mem_1, \bms \confeq{\env} \reg_2, \mbuf_2,\mem_2,\bms
      }
      {
        \reg_1 \boteq \reg_2 \quad
        \reg_1 \steq[\bms]{\env}\reg_2\quad
        \mem_1 \steq{\env}\mem_2\quad
        \mbuf_1 \steq[\bms]{\env}\mbuf_2 \quad
}
      \quad
      \inferrev{
        \reg_1 \boteq \reg_2
      }
      {
        \forall \vx \in \Regs. \reg_1(\vx) = \bot \Leftrightarrow \reg_2(\vx) = \bot
      }
      \quad
      \inferrev{
        \mem_1 \steq{\env}\mem_2
      }
      {
        \forall \ar \in \Arr. \env(\ar) = \tylow \Rightarrow \mem_1(\ar)=\mem_2(\ar)
      }
      \\[4mm]
      \inferrev{
        \reg_1 \steq[\bot]{\env}\reg_2
      }
      {
        \forall \vx \in \Regs. \pi_1(\env(\vx)) = \tylow \lor \pi_2(\env(\vx)) = \tylow \Rightarrow \reg_1(\vx)=\reg_2(\vx)
      }
      \quad
      \inferrev{
        \reg_1 \steq[\top]{\env}\reg_2
      }
      {
        \forall \vx \in \Regs.\pi_2(\env(\vx)) = \tylow \Rightarrow \reg_1(\vx)=\reg_2(\vx)
      }
     \\[4mm]
      \inferrev{
        {{\bitem {(\ar,\nat)} {\val{}_1}} \cons\mbuf_1}  \steq[\bot]{\env} {{\bitem {(\ar,\nat)} {\val{}_2}} \cons\mbuf_2}
      }
      {
        \env(\ar) = \tylow \Rightarrow \val{}_1 = \val{}_2
        \quad
        {\mbuf_1}  \steq[\bot]{\env} {\mbuf_2}
      }
      \quad
      \inferrev{
        {{\bitem {(\ar,\nat)} {\val{}_1}} \cons\mbuf_1}  \steq[\top]{\env} {{\bitem {(\ar,\nat)} {\val{}_2}} \cons\mbuf_2}
      }
      {
        {\mbuf_1}  \steq[\bot]{\env} {\mbuf_2}
      }
      \quad
      \inferrev{
        \nil \steq[\bms]{\env}\nil
      }
      {
      }
      \quad
      \inferrev{
        \err \confeq{\env}\err
      }
      {
      }
    \end{rules}
    \begin{rules}[Rules for indexed configurations:]
            \inferrev{
        \reg_1, \mbuf_1, \mem_1, \bms, \idx \confeq{\env} \reg_2, \mbuf_2,\mem_2,\bms, \idx
      }
      {
        \reg_1, \mbuf_1, \mem_1, \bms \confeq{\env} \reg_2, \mbuf_2,\mem_2,\bms
      }
    \end{rules}
  \caption{Rules defining the $\confeq \env$ relation.}
  \label{fig:confeq}
\end{figure*}

\noindent
For proving \Cref{lemma:mainlemma}, it is
convenient to define a notion of proof derivation. With the
metavariable $\pt$ we denote a proof derivation of a typing judgment,
and we use $\deriv$ to denote the relation of being a proof of a
judgment. For instance we now can write:
\[
  \begin{array}[c]{c}
  \infer[\textsc{TSeq}]
  {
    \pt_1 \deriv \sjudg {\env_1} {\cmd_1} {\env_2}
    \quad
    \pt_2
  }
  {\pt \deriv \sjudg {\env_1} {\cmd_1\sep \cmd_2} {\env_3}}
  \end{array}
\]
In this notation, $\pt$ is the proof tree for the entire derivation of
$\sjudg {\env_1} {\cmd_1\sep \cmd_2} {\env_3}$, and $\pt_1, \pt_2$ are
proof trees for the sub-derivations on $\cmd_1$ and $\cmd_2$,
respectively. For brevity's sake, the judgment established by $\pt_2$
, i.e. $\sjudg {\env_2} {\cmd_2} {\env_3}$, is omitted from the
notation. In the following, we will often employ this convention.

Due to the presence of circular jump patterns
in $\jmpcont$, our typing system encompasses a form of fix-point
reasoning over $\jmpcont$ embodied by Rule \ref{TY:TCont}.

\[
  \Infer*[TY][TCont]
  {\sjudg[\emptyset] {\env} {\cmd} \envtwo}
  {\forall \left( \sjudg[] {\env'} {\cmdtwo} {\envtwo'} \right) \in \jmpctx.\
    \sjudg {\env'} {\cmdtwo} {\envtwo'}
    \quad
    \sjudg {\env} {\cmd} \envtwo
  }
\]

With this
rule, in order to type a program $\cmd$, it is possible to populate
$\jmpctx$ with additional assumption on jump targets, prove their
validity and finally use them to type $\cmd$.

We justify this form of reasoning via a Scott-inductive argument.  To
this end, we instrument our semantics with a parameter that represents
the number of allowed jump unfoldings.  More precisely, we first
inductively show that our type system is sound for every initial value
of $\idx \in \Nat$, then we will show if a program is not
constant-time in the normal semantics, it is also not constant-time in
the indexed one.

Specifically, we extend transitions with an index $\idx$ as follows:
\[
  \sstep{\cmd, \st, \bms, \idx }{\cmd', \st', \bms', \idx'}\dir
  \obs.
\]
For brevity's sake we do not report the new semantic rules endowed
with indexes, as all of them, except for those for jumps, are
analogous to the corresponding rule given in \Cref{sec:ts} and leave
the index $i$ unchanged. Rule~\ref{SIE:Jmp} is replaced by the two
following rules:
\begin{gather*}
  \Infer[SIE][JmpNonZero]{
    \sstep{\kwd{jmp}\ \expr, \st, \bms ,\idx+1}
    {\phi(\sem \expr_\st), \st, \bms, \idx}
    {\dstep} {\obranch{\sem \expr_\st}}
  }
  {
  }
  \\
  \Infer[SIE][JmpZero]{
    \sstep{\kwd{jmp}\ \expr, \st, \bms ,0}
    {\err}
    {\dstep} {\obranch{\sem \expr_\st}}
  }
  {
    \sem\expr_\st \in \Loc
  }
\end{gather*}

We call $\States_\idx$ the set of states whose index is $\idx$.
We now turn to our soundness proof. To this end we define and
interpretation of judgments parameterized on the index $\idx$ as
described in \Cref{fig:semjudg}.
\newcommand{\terminal}[1]{#1\downarrow}

\begin{figure*}[t]
  \centering
    \columnwidth=\linewidth
    \begin{align*}
      \sem {\sjudg[] \env \cmd \envtwo}^{\idx}_\nat &\mydef
      \forall \nat'\le \nat.\forall \st, \sttwo \in \States_\idx. \forall \Ds, \Os. \st \confeq\env \sttwo \Rightarrow
    \forall \cmd', \st'.\ \sstep[\nat'] {\cmd, \st} {\cmd', \st'} \Ds \Os \Rightarrow \exists \sttwo'.\ \sstep[\nat'] {\cmd, \sttwo} {\cmd', \sttwo'} \Ds \Os  \land (\cmd'=\cnil \Rightarrow  \st' \confeq\envtwo \sttwo')\\
      \sem {\sjudg[] \env \cmd \envtwo}^{\infty}_\nat &\mydef
                                                   \forall \nat' \le \nat. \forall \st, \sttwo \in \States. \forall \Ds, \Os. \st \confeq\env \sttwo \Rightarrow
                                                   \forall \cmd', \st'. \sstep[\nat'] {\cmd, \st} {\cmd', \st'} \Ds \Os \Rightarrow \exists \sttwo'.\ \sstep[\nat'] {\cmd, \sttwo} {\cmd', \sttwo'} \Ds \Os  \land (\cmd' = \cnil \Rightarrow  \st' \confeq\envtwo \sttwo')\\
      \sem {\sjudg \env \cmd \envtwo}^\idx_\nat &\mydef \sem \jmpctx^\idx_\nat \Rightarrow \sem {\sjudg[] \env \cmd \envtwo}^{\idx+1}_\nat\\[0.5em]
      \sem \jmpctx^\idx_\nat &\mydef  \forall \left( \sjudg[] \env \cmd \envtwo\right) \in \jmpctx. \sem {\sjudg[] \env \cmd \envtwo}^\idx_\nat
    \end{align*}
  \caption{Definition of judgment's semantics }
  \label{fig:semjudg}
\end{figure*}

Thanks to this interpretation, we can now spell out the soundness of
our type system as follows:
\begin{equation}
  \tag{Soundness}
  \label{eq:sound}
  \pt\deriv \sjudg[\emptyset] \env \cmd \envtwo \quad \Rightarrow \quad
  \forall \nat \in \Nat.\ \sem{\sjudg[] \env \cmd \envtwo}^{\infty}_\nat.
\end{equation}
By leveraging that $\sem{\sjudg[] \env \cmd \envtwo}$ is a
Scott-continuous property---i.e., that
$\forall \nat \in \Nat \left(\forall \idx \in \Nat.\ \sem{\sjudg[] \env \cmd \envtwo}^{\idx}_\nat
\right)\Rightarrow \sem{\sjudg[] \env \cmd \envtwo}^{\infty}_\nat$, as shown in
\Cref{lemma:sc}---we can reduce our goal to proving
\begin{equation}
  \tag{C1}
  \label{eq:C1}
  \pt\deriv \sjudg[\emptyset] \env \cmd \envtwo \quad \Rightarrow \quad
  \forall \nat, \idx \in \Nat.\ \sem{\sjudg[] \env \cmd \envtwo}^{\idx}_\nat.
\end{equation}

To support inductive reasoning, our claim will be strengthened as follows:
\[
  \forall \nat \in \Nat. \pt\deriv \sjudg \env \cmd \envtwo \  \Rightarrow \
  \forall \idx \in \Nat.\ \sem{\sjudg[] \env \cmd \envtwo}^0_\nat\land \sem{\sjudg \env \cmd \envtwo}^{\idx}_\nat.
\]
\noindent
Note that when $\jmpctx=\emptyset$, this judgment corresponds to   \eqref{eq:C1}.

\begin{lemma}
  \label{lemma:mainlemma}
  \[
  \forall \nat \in \Nat. \pt\deriv \sjudg \env \cmd \envtwo \  \Rightarrow \
  \forall \idx \in \Nat.\ \sem{\sjudg[] \env \cmd \envtwo}^0_\nat\land \sem{\sjudg \env \cmd \envtwo}^{\idx}_\nat.
\]
\end{lemma}
\begin{proof}
  The proof is by induction on $m = \nat + \size\pt$, where $\size\pt$ is the size of the proof tree, defined by counting the number of applied rules.
  \begin{proofcases}
    \proofcase{0} This case is trivially valid by vacuity of the premise since each proof tree consists of at least one applied rule.
    \proofcase{$m = m'+1$} in this case, we go by cases on the derivation of $\pt$.
    \begin{proofcases}
      \proofcase{\ref{TY:TEmpty}} In this case, we prove a stronger
      claim, namely:
      $\forall \idx, \nat \in \Nat. \sem{\sjudg[] \env \cnil
        \env}^{\idx}_\nat$.  Fix $\nat \in \Nat$. By introspection of the
      definition of $\sem{\sjudg[] \env \cnil \env}^{\idx}_\nat$,
      it must be the case where $\nat'=0$, and therefore the validity
      of the claim can be trivially verified.
      \proofcase{\ref{TY:TWeak}} Let $\nat$ be the length of the
      reduction sequence.  Due to the case analysis, we can assume
      that the proof tree looks as follows:
      \[
        \inferrev[\textsc{TWeak}]
        {
          \pt \deriv \sjudg {\env} {\cmd} {\env'}
        }
        {
          \env \le \envtwo
          \quad
          \pt' \deriv \sjudg {\envtwo} {\cmd} {\envtwo'}
          \quad
          \envtwo' \le \env'
        }
      \]
      By applying the induction hypothesis on $\nat$ and $\pt'$, we prove
      \[
        \forall \idx \in \Nat.\ \sem{\sjudg[] \envtwo \cmd \envtwo'}^0_\nat\land \sem{\sjudg \envtwo \cmd \envtwo'}^{\idx}_\nat.
      \]
      The conclusion follows from \Cref{lemma:monotone}.
      \proofcase{\ref{TY:TCont}} Fix $\nat$.  We assume that the
      derivation has the following shape:
      \[
        \Infer*[TY][TCont]
        {\sjudg[\emptyset] {\env} {\cmd} \envtwo}
        {\forall \left( \sjudg[] {\env'} {\cmdtwo} {\envtwo'} \right) \in \jmpctx.\
          \sjudg {\env'} {\cmdtwo} {\envtwo'}
          \quad
          \pt' \deriv \sjudg {\env} {\cmd} \envtwo
        }
      \]
      We can prove the claim by applying the inductive hypothesis to $\pt'$ and $\nat$.
      However, to do so, we must discharge the assumption
      $\forall \idx \in \Nat. \sem \jmpctx ^\idx_\nat$.
      The proof proceeds by induction on $\idx$.
      \begin{proofcases}
        \proofcase{0} In this case, the conclusion follows trivially
        from the inductive hypotheses on the derivations
        \[
          \left \{ \left(\sjudg {\env'} {\cmdtwo} {\envtwo'} \right)\mid \left( \sjudg[] {\env'} {\cmdtwo} {\envtwo'} \right) \in \jmpctx\right \}.
        \]
        \proofcase{$\idx = \idx' +1$}
        The conclusion follows from the inductive hypothesis on $\idx'$ and from the main inductive hypotheses on the derivations
        \[
          \left \{ \left(\sjudg {\env'} {\cmdtwo} {\envtwo'} \right)\mid \left( \sjudg[] {\env'} {\cmdtwo} {\envtwo'} \right) \in \jmpctx\right \}.
        \]
        From these, we obtain:
        \[
          \forall \idx, \nat. \forall \left( \sjudg[] {\env'} {\cmdtwo} {\envtwo'} \right) \in \jmpctx.\
          \sem{\sjudg {\env'} {\cmdtwo} {\envtwo'}}^\idx_\nat,
        \]
        which is equivalent to:
        \[
          \forall \idx, \nat. \forall \left( \sjudg[] {\env'} {\cmdtwo} {\envtwo'} \right) \in \jmpctx.\
          \sem{\jmpctx}^\idx_\nat \Rightarrow \sem{\sjudg[] {\env'} {\cmdtwo} {\envtwo'}}^{\idx+1}_\nat.
        \]
        By instantiating the above on $\idx'$ and applying the
        inductive hypothesis on $\idx'$ to discharge the premise of
        the implication, we conclude:
        \[
          \forall \nat. \forall \left( \sjudg[] {\env'} {\cmdtwo} {\envtwo'} \right) \in \jmpctx.\
          \sem{\sjudg[] {\env'} {\cmdtwo} {\envtwo'}}^{\idx'+1}_\nat,
        \]
        which corresponds precisely to $\sem{\jmpctx}^\idx_\nat$.
      \end{proofcases}

      \proofcase{\ref{TY:TSeq}} Due to the case analysis, we can assume
    that the proof tree looks as follows:
    \[
      \inferrev[\textsc{TSeq}]
      {\pt \deriv \sjudg {\env_1} {\stat\sep \cmd_2} {\env_3} }
      {
        \pt_1 \deriv \sjudg {\env_1} {\stat} {\env_2}
        \quad
        \pt_2 \deriv \sjudg {\env_2} {\cmd_2} {\env_3}
      }
    \]
    The proof goes by case analysis on $\pt_1$.
    \begin{proofcases}
      \proofcase{\ref{TY:TJmp}} In this case, we can assume that the type derivation has the following shape:
\[
\inferrev[\textsc{TSeq}]
  {\pt \deriv \sjudg{\env_1}{\kwd{jmp}\ \expr_V \sep \cmd}{\env_3}}
  {
    \inferrev[\textsc{TJmp}]
      {\sjudg{\env_1}{\kwd{jmp}\ \expr_V}{\env_2}}
      {\vcenter{\hbox{$
          \begin{array}{@{}c@{}}
            \env_2 = \bigsqcup_{\loc \in V}\envtwo_\loc\\[4pt]
            \sjudg{\env_1}{\expr}{\typair \tylow}\\[4pt]
            \forall \loc \in V.\ \bigl(\sjudg[] {\env_1} {\phi(\loc)} {\envtwo_\loc}\bigr) \in \jmpctx
          \end{array}
        $}}}
    \quad
    \pt_2
  }
\]      \noindent
      Fix $\nat \in \Nat$. Our proof obligation is to establish:
        \begin{multline*}
          \sem {\sjudg[] {\env_1} {\kwd{jmp}\ \expr_V\sep \cmd} {\env_3}}^0_\nat \land
          \forall \idx \in \Nat.\ \sem \jmpctx^\idx_\nat \Rightarrow\\ \sem {\sjudg[] {\env_1} {\kwd{jmp}\ \expr_V\sep \cmd} {\env_3}}^{\idx+1}_\nat.
        \end{multline*}
      We begin by proving the left conjunct. By inspecting the
      semantics, we observe that
      Rule~\ref{SIE:JmpZero} can be applied, and transitions to the target
      configuration $\err$ in a single step. Therefore,
      $\sem {\sjudg[] {\env_1} {\kwd{jmp}\ \expr_V\sep \cmd}
        {\env_3}}^0_\nat$ holds trivially.  We now aim to prove that
      $\forall \idx.\ \sem \jmpctx^\idx_\nat \Rightarrow \sem
      {\sjudg[] {\env_1} {\kwd{jmp}\ \expr_V\sep \cmd}
        {\env_3}}^{\idx+1}_\nat$. Fix $\idx \in \Nat$,
      assume:
      \begin{equation}
        \label{eq:hjmp}
        \tag{H}
        \sem \jmpctx^\idx_\nat,
      \end{equation}
      and we aim to establish:
      \begin{equation}
        \label{eq:claimjmp}
        \tag{C}
        \sem {\sjudg[] {\env_1} {\kwd{jmp}\ \expr_V\sep \cmd} {\env_3}}^{\idx+1}_\nat.
      \end{equation}
      By expanding the definition of $\sem\cdot$, the claim becomes:
      \begin{multline*}
        \forall \nat'\le\nat. \forall \st, \sttwo \in \States_{\idx+1}. \forall \Ds, \Os. \st \confeq{\env_1} \sttwo \Rightarrow\\
        \forall \cmd', \st'.\ \sstep[\nat'] {\kwd{jmp}\ \expr_V\sep \cmd, \st} {\cmd', \st'} \Ds \Os \Rightarrow\\
        \exists \sttwo'.\ \sstep[\nat'] {\kwd{jmp}\ \expr_V\sep \cmd, \sttwo} {\cmd', \sttwo'} \Ds \Os  \land (\cmd'=\cnil \Rightarrow  \st' \confeq{\env_3} \sttwo').
      \end{multline*}
      Fix $\nat'$, $\st, \sttwo \in \States_{\idx+1}$, a sequence of directives $\Ds$, and a sequence of observations $\Os$. Fix $\cmd', \st'$ such that:
      \begin{gather*}
        \tag{H1}
        \st \confeq{\env_1} \sttwo\\
        \tag{H2}
        \sstep[\nat'] {\kwd{jmp}\ \expr_V\sep \cmd, \st} {\cmd', \st'} \Ds \Os
      \end{gather*}
      Our goal is to prove that there exists $\sttwo'$ such that:
      \begin{gather*}
        \tag{C1}
        \sstep[\nat'] {\kwd{jmp}\ \expr_V\sep \cmd, \sttwo} {\cmd', \sttwo'} \Ds \Os\\
        \tag{C2}
        (\cmd'=\cnil \Rightarrow  \st' \confeq{\env_3} \sttwo')
      \end{gather*}
      If $\nat'=0$, then the conclusion is trivial. Hence, we assume
      $\nat' > 0$ without loss of generality. Under this assumption,
      from (H1) and \Cref{lemma:idexpr}, we deduce that either
      $\sem\expr_\st = \sem\expr_\sttwo = \bot$ or
      $\sem\expr_\st = \sem\expr_\sttwo = \loc$ for some $\loc$.
      Notice that the first case contradicts $\nat'=0$, so we must be in the second one.
      In this case, (H2) can be
      rewritten as:
      \[
        \sstep {\kwd{jmp}\ \expr_V\sep \cmd, \st} {\jmpcont(\loc)\sep \cmd, \st[\idx+1\upd\idx]} \dir \obs
        \xrightarrow[\Ds']{\Os'}^{\nat'-1}{\cmd', \st'},
      \]
      and we observe that:
      \[
        \sstep {\kwd{jmp}\ \expr_V\sep \cmd, \sttwo} {\jmpcont(\loc)\sep \cmd, \sttwo[\idx+1\upd\idx]} \dir \obs.
      \]
      In order  to prove (C1), our goal becomes:
      \[
        \sstep[\nat'-1] {\jmpcont(\loc)\sep \cmd, \sttwo[\idx+1\upd\idx]} {\cmd', \sttwo'}
        {\Ds'}{\Os'}
      \]
      From (H1), we deduce that $\st[\idx+1\upd\idx] \confeq{\env_1} \sttwo[\idx+1\upd\idx]$.
      The proof proceeds by case analysis. If
      \[
        \sstep[\nat'-1] {\jmpcont(\loc), \st[\idx+1\upd\idx]} {\cmd'', \st''}
        {\Ds'}{\Os'}
      \]
      for some ${\cmd'', \st''}$, then (C1) is a direct consequence of
      (H) and \Cref{lemma:sepaux}. Claim (C2) is a direct consequence of (H) and
      \Cref{lemma:monotone}. Otherwise, introspection of the semantics
      shows that there must exist $\nat'' < \nat$, $\Os''$, and
      $\Ds''$ such that:
      \[
        \sstep[\nat'-1] {\jmpcont(\loc), \st[\idx+1\upd\idx]} {\err}
        {\Ds''}{\Os''},
      \]
      \[
        \sstep[\nat''] {\jmpcont(\loc), \st[\idx+1\upd\idx]} {\err}
        {\Ds''}{\Os''}
      \]
      or
      \[
        \sstep[\nat''] {\jmpcont(\loc), \st[\idx+1\upd\idx]} {\cmd'', \st''}
        {\Ds''}{\Os''} \downarrow
      \]
      where with $\terminal{\cdot}$, we indicate that the
      configuration is terminal.  We proceed by further case
      analysis. The first two cases are absurd, since we assumed that
      ${\kwd{jmp}\ \expr_V\sep \cmd, \st}$ reduces for at least
      $\nat'$ steps. Thus, we must be in the third case.  Here, if
      $\cmd'' \neq \cnil$, we again reach a contradiction using the
      same argument. We are then left with the case where
      $\cmd'' = \cnil$. In this case, the conclusion follows from the
      inductive hypothesis on $\pt_2, \nat$, and by applying
      \Cref{lemma:monotone} and \Cref{lemma:monotone2}.

      \proofcase{\ref{TY:TFence}} In this
      case, we can assume that the type derivation has the following
      shape:
      \[
        \inferrev[\textsc{TSeq}]
        {\pt \deriv \sjudg {\env_1} {\fence\sep \cmd} {\env_3}}
        {
          \inferrev[\textsc{TFence}]{\sjudg {\env_1} {\fence} {\env_2} }
          {\forall \vx \in \Regs. \env_2(\vx) = \typair{\pi_1(\env_1(\vx))} }
          \quad
          \pt_2
        }
      \]
      With $\pt_2  \deriv \sjudg {\env_2} {\cmd} {\env_3} $.
      We fix $\nat \in \Nat$, and our goal is to prove:
      \begin{multline*}
        \sem {\sjudg[] {\env_1} {\fence\sep \cmd} {\env_3}}^{0}_\nat \land
        \forall \idx \in \Nat.\ \sem \jmpctx^{\idx}_\nat \Rightarrow\\ \sem {\sjudg[] {\env_1} {\fence\sep \cmd} {\env_3}}^{\idx+1}_\nat.
      \end{multline*}
      and the inductive hypothesis on $\pt_2, \nat$ yields
      \[
        \sem {\sjudg[] {\env_1} {\cmd} {\env_3}}^{0}_\nat \land
        \forall \idx \in \Nat.\ \sem \jmpctx^{\idx}_\nat \Rightarrow \sem {\sjudg[] {\env_1} {\cmd} {\env_3}}^{\idx+1}_\nat.
      \]
      Therefore, via \Cref{lemma:ihaux}, we can reduce our goal to showing:
      \[
        \forall \idx \in \Nat.\ \sem {\sjudg[] {\env_2} {\cmd} {\env_3}}^{\idx}_\nat\ \  \Rightarrow\ \
        \sem {\sjudg[] {\env_1} {\fence\sep \cmd} {\env_3}}^{\idx}_\nat.
      \]
      Fix $\idx$ and assume
      \[
        \tag{H}
        \sem {\sjudg[] {\env_1} {\cmd} {\env_3}}^{\idx}_\nat.
      \]
      By expanding the definition of $\sem\cdot$, the claim becomes:
      \begin{multline*}
        \forall \nat'\le\nat. \forall \st, \sttwo \in \States_{\idx}. \forall \Ds, \Os. \st \confeq{\env_1} \sttwo \Rightarrow\\
        \exists \cmd', \st'.\ \sstep[\nat'] {\fence\sep \cmd, \st} {\cmd', \st'} \Ds \Os \Rightarrow\\
        \forall \sttwo'.\ \sstep[\nat'] {\fence\sep \cmd, \sttwo} {\cmd', \sttwo'} \Ds \Os  \land (\cmd'=\cnil \Rightarrow  \st' \confeq{\env_3} \sttwo').
      \end{multline*}
      Fix $\nat'$,
      \(
        \st = \reg_1, \mbuf_1, \mem_1, \bms, \idx\) and \(\sttwo= \reg_2, \mbuf_2, \mem_2, \bms, \idx,
      \)
      a sequence of directives $\Ds$, and a sequence of observations
      $\Os$. Fix $\cmd', \st'$ such that:
      \[
        \sstep[\nat'] {\fence\sep \cmd, \st} {\cmd', \st'} \Ds \Os.
      \]
      The proof goes by cases on $\bms$.
      \begin{proofcases}
        \proofcase{$\top$} In this case, no transition is possible, so we have $\nat'=0$, $\sttwo' = \sttwo$ and $\cmd' = \fence\sep \cmd$, therefore the conclusion is trivial.
        \proofcase{$\bot$} In this case, if $\nat'=0$, the conclusion is trivial (see previous case), otherwise, we have
        \[
          \sstep {\fence\sep \cmd, \st} {\cmd, \reg_1, \nil, \overline{\bm {\mbuf_1} {\mem_1}}, \bot, \idx} \dstep \onone
        \]
        and
        \[
          \sstep {\fence\sep \cmd, \sttwo} {\cmd, \reg_2, \nil, \overline{\bm {\mbuf_2} {\mem_2}}, \bot, \idx} \dstep \onone
        \]
        In order to apply (H) and conclude the proof, we need to establish
        ${\reg_1, \nil, \overline{\bm {\mbuf_1} {\mem_1}}, \bot, \idx}\confeq{\env_2} {\reg_2, \nil, \overline{\bm {\mbuf_2} {\mem_2}}, \bot, \idx}$, which requires some additional reasoning. In particular, observe that
      \begin{gather*}
        \forall \vx \in \Regs.\pi_1({\env_1} (\vx)) = \tylow \Rightarrow {\env_2}(\vx) = \typair \tylow\\
        \forall \vx \in \Regs.\pi_1({\env_1} (\vx)) = \tyhigh \Rightarrow {\env_2} (\vx) = \typair \tyhigh\\
        \forall \ar \in \Arr.\env_1(\ar) = \tylow \Rightarrow {\env_2}(\ar) = \tylow\\
        \forall \ar \in \Arr.\env_1(\ar) = \tyhigh \Rightarrow \env_2 (\ar) = \tyhigh.
      \end{gather*}
      With this observation and \Cref{lemma:auxfence}, we deduce that:
      \[
        \reg_1, \mbuf_1, \mem_1, \bot
        \confeq{{\env_2}}
        \reg_2, \mbuf_2, \mem_2, \bot
      \]
      From this intermediate claim and \Cref{lemma:fenceflush}, we conclude
      \[
        {\reg_1, \nil, \overline{\bm {\mbuf_1} {\mem_1}}, \bot, \idx}\confeq{\env_2} {\reg_2, \nil, \overline{\bm {\mbuf_2} {\mem_2}}, \bot, \idx}.
      \]
      \end{proofcases}

      \proofcase{\ref{TY:TDfence}} In this case, we can assume that the type derivation has the following shape:
      \small
      \[
        \inferrev[\textsc{TSeq}]
        {\pt \deriv \sjudg {\env_1} {\dfencex\sep \cmd} {\env_2}}
        {
          \inferrev[\textsc{TDfence}]{\sjudg {\env_1} {\dfencex} {\supdate {\env_1} \vx {\typair{\pi_1(\env_1(\vx))}}} }
          {}\ \
          \pt'
        }
      \]
      \normalsize
      With
      $\pt'\deriv \sjudg {\supdate {\env_1} \vx
        {\typair{\pi_1(\env_1(\vx))}} } {\cmd} {\env_2}$.
      We fix $\nat \in \Nat$, and our goal is to prove:
      \begin{multline*}
        \sem {\sjudg[] {\env_1} {\dfencex\sep \cmd} {\env_3}}^{0}_\nat \land
        \forall \idx \in \Nat.\ \sem \jmpctx^{\idx}_\nat \Rightarrow \\
        \sem {\sjudg[] {\env_1} {\dfencex \sep \cmd} {\env_3}}^{\idx+1}_\nat.
      \end{multline*}
      By reasoning similar to the previous case, with the help of \Cref{lemma:ihaux}, we can reduce our goal to proving
      \begin{multline*}
        \hspace{-1em}\forall \idx \in \Nat. \sem {\sjudg[] {{\supdate {\env_1} \vx {\typair{\pi_1(\env_1(\vx))}}}} {\cmd} {\env_3}}^{\idx}_\nat\Rightarrow\\
        \sem {\sjudg[] {\env_1} {\dfencex\sep \cmd} {\env_3}}^{\idx}_\nat,
      \end{multline*}
      Fix $\idx$. Our goal is to establish:
      \[
        \tag{C}
        \sem {\sjudg[] {\env_1} {\dfencex\sep \cmd} {\env_3}}^{\idx}_\nat
      \]
      under the assumption
      \[
        \tag{H}
        \sem {\sjudg[] {{\supdate {\env_1} \vx {\typair{\pi_1(\env_1(\vx))}}}} {\cmd} {\env_3}}^{\idx}_\nat
      \]
      By expanding the definition of $\sem\cdot$,
      we rewrite the claim as follows:
      \begin{multline*}
        \forall \nat'\le\nat. \forall \st, \sttwo \in \States_{\idx}. \forall \Ds, \Os. \st \confeq{\env_1} \sttwo \Rightarrow\\
        \forall \cmd', \st'.\ \sstep[\nat'] {\dfencex\sep \cmd, \st} {\cmd', \st'} \Ds \Os \Rightarrow\\
        \exists \sttwo'.\ \sstep[\nat'] {\dfencex\sep \cmd, \sttwo} {\cmd', \sttwo'} \Ds \Os  \land (\cmd'=\cnil \Rightarrow  \st' \confeq{\env_3} \sttwo').
      \end{multline*}
      Fix $\nat'$, $\st = \reg_1, \mbuf_1, \mem_1, \bms, \idx$ and $\sttwo= \reg_2, \mbuf_2, \mem_2, \bms, \idx$,
      a sequence of directives $\Ds$, and a sequence of observations
      $\Os$. Fix $\cmd', \st'$ such that:
      \[
        \sstep[\nat'] {\dfencex\sep \cmd, \st} {\cmd', \st'} \Ds \Os.
      \]
      Then, observe
      \[
        \sstep {\dfencex\sep \cmd, \st} {\cmd, \st[\vx \upd z]} \dstep \onone
      \]
      and
      \[
        \sstep {\dfencex\sep \cmd, \sttwo} {\cmd, \sttwo[\vx \upd z]} \dstep \onone,
      \]
      where $z = \cif {\bms =\top }{\bot} \vx$.  In order to establish
      the claim, we can apply the inductive hypothesis on
      $\pt', \nat$, and instantiate $\nat'$ in (H) with $\nat'-1$. To conclude, we then just need to
      show:
      $\st[\vx \upd z] \confeq{\supdate {\env_1} \vx
        {\typair{\pi_1(\env_1(\vx))}}} \sttwo[\vx \upd z]$, which can
      be showed by expanding the definition of $\confeq{}$ and by
      applying \Cref{lemma:sfenceaux}.

      \proofcase{\ref{TY:TAsgn}} In this case, we can assume that the type derivation has the following shape:
      \[
        \inferrev[\textsc{TSeq}]
        {\pt \deriv \sjudg {\env_1} {\vx \asgn \expr\sep \cmd} {\env_2} }
        {
          \inferrev[\textsc{TAsgn}]{\sjudg {\env_1} {\vx \asgn \expr} {\supdate {\env_1} \vx \tty}}
          {\sjudg {\env_1} {\expr} {\tty} }
          \quad
          \pt_2
        }
      \]
      With
      $\pt_ 2 \deriv \sjudg {\supdate {\env_1} \vx \tty} {\cmd}
      {\env_2} $.
            We fix $\nat \in \Nat$, and our goal is to prove:
            \begin{multline*}
        \sem {\sjudg[] {\env_1} { \vx\asgn \expr\sep \cmd} {\env_3}}^{0}_\nat \land
        \forall \idx \in \Nat.\ \sem \jmpctx^{\idx}_\nat \Rightarrow\\ \sem {\sjudg[] {\env_1} {\vx \asgn \expr \sep \cmd} {\env_3}}^{\idx+1}_\nat.
      \end{multline*}
      As in the previous cases, by applying \Cref{lemma:ihaux} to the inductive hypothesis on
      $\pt_2, \nat$, we can reduce our claim to proving
      \[
        \tag{C}\sem {\sjudg[] {\env_1} {\vx\asgn \expr\sep \cmd} {\env_3}}^{\idx}_\nat
      \]
      under the assumption
      \[
        \tag{H}\sem {\sjudg[] {\supdate {\env_1} \vx \tty} {\cmd} {\env_3}}^{\idx}_\nat,
      \]
      for a fixed $\idx$. By expanding the definition of $\sem\cdot$,
      we rewrite the claim as follows:
      \begin{multline*}
        \forall \nat'\le\nat. \forall \st, \sttwo \in \States_{\idx}. \forall \Ds, \Os. \st \confeq{\env_1} \sttwo \Rightarrow\\
        \forall \cmd', \st'.\ \sstep[\nat'] {\vx\asgn \expr\sep \cmd, \st} {\cmd', \st'} \Ds \Os \Rightarrow\\
        \exists \sttwo'.\ \sstep[\nat'] {\vx\asgn \expr\sep \cmd, \sttwo} {\cmd', \sttwo'} \Ds \Os  \land (\cmd'=\cnil \Rightarrow  \st' \confeq{\env_3} \sttwo').
      \end{multline*}
      If $\nat'=0$, then the conclusion is trivial. Therefore, in the
      following we can assume $\nat' > 0$ without loss of
      generality. Under this assumption, from
      $\st \confeq{\env_1} \sttwo$, we deduce $\reg_1 \boteq \reg_2$, therefore
      $\sem\expr_\st = \bot \Leftrightarrow \sem\expr_\sttwo = \bot$.
      The proof proceeds by case analysis. If $\sem\expr_\st = \bot$ and
      $\sem\expr_\sttwo = \bot$, then the reduction sequence must be empty, i.e., $\nat'=0$,
      which contradicts our assumption on $\nat'$, therefore, we must be in the case where
      $\sem\expr_\st \neq \bot$ and $\sem\expr_\sttwo \neq \bot$.
      Observe
      \[
        \sstep {\vx \asgn \expr\sep \cmd, \st} {\cmd, \st[\vx \upd \sem \expr_\st]} \dstep \onone
      \]
      and
      \[
        \sstep {\vx \asgn \expr\sep \cmd, \sttwo} {\cmd, \sttwo[\vx \upd \sem \expr_\sttwo]} \dstep \onone.
      \]
      In order to conclude the proof of this sub-derivation, we can
      employ (H). To this end, we need to conclude
      $\confone_1'\confeq{{\supdate {\env_1} \vx \tty}} \confone_2'$.
      In turn, this conclusion follows from \Cref{lemma:asgnaux} and
      \[
        \supdate {\reg_1}\vx {\sem \expr_{\reg_1}} \boteq \supdate {\reg_2}\vx {\sem \expr_{\reg_2}}.
      \]

      \proofcase{\ref{TY:TIf}} In this case, we can assume that the type derivation is the one below:
      \[
        \inferrev[\textsc{TSeq}]
        {\pt \deriv \sjudg {\env_1} {{\cif \expr {\cmd_1} {\cmd_2}}\sep \cmd} {\env_3}}
        {
          \inferrev[\textsc{TIf}]
          {
            \sjudg {\env_1} {\cif \expr {\cmd_1} {\cmd_2}} {\env_2}
          }
          {
            \sjudg[] {\env_1} {\expr} {\tty}
            \quad
            \sjudg {\env_1} {\cmd_1} {\env_2}
            \quad
            \sjudg {\env_1} {\cmd_2} {\env_2}
          }
          \quad
          \pt_2
        }
      \]
      With $\pt_2 \deriv \sjudg {\env_2} {\cmd} {\env_3}$.
      Fix $\nat \in \Nat$. Our proof obligation is to establish:
      \begin{gather*}
        \tag{C1}
        \sem {\sjudg[] {\env_1} {{\cif \expr {\cmd_1} {\cmd_2}}\sep \cmd} {\env_3}}^0_\nat\\
        \tag{C2}
        \begin{multlined}
        \forall \idx \in \Nat.\ \sem \jmpctx^\idx_\nat \Rightarrow\\ \sem {\sjudg[] {\env_1} {{\cif \expr {\cmd_1} {\cmd_2}}\sep \cmd} {\env_3}}^{\idx+1}_\nat.
      \end{multlined}
    \end{gather*}
      For brevity's sake we will only show the proof of (C2), the one of
      (C1) is structurally identical, but less complex
      as it does not require reasoning about $\sem \jmpctx^\idx_\nat$.
      Fix $\idx \in \Nat$,
      assume:
      \begin{equation}
        \tag{H}
        \sem \jmpctx^\idx_\nat,
      \end{equation}
      and we aim to establish:
      \begin{equation}
        \tag{C}
        \sem {\sjudg[] {\env_1} {{\cif \expr {\cmd_1} {\cmd_2}}\sep \cmd} {\env_3}}^{\idx+1}_\nat.
      \end{equation}
      By expanding the definition of $\sem\cdot$, the claim becomes:
      \small
      \begin{multline*}
        \forall \nat'\le\nat. \forall \st, \sttwo \in \States_{\idx+1}. \forall \Ds, \Os. \st \confeq{\env_1} \sttwo \Rightarrow\\
        \forall \cmd', \st'.\ \sstep[\nat'] {{\cif \expr {\cmd_1} {\cmd_2}}\sep \cmd, \st} {\cmd', \st'} \Ds \Os \Rightarrow\\
        \exists \sttwo'.\ \sstep[\nat'] {{\cif \expr {\cmd_1} {\cmd_2}}\sep \cmd, \sttwo} {\cmd', \sttwo'} \Ds \Os  \land\\
        (\cmd'=\cnil \Rightarrow  \st' \confeq{\env_3} \sttwo').
      \end{multline*}
      \normalsize
      Fix $\nat'$, $\st, \sttwo \in \States_{\idx+1}$, a sequence of directives $\Ds$, and a sequence of observations $\Os$. Fix $\cmd', \st'$ such that:
      \begin{gather*}
        \tag{H1}
        \st \confeq{\env_1} \sttwo\\
        \tag{H2}
        \sstep[\nat'] {{\cif \expr {\cmd_1} {\cmd_2}}\sep \cmd, \st} {\cmd', \st'} \Ds \Os
      \end{gather*}
      Our goal is to prove that there exists $\sttwo'$ such that:
      \begin{gather*}
        \tag{CA}
        \sstep[\nat'] {{\cif \expr {\cmd_1} {\cmd_2}}\sep \cmd, \sttwo} {\cmd', \sttwo'} \Ds \Os\\
        \tag{CB}
        \cmd'=\cnil \Rightarrow  \st' \confeq{\env_3} \sttwo'
      \end{gather*}
      If $\nat'=0$, then the conclusion is trivial. Hence, we assume
      $\nat' > 0$ without loss of generality. Under this assumption,
      from (H1) and \Cref{lemma:idexpr}, we deduce that either
      $\sem\expr_\st = \sem\expr_\sttwo = \bot$ or
      $\sem\expr_\st = \sem\expr_\sttwo = \val{}\in \Vals$.  As for
      the case of indirect jumps, notice that the first case
      contradicts the assumption $\nat'=0$, so we must be in the
      second one.  In this case, (H2) can be rewritten as:
      \small
        \begin{multline*}
        \sstep {{\cif \expr {\cmd_1} {\cmd_2}}\sep \cmd, \st}
        {} {{\dbranch b}} {{\obranch {\sem \expr}_\st}}\\
        {\cmd_i\sep \cmd, \st} \xrightarrow[\Ds']{\Os'}^{\nat'-1}{\cmd', \st'},
      \end{multline*}
      \normalsize
      with $i \in \{1, 2\}$.
      We also observe that:
      \[
        \sstep {{\cif \expr {\cmd_1} {\cmd_2}}\sep \cmd, \sttwo} {\cmd_i\sep \cmd, \sttwo} {{\dbranch b}} {{\obranch {\sem \expr}_\sttwo}},
      \]
      via Rule~\ref{SIE:SIf}. In particular, \Cref{lemma:idexpr} guarantees that
      ${\sem \expr}_\st = {\sem \expr}_\sttwo$.
      In order  to prove (CA), the goal becomes:
      \[
        \sstep[\nat'-1] {\cmd_i\sep \cmd, \sttwo} {\cmd', \sttwo'}
        {\Ds'}{\Os'}.
      \]
      Note that by combining (H) with the inductive hypothesis on the
      derivation $\sjudg {\env_1} {\cmd_i} {\env_2}$ and $\nat$
      we conclude:
      \[
        \tag{H'}
        \sem {\sjudg[] {\env_1} {\cmd_i} {\env_2}}^{\idx+1}_\nat.
      \]

      The proof proceeds by case analysis. If
      \[
        \sstep[\nat'-1] {\cmd_i, \st} {\cmd'', \st''}
        {\Ds'}{\Os'}
      \]
      for some ${\cmd'', \st''}$, then (CA) and (CB) are direct consequences of
      (H') and \Cref{lemma:sepaux}. Otherwise, introspection of the semantics
      shows that there must exist $\nat'' < \nat$, $\Os''$, and
      $\Ds''$ such that:
      \[
        \sstep[\nat'-1] {\cmd_i, \st} {\err}
        {\Ds''}{\Os''},
      \]
      \[
        \sstep[\nat''] {\cmd_i, \st} {\err}
        {\Ds''}{\Os''}
      \]
      or
      \[
        \sstep[\nat''] {\cmd_i, \st} {\cmd'', \st''}
        {\Ds''}{\Os''} \downarrow.
      \]
      The proof proceeds by case analysis. The first two cases lead to
      a contradiction, since we assumed that the configuration
      ${\cif \expr {\cmd_1} {\cmd_2}}\sep \cmd, \st$ reduces for at
      least $\nat'$ steps. Thus, we must be in the last case.  With
      another case analysis we observe that, if $\cmd'' \neq \cnil$,
      we again reach a contradiction using the same argument. We are
      then left with the case where $\cmd'' = \cnil$. In this case,
      the conclusion follows from the inductive hypothesis on
      $\pt_2, \nat$, (H)
      and \Cref{lemma:monotone2}.

      \proofcase{\ref{TY:TWhile}} In this case, we can assume that the type derivation is the following one:
      \[
        \inferrev[\textsc{TSeq}]
        {\pt \deriv \sjudg {\env} {{\cwhile \expr {\cmd}}\sep \cmdtwo} {\envtwo}}
        {
          \pt_1' \deriv \inferrev[\textsc{TWhile}]
          {
            \sjudg {\env} {\cwhile \expr {\cmd}} {\env}
          }
          {
            \pt_\expr\deriv \sjudg[] {\env} {\expr} {\tty}
            \quad
            \pt_1\deriv \sjudg {\env} {\cmd} {\env}
}
          \quad
          \pt_2
        }
      \]
With $          \pt_2 \deriv \sjudg {\env} {\cmdtwo} {\envtwo}$.
      We fix $\nat$, and we state our claim:
Our proof obligation is to establish:
      \begin{gather*}
        \tag{C1}
        \sem {\sjudg[] {\env} {{\cwhile \expr {\cmd} }\sep \cmdtwo} {\envtwo}}^0_\nat\\
        \tag{C2}
        \forall \idx \in \Nat.\ \sem \jmpctx^\idx_\nat \Rightarrow \sem {\sjudg[] {\env} {{\cwhile \expr {\cmd}}\sep \cmdtwo} {\envtwo}}^{\idx+1}_\nat.
      \end{gather*}
      For brevity's sake we will only show the proof of (C2), the one of
      (C1) is structurally identical, but less complex
      since it does not require reasoning about $\sem \jmpctx^\idx_\nat$.
      Fix $\idx \in \Nat$,
      assume:
      \begin{equation}
        \tag{H}
        \sem \jmpctx^\idx_\nat,
      \end{equation}
      and we aim to establish:
      \begin{equation}
        \tag{C}
        \sem {\sjudg[] {\env} {\cwhile \expr {\cmd} \sep \cmdtwo} {\envtwo}}^{\idx+1}_\nat.
      \end{equation}
      By expanding the definition of $\sem\cdot$, the claim becomes:
      \small
      \begin{multline*}
        \forall \nat'\le\nat. \forall \st, \sttwo \in \States_{\idx+1}. \forall \Ds, \Os. \st \confeq{\env} \sttwo \Rightarrow\\
        \forall \cmd', \st'.\ \sstep[\nat'] {\cwhile \expr {\cmd} \sep \cmdtwo, \st} {\cmd', \st'} \Ds \Os \Rightarrow\\
        \exists \sttwo'.\ \sstep[\nat'] {\cwhile \expr {\cmd} \sep \cmdtwo, \sttwo} {\cmd', \sttwo'} \Ds \Os  \land\\ (\cmd'=\cnil \Rightarrow  \st' \confeq{\envtwo} \sttwo').
      \end{multline*}
      \normalsize We fix $\nat'$, directives $\Ds$ and $\Os$, states
      $\st \confeq{\env} \sttwo$. We also fix $\cmd'$ and $\st'$ such that:
      \[
        \sstep[\nat'] {\cwhile \expr {\cmd} \sep \cmdtwo, \st} {\cmd', \st'} \Ds \Os
      \]
      Our proof obligation now imposes us to show that there exists $\sttwo'$ such that:
      \begin{gather*}
        \tag{CA}
        \sstep[\nat'] {\cwhile \expr {\cmd} \sep \cmdtwo, \sttwo} {\cmd', \sttwo'} \Ds \Os  \\
        \tag{CB}
        \cmd'=\cnil \Rightarrow  \st' \confeq{\envtwo} \sttwo'.
      \end{gather*}
      Notice that if $\nat'=0$, then the
      conclusion is trivial. Therefore, in the following, we assume
      that $0< \nat'\le \nat$ without loss of generality. This means
      that $\Ds=\dir:\Ds'$ and $\Os=\obs:\Os'$, and that we can assume
      \[
        \sstep {\cwhile \expr {\cmd} \sep \cmdtwo, \st} {\cmd'', \st} \dir \obs \xrightarrow[\Ds'] {\Os'} {\cmd', \st'}.
      \]
      By introspection of the semantics we conclude
      $\dir = \dbranch \bool$ and that
      $\obs = \obranch{{\sem \expr}_\st}$.  Note that from
      $\sjudg[] {\env} {\expr} {\tty}$ and \Cref{lemma:idexpr}, we
      deduce $\sem \expr_\st = \sem \expr_\sttwo$, therefore we also have:
      \[
        \sstep {\cwhile \expr {\cmd} \sep \cmdtwo, \sttwo} {\cmd'', \sttwo} \dir \obs .
      \]
      To conclude (CA) we need to establish:
      \[
        {\cmd'', \sttwo} \xrightarrow[\Ds'] {\Os'} {\cmd', \sttwo'}
      \]
      The proof of proceeds by cases on $\bool$.
      \begin{proofcases}
        \proofcase{$\bool= \bot$} In this case, we have
        $\cmd'' = \cmdtwo$, therefore the conclusions (CA) and (CB) are direct
        consequences of the inductive hypothesis on $\pt_2$, $\nat$ and (H).
        \proofcase{$\bool= \top$} In this case, we have
        \[
          \cmd'' = \cmd \sep \cwhile \expr \cmd \sep \cmdtwo.
        \]
        The proof proceeds by case analysis. If
        \newcommand{\cmdthree}{\mathsf{H}}
        \newcommand{\stthree}{u}
        \[
          \sstep[\nat'-1] {\cmd, \st} {\cmdthree, \stthree}
          {\Ds'}{\Os'}
        \]
        for some ${\cmdthree, \stthree}$,
        then (CA) is direct consequence of the induction hypothesis
        on $\pt_1$, $\nat$ and \Cref{lemma:sepaux}. Note that (CB) is
        trivially satisfied because the configuration reached after $\nat'-1$
        reductions from $\cmd''$ cannot be $\cnil$.  Otherwise, introspection of the
        semantics shows that there must exist $\nat'' < \nat'-1$,
        $\Os''$, and $\Ds''$ such that:
        \[
          \sstep[\nat'-1] {\cmd, \st} {\err}
          {\Ds''}{\Os''},
        \]
        \[
          \sstep[\nat''] {\cmd, \st} {\err}
          {\Ds''}{\Os''},
        \]
        or
        \[
          \sstep[\nat''] {\cmd, \st} {\cmdthree, \stthree}
          {\Ds''}{\Os''} \downarrow.
        \]
        The proof proceeds by case analysis. The first two cases lead
        to a contradiction, since we assumed that the configuration
        ${\cwhile \expr {\cmd} \sep \cmdtwo, \sttwo}$ reduces for at
        least $\nat'$ steps and reaches a configuration different from
        $\err$. Thus, we must be in the third case.  In this case, if
        $\cmd'' \neq \cnil$, we again reach a contradiction using the
        same argument. We are then left with the case where
        $\cmd'' = \cnil$. In this case, the conclusion follows from
        the inductive hypothesis on $\pt, \nat-\nat''$, (H), and \Cref{lemma:monotone2}.
      \end{proofcases}

      \proofcase{\ref{TY:TLoad}}
      In this case, we can assume that the type derivation has the following shape:
      \[
      \hspace{-0.5em}
          \inferrev[\textsc{TSeq}]
          {\pt \deriv \sjudg {\env_1} {\cmemread \vx  {\ar[\expr]}\sep \cmd} {\env_2} }
          {
            \inferrev[\textsc{TLoad}]
            {
              \sjudg {\env_1} {\cmemread \vx  {\ar[\expr]}} {\supdate {\env_1} \vx {\tty}}
            }
            {
              \sjudg {\env_1} {\expr} {\typair \tylow}
              \quad
              \tty = (\env_1(a), \tyhigh)
            }
            \quad
            \pt_2
          }
        \]
        With $\pt_2 \deriv \sjudg {\supdate {\env_1} \vx {\tty}} {\cmd} {\env_2}$.
        We fix $\nat \in\Nat$, and by applying \Cref{lemma:ihaux},
        we can reduce our goal to proving:
        \[
          \tag{C}
          \sem {\sjudg[] {{\env_1}} {\cmemread \vx  {\ar[\expr]}} {\env_2}}^\idx_{\nat}
        \]
        under the assumption
        \[
          \tag{H}
          \sem {\sjudg[] {\supdate {\env_1} \vx {\tty}} {\cmd} {\env_2}}^\idx_{\nat}
        \]
        for a fixed $\idx \in \Nat$. By expanding the definition of $\sem \cdot$,
        we can rephrase the claim as follows:
        \begin{multline*}
          \forall \nat'\le\nat. \forall \st, \sttwo \in \States_{\idx}. \forall \Ds, \Os. \st \confeq{\env_1} \sttwo \Rightarrow\\
          \forall \cmd', \st'.\ \sstep[\nat'] {\cmemread \vx {\ar[\expr]} \sep \cmd, \st} {\cmd', \st'} \Ds \Os \Rightarrow\\
          \exists \sttwo'.\ \sstep[\nat'] {\cmemread \vx {\ar[\expr]} \sep \cmd, \sttwo} {\cmd', \sttwo'} \Ds \Os  \land (\cmd'=\cnil \Rightarrow  \st' \confeq{\env_2} \sttwo').
        \end{multline*}
        Fix $\nat' \in \Nat$,
        $\st = \reg_1, \mbuf_1, \mem_1, \bms, \idx$ and
        $\sttwo= \reg_2, \mbuf_2, \mem_2, \bms, \idx$ such that
        $\st \confeq{\env} \sttwo$, directives $\Ds$ and observations
        $\Os$. Note that when $\nat'=0$, the claim is
        trivial. Therefore, in the following, we assume without loss
        of generality $0< \nat'\le \nat$. By leveraging this
        observation and by introspection of the semantics, this also
        means that $\Os = \omem {(a, {\sem \expr}_\st)}: \Os'$ and
        $\Ds = \dir: \Ds'$ for some $\dir$, $\Ds'$, $\Os'$.  Under
        this assumption, our proof obligation is showing
        \begin{gather*}
          \tag{CA}
          \sstep[\nat'] {\cmemread \vx {\ar[\expr]} \sep \cmd, \sttwo}
          {\cmd', \sttwo'} \Ds \Os \\
          \tag{CB}
          \cmd'=\cnil \Rightarrow \st'
          \confeq{\env_2} \sttwo'
        \end{gather*}
        where $\cmd', \st'$ are such that
        \[
          \sstep {\cmemread \vx {\ar[\expr]} \sep \cmd, \st} {\cmd'',
            \st''} \dir {\omem {(a, {\sem \expr}_\st)}} \xrightarrow[\Ds']{\Os'}\!\!{}^{\nat'-1}\, \cmd', \st'.
        \]
        Also note that from $\sjudg {\env_1} {\expr} {\typair \tylow}$,
        $\st \confeq{\env} \sttwo$ and \Cref{lemma:tyexpr}, we deduce
        that $\sem \expr_{\sttwo} =\sem \expr_{\st} = h$. The proof
        proceeds by case analysis on the rule that was used to show
        the transition.

\begin{proofcases}
          \proofcase{\ref{SIE:SLoad-PSF}} By introspection of the
          rule, we deduce that $\dir = \dload j$,
          $\sem \expr_{\st}\in \Nat$, that the misspeculation flag of
          the target configurations is $\top$, and that
          $\bufreadi {\mbuf_2} j = \val{}_1$. Since
          $\sem \expr_{\sttwo} =\sem \expr_{\st} = h$,
          Rule~\ref{SIE:SLoad-PSF} also applies to
          ${\cmemread \vx {\ar[\expr]} \sep \cmd, \sttwo}$ consuming
          the same directive and producing the same observation. More
          precisely, in this sub-derivation, we have:
          \small
          \begin{gather*}
            \sstep {\cmemread \vx {\ar[\expr]} \sep \cmd, \st} {\cmd,
               {\supdate{\reg_1}{\vx} {\val{}_1}, \mbuf_1, \mem_1, \top, \idx}} {\dload j} {\omem {(a, h)}} \\
            \sstep {\cmemread \vx {\ar[\expr]} \sep \cmd, \sttwo} {\cmd,
              {\supdate{\reg_2}{\vx} {\val{}_2}, \mbuf_2, \mem_2, \top, \idx}} {\dload j} {\omem {(a, h)}},
          \end{gather*}
          \normalsize
          where
          \begin{align*}
            \val{}_1 &= {\bufreadi {\mbuf_1} j}\\
            \val{}_2 &= {\bufreadi {\mbuf_2} j}.
          \end{align*}

          We can conclude the sub-derivation by employing (H),
          provided that we can establish except for
          \small
          \[
            {\supdate{\reg_1}{\vx} {\val{}_1}, \mbuf_1, \mem_1, \top, \idx}
            \confeq{{\supdate
                {\env_1} \vx {\tty}}}
            {\supdate{\reg_2}{\vx} {\val{}_2}, \mbuf_2, \mem_2, \top, \idx},
          \]
          \normalsize
          which a consequence of $\st \confeq{\env} \sttwo$ and
          \Cref{lemma:sloadaux}, showing that
          $\supdate{\reg_1}\vx {\val{}_1} \steq[\top]{ {\supdate
              {\env_1} \vx {\tty}} } \supdate{\reg_2}\vx
          {\val{}_2}$. Also
          $\supdate{\reg_1}\vx {{\val{}_1}} \boteq
          \supdate{\reg_2}\vx {\val{}_2}$ follows easily
          from $\reg_1 \boteq \reg_2$ (which is a consequence of
          $\st \confeq{\env} \sttwo$) and the observation that, since
          Rule~\ref{SIE:SLoad-PSF} applied, we have
          \[
            {\bufreadi {\mbuf_1} j}, {\bufreadi {\mbuf_2} j}\in \Vals.
          \]

          \proofcase{\ref{SIE:SLoad-Step}} By introspection of the
          rule, we deduce that $\dir = \dstep$,
          $\sem \expr_{\st}\in \Nat$,
          $\sem \expr_{\st} \le \size \ar$. Since
          $\sem \expr_{\st} =\sem \expr_{\sttwo}$,
          Rule~\ref{SIE:SLoad-Step} also applies to
          ${\cmemread \vx {\ar[\expr]} \sep \cmd, \sttwo}$ with the
          same directive and observation. In summary, we have:
          \begin{align*}
            \sstep {\cmemread \vx {\ar[\expr]} \sep \cmd, \st} {\cmd,
            \supdate \st {\vx} {\val{}_1}} {\dstep} {\omem {(a, h)}} \\
            \sstep {\cmemread \vx {\ar[\expr]} \sep \cmd, \sttwo} {\cmd,
            \supdate \sttwo {\vx} {\val{}_2}} {\dstep} {\omem {(a, h)}},
          \end{align*}
          where
          \begin{align*}
            \val{}_1 &= \buflookup {\bm{\mbuf_1}{\mem_1}} {\ar,
            h}\\
            \val{}_2 &= \buflookup {\bm{\mbuf_2}{\mem_2}} {\ar,
            h}.
          \end{align*}
          Notice that, in comparison to the previous proof case, now
          we do not have anymore the guarantee that the misspeculation
          flag of the target configuration is $\top$, but we know that
          the loaded value is the most recent associated to
          $\ar, \sem\expr_{\reg_i}$ in the buffered memory. As before,
          we conclude the proof by applying (H), which requires proving
          $\supdate \st {\vx} {\val{}_1}\confeq{{\supdate {\env_1} \vx
              {\tty}}} \supdate \sttwo {\vx} {\val{}_2}$. To this end,
          observe that
          $\supdate{\reg_1}\vx {\val{}_1} \boteq \supdate{\reg_2}\vx
          {\val{}_2}$,
          $\mbuf_1 \steq[\bms]{{\supdate {\env_1} \vx {\tty}} }
          \mbuf_2$ and
          $\mem_1 \steq{{\supdate {\env_1} \vx {\tty}} } \mem_2$ are
          trivial consequences of $\st \confeq{\env} \sttwo$, and
          $\supdate{\reg_1}\vx {v_1} \steq[\bms]{{\supdate {\env_1}
              \vx {\tty}}} \supdate{\reg_2}\vx {v_2}$ is a consequence
          of $\st \confeq{\env} \sttwo$, $\tty = (\env(\ar), \tyhigh)$, and
          \Cref{lemma:sloadstepaux}.

          \proofcase{\ref{SIE:SLoad-OOB}} By introspection of the
          rule, we deduce that $\dir = \doob {(\artwo, w)}$,
          $\sem \expr_{\reg_1}\in \Nat$,
          $\sem \expr_{\reg_1} \geq \size \ar$. Note that
          Rule~\ref{SIE:SLoad-OOB} also applies to
          ${\cmemread \vx {\ar[\expr]} \sep \cmd, \sttwo}$ with the
          same directive and observation. Moreover, by introspection
          of such rule, we deduce that $\bms = \top$. In summary, we
          have:
          \begin{align*}
            \sstep {\cmemread \vx {\ar[\expr]} \sep \cmd, \st} {\cmd,
            \supdate \st {\vx} {\val{}_1}} {\doob (\artwo, w)} {\omem {(a, h)}} \\
            \sstep {\cmemread \vx {\ar[\expr]} \sep \cmd, \sttwo} {\cmd,
            \supdate \sttwo {\vx} {\val{}_2}} {\doob (\artwo, w)} {\omem {(a, h)}},
          \end{align*}
          where
          \begin{align*}
            \val{}_1 &= \mem_1(\artwo, w)\\
            \val{}_2 &= \mem_2(\artwo,
            w).
          \end{align*}
          Again, we conclude the proof of the claim by applying (H),
          which requires us to prove
          $\supdate \st {\vx} {\val{}_1}\confeq{{\supdate {\env_1} \vx
              {\tty}}} \supdate \sttwo {\vx} {\val{}_2}$. To this end,
          observe that
          $\supdate{\reg_1}\vx {\val{}_1} \boteq \supdate{\reg_2}\vx
          {\val{}_2}$,
          $\mbuf_1 \steq[\top]{{\supdate {\env_1} \vx {\tty}} }
          \mbuf_2$ and
          $\mem_1 \steq{{\supdate {\env_1} \vx {\tty}} } \mem_2$ are
          trivial consequences of $\st \confeq{\env} \sttwo$. Finally,
          also observe that
          $\supdate{\reg_1}\vx {v_1} \steq[\top]{ {\supdate {\env_1}
              \vx {\tty}} } \supdate{\reg_2}\vx {v_2}$ is a
          consequence of $\st \confeq{\env} \sttwo$ and
          \Cref{lemma:sloadaux}.

          \proofcase{\ref{SIE:SLoad-SSB}} This case is analogous to the one of
          Rule~\ref{SIE:SLoad-PSF}.
        \end{proofcases}

        \proofcase{\ref{TY:TStore}}
        In this case, we can assume that the type derivation looks like the following one:
        \small
        \[
            \inferrev[\textsc{TSeq}]
            {
              \pt \deriv \sjudg {\env_1} {\cmemasgn  {\ar[\expr]} \exprtwo\sep \cmd} {\env_2}
            }
            {
              \inferrev[\textsc{TStore}]
              {
                \sjudg {\env_1} {\cmemasgn  {\ar[\expr]} \exprtwo} {{\supdate \env \ar {\max(\pi_1(\tty_\exprtwo),~\env(\ar))}}}
              }
              {
                \sjudg {\env_1} {\expr} {\typair \tylow}
                \quad
                \sjudg {\env_1} {\exprtwo} {\tty_\exprtwo}
              }
              \quad
              \pt_2
            }
          \]
          \normalsize
          where $\pt_2 \deriv\sjudg {{\supdate \env \ar {\max(\pi_1(\tty_\exprtwo),~\env(\ar))}}} \cmd \env_2 $. We fix $\nat \in\Nat$, and by applying \Cref{lemma:ihaux},
        we can reduce our goal to proving:
        \[
          \tag{C}
          \sem {\sjudg[] {{\env_1}} {\cmemasgn {\ar[\expr]}\exprtwo} {\env_2}}^\idx_{\nat}
        \]
        under the assumption
        \[
          \tag{H}
          \sem {\sjudg[] {\supdate \env \ar {\max(\pi_1(\tty_\exprtwo),~\env(\ar))}} {\cmd} {\env_2}}^\idx_{\nat}
        \]
        for a given $\idx \in \Nat$. By expanding the definition of $\sem \cdot$,
        we can rephrase the claim as follows:
        \begin{multline*}
          \forall \nat'\le\nat. \forall \st, \sttwo \in \States_{\idx}. \forall \Ds, \Os. \st \confeq{\env_1} \sttwo \Rightarrow\\
          \forall \cmd', \st'.\ \sstep[\nat'] {\cmemasgn  {\ar[\expr]} \exprtwo \sep \cmd, \st} {\cmd', \st'} \Ds \Os \Rightarrow\\
          \exists \sttwo'.\ \sstep[\nat'] {\cmemasgn {\ar[\expr]} \exprtwo  \sep \cmd, \sttwo} {\cmd', \sttwo'} \Ds \Os  \land (\cmd'=\cnil \Rightarrow  \st' \confeq{\env_2} \sttwo').
        \end{multline*}
        Fix $\nat' \in \Nat$,
        $\st = \reg_1, \mbuf_1, \mem_1, \bms, \idx$ and
        $\sttwo= \reg_2, \mbuf_2, \mem_2, \bms, \idx$ such that
        $\st \confeq{\env} \sttwo$, directives $\Ds$ and observations
        $\Os$. Note that when $\nat'=0$, the claim is
        trivial. Therefore, in the following, we assume without loss
        of generality $0< \nat'\le \nat$. By leveraging this
        observation and by introspection of the semantics, this also
        means that $\Os = \omem {(a, {\sem \expr}_\st)}: \Os'$ and
        $\Ds = \dir: \Ds'$ for some $\dir$, $\Ds'$, $\Os'$.  Under
        this assumption, our proof obligation is showing
        \begin{gather*}
          \tag{CA}
          \sstep[\nat'] {\cmemasgn {\ar[\expr]} \exprtwo \sep \cmd, \sttwo}
          {\cmd', \sttwo'} \Ds \Os \\
          \tag{CB}
          \cmd'=\cnil \Rightarrow \st'
          \confeq{\env_2} \sttwo'
        \end{gather*}
        where $\cmd', \st'$ are such that
        \[
          \sstep {\cmemasgn {\ar[\expr]} \exprtwo \sep \cmd, \st} {\cmd'',
            \st''} \dir {\omem {(a, {\sem \expr}_\st)}} \xrightarrow[\Ds']{\Os'}\!\!{}^{\nat'-1}\, \cmd', \st'.
        \]
        Note that from $\sjudg {\env_1} {\expr} {\typair \tylow}$,
        $\st \confeq{\env} \sttwo$ and \Cref{lemma:tyexpr}, we deduce
        that $\sem \expr_{\reg_2} =\sem \expr_{\reg_1}$. As for load
        operations, there are multiple rules matching store
        instructions that can have been applied for showing this
        transition. The proof proceeds by cases on that rule.
        \begin{proofcases}
          \proofcase{\ref{SIE:SStore-Step}} By introspection of the
          rule, we deduce that $\dir = \dstep$,
          $\sem \exprtwo_{\reg_1}\in \Vals$,
          $\sem \expr_{\reg_1}\in \Nat$, and
          $\sem \expr_{\reg_1}< \size{\ar}$.  Since
          $\sem \expr_{\reg_2} =\sem \expr_{\reg_1}$,
          Rule~\ref{SIE:SStore-Step} also applies to
          ${\cmemasgn {\ar[\expr]} \exprtwo \sep \cmd, \sttwo}$,
          producing the same observation.  In summary, we have:
          \begin{align*}
            \sstep {\cmemasgn {\ar[\expr]} \exprtwo \sep \cmd, \st}
            {\cmd, \bitem {(\ar, h)} {\val{}_1} \cons \st} \dstep {\omem {(a, h)}}\\
            \sstep {\cmemasgn {\ar[\expr]} \exprtwo \sep \cmd, \sttwo}
            {\cmd, \bitem {(\ar, h)} {\val{}_2}\cons \sttwo} \dstep {\omem {(a, h)}}
          \end{align*}
          where
          \begin{align*}
            \val{}_1 &= \sem \exprtwo_\st &
            \val{}_2 &= \sem \exprtwo_\sttwo.
          \end{align*}
          We conclude the proof of the claim by applying (H).  To
          discharge the premise of (H), we need to prove that
          \[
            \supdate \st {(\ar, h)} {\val{}_1}
            \confeq{\supdate \env \ar {\max(\pi_1(\tty_\exprtwo),~\env(\ar))}}
            \supdate \sttwo {(\ar, h)} {\val{}_2}
          \]

          Observe that $\reg_1 \boteq \reg_2$ is a consequence of
          $\st \confeq{\env} \sttwo$; $\mem_1 \steq{\env_1'} \mem_2$
          and $\reg_1 \steq[\bms]{{{\env_1'}} } \reg_2$ follow from
          $\st \confeq{\env} \sttwo$ and \Cref{lemma:steqweak}. To
          conclude, we only have to establish
          $\bitem{(\ar, \sem
            \expr_{\reg_1})}{h}\cons\mbuf_1
          \steq[\bms]{\env_1' } \bitem{(\ar, \sem
            \expr_{\reg_2})}{h}\cons\mbuf_2$. This
          result, in turn, is a consequence of \Cref{lemma:sstoreaux}.

          \proofcase{\ref{SIE:SStore-OOB}} By introspection of the
          rule, we deduce that $\dir = \doob {(\artwo, w)}$,
          $\sem \exprtwo_{\reg_1}\in \Vals$,
          $\sem \expr_{\reg_1}\in \Nat$, and
          $\sem \expr_{\reg_1} \ge \size\ar$. Since
          $\sem \expr_{\reg_2} =\sem \expr_{\reg_1}$, rule
          \ref{SIE:SStore-OOB} also applies to $\confone_2$ with the
          same directive and observation. Moreover, due to the
          safety assumption, we deduce that $\bms = \top$.
          In summary, we have:
          \begin{align*}
            \begin{multlined}
              \sstep {\cmemasgn {\ar[\expr]} \exprtwo \sep \cmd, \st}
              {} {\doob {(\artwo, w)}} {\omem {(a, h)}}\\ {\cmd, \reg_1, \bitem {(\artwo, w)} {\sem \exprtwo_{\reg_1}} \cons \mbuf_1, \mem_1, \top, \idx}
            \end{multlined}\\
            \begin{multlined}
              \sstep {\cmemasgn {\ar[\expr]} \exprtwo \sep \cmd, \sttwo}
              {} {\doob {(\artwo, w)}} {\omem {(a, h)}}\\ {\cmd, \reg_2, \bitem {(\artwo, w)} {\sem \exprtwo_{\reg_2}} \cons \mbuf_2, \mem_2, \top, \idx}.
            \end{multlined}
          \end{align*}

          We conclude the proof with an application of (H). To
          discharge its premise, we need to prove that the stores
          ${\reg_1, \bitem {(\artwo, w)} {\sem \exprtwo_{\reg_1}}
            \cons \mbuf_1, \mem_1, \top, \idx}$ and
          $\reg_2, \bitem {(\artwo, w)} {\sem \exprtwo_{\reg_2}} \cons
          \mbuf_2, \mem_2, \top, \idx$ are
          $\confeq{\supdate \env \ar
            {\max(\pi_1(\tty_\exprtwo),~\env(\ar))}}$-related. To this
          end, observe that $\reg_1 \boteq \reg_2$ is a consequence of
          $\st \confeq{\env} \sttwo$, 
          \[
            \mem_1 \steq{\supdate \env \ar
              {\max(\pi_1(\tty_\exprtwo),~\env(\ar))}} \mem_2,
          \]
          and
          $\reg_1 \steq[\top]{\supdate \env \ar
            {\max(\pi_1(\tty_\exprtwo),~\env(\ar))}} \reg_2$ follow from
          typability, $\st \confeq{\env} \sttwo$ and \Cref{lemma:steqweak}. The conclusion
          $\bitem{(\artwo, t)}{\sem\exprtwo_{\reg_1}}\cons\mbuf_1
          \steq[\top]{ {{\env_1'}} } \bitem{(\artwo,
            t)}{\sem\exprtwo_{\reg_1}}\cons\mbuf_2$ is trivial because
          the misspeculation flag is $\top$.
         \qedhere
        \end{proofcases}
      \end{proofcases}
    \end{proofcases}
  \end{proofcases}
\end{proof}

Our next goal is to prove that
$\sem{\sjudg[] \env \cmd \envtwo}^{\infty}_\nat$ is
Scott-continuous. To this end, we need two auxiliary lemmas relating
reductions sequences with the ordinary and indexed semantics.

\begin{lemma}
  \label{lemma:scottcontinuityaux1}
  Let $\confone_1= {\cmd, \reg_1, \mbuf_1, \mem_1, \bms}$ and
  $\confone_2$ be two
  configurations. For
  every $\nat, \Ds, \Os$, such that
  $\nsstep \nat{\confone_1}{\confone_2}\Ds \Os$ with the ordinary
  semantics, there exist $\idx, \idx'$ such that
  $\nsstep \nat{\confone_1, \idx}{\confone_2, \idx'}\Ds \Os$ with the indexed semantics.
\end{lemma}
\begin{proof}
  It is sufficient to take any $\idx \ge \nat$, with this choice, the
  claim can be proven easily by induction on $\nat$.
\end{proof}

\begin{lemma}
  \label{lemma:scottcontinuityaux2}
  Let $\confone_1= {\cmd, \reg_1, \mbuf_1, \mem_1, \bms}$ and
  $\confone_2$ be two ordinary configurations. For every
  $\nat, \Ds, \Os$, such that
  $\nsstep \nat{\confone_1, \idx}{\confone_2, \idx'}\Ds \Os$ with the
  indexed semantics, we have
  $\nsstep \nat{\confone_1}{\confone_2, \idx'}\Ds \Os$ with the
  ordinary semantics.
\end{lemma}
\begin{proof}
  By induction on $\nat$. At each step it suffices to apply the
  corresponding rule of the ordinary semantics.
\end{proof}

It is now time to prove Scott-continuity.

\begin{lemma}[Scott-continuity]
  \label{lemma:sc}
  \[
    \forall \nat \in \Nat. \left(\forall \idx \in \Nat.\ \sem{\sjudg[] \env \cmd \envtwo}^{\idx}_\nat
    \right)\Rightarrow \sem{\sjudg[] \env \cmd \envtwo}^{\infty}_\nat
  \]
\end{lemma}
\begin{proof}
  The proof is by contraposition.  Fix $\nat$, assume
    \begin{gather*}
    \lnot \sem{\sjudg[] \env \cmd \envtwo}^{\infty}
    \tag{HA}\\
    \forall \idx . \sem{\sjudg[] \env \cmd \envtwo}^{\idx}
    \tag{HB}
  \end{gather*}
  By expanding (HA), we have:
  \begin{gather*}
    \tag{H1}
    \st \confeq\env \sttwo \\
    \tag{H2}
    \sstep[\nat'] {\cmd, \st} {\cmd', \st'} \Ds \Os \\
    \tag{H3}
    \forall \sttwo'. \lnot \sstep[\nat'] {\cmd, \sttwo} {\cmd', \sttwo'} \Ds \Os \lor
    \left (\cmd = \cnil \land
    \st' \not\confeq\envtwo \sttwo' \right)
  \end{gather*}
  for some $\nat', \st,\sttwo, \cmd', \st'$.
  By applying \Cref{lemma:scottcontinuityaux1} to (H2), we obtain
  \[
    \sstep[\nat'] {\cmd, \st, \idx} {\cmd', \st', \idx'} \Ds \Os \\
    \tag{H2'}
  \]
  By applying (HA) to $\nat', \nat',\st, \sttwo$, (H1), $\cmd', \st'$ and (H2') we prove
  \[
    \exists \sttwo', \idx'.\ \sstep[\nat'] {\cmd, \sttwo, \idx} {\cmd', \sttwo', \idx'} \Ds \Os  \land \cmd = \cnil \Rightarrow  \st', \idx \confeq\envtwo \sttwo', \idx'
  \]
  fix $\sttwo', \idx'$. By applying \Cref{lemma:scottcontinuityaux2} we
  establish $\sstep[\nat'] {\cmd, \sttwo} {\cmd', \sttwo'} \Ds \Os$
  with the non-indexed semantics. This last observation, in
  conjunction with $\cmd = \cnil \Rightarrow  \st', \idx \confeq\envtwo \sttwo', \idx'$
  contradicts (H3).
\end{proof}

Soundness of our type system for $\sem\cdot^\infty$ is now a direct
consequence of Scott-continuity and \Cref{lemma:mainlemma}.

\begin{corollary}[Soundness]
  \label{cor:soundness}
  \[
    \pt\deriv \sjudg[\emptyset] \env \cmd \envtwo \quad \Rightarrow \quad
    \forall \nat \in \Nat.\ \sem{\sjudg[] \env \cmd \envtwo}^{\infty}_\nat.
  \]
\end{corollary}
\begin{proof}
  Assume
  \[
    \pt\deriv \sjudg[\emptyset] \env \cmd \envtwo
    \tag{H}
  \]
  The claim is:
  \[
    \forall \nat \in \Nat.\ \sem{\sjudg[] \env \cmd \envtwo}^{\infty}_\nat.
    \tag{C}
  \]
  From (H) and \Cref{lemma:mainlemma} (instantiated with $\jmpctx=\emptyset$), we prove
  \[
    \forall \nat, \idx \in \Nat.\ \sem{\sjudg[] \env \cmd \envtwo}^{\idx}_\nat.
  \]
  We conclude the proof with an application of \Cref{lemma:sc}.
\end{proof}

We now turn our attention to proving that typable programs are
constant-time, as stated in \Cref{thm:soundness}. To this end, we
start by proving that whenever
$\sem{\sjudg[] \env \cmd \envtwo}^{\infty}_\nat$ holds, then $\cmd$ is
$\confeq\env$-SCT, as stated by \Cref{lemma:glue}
below. \Cref{thm:soundness} will then be a direct consequence of
\Cref{cor:soundness} and \Cref{lemma:glue}.

\begin{lemma}
  \label{lemma:glue}
  If for every $\nat \in \Nat$, we have
  $\sem{\sjudg[] \env \cmd \envtwo}^{\infty}_\nat$, then
  $\cmd$ is $\confeq\env$-SCT.
\end{lemma}

\begin{proof}
  Assume that $\cmd$ is not $\confeq\env$-SCT. This means that there
  are $\bms$,  $\st_1 \confeq \env \st_2$ and two reductions
  \[
    \sstep{\cmd, \st_1}{\cmd'_1 , \st'_1}{\Ds}{\Os_1}
  \]
  and
  \[
    \sstep{\cmd, \st_2}{\cmd'_2 , \st'_2}{\Ds}{\Os_2}
  \]
  such that $\Os_1 \neq \Os_2$.  Let $\nat$ be the number of
  directives in $\Ds$. From $\sem{\sjudg[] \env \cmd \envtwo}^{\infty}_\nat$,
  we deduce that there is some $\sttwo$ such that
  \[
    \sstep[\nat]{\cmd, \st_2}{\cmd'_1 , \sttwo}{\Ds}{\Os_1}.
  \]
  Since our semantics is deterministic, it must be the case where
  $\cmd'_1 = \cmd'_2$, $\sttwo = \st'_2$ and,
  in particular, $\Os_1=\Os_2$.
\end{proof}

\begin{proof}[Proof of \Cref{thm:soundness}]
  The proof is by contraposition. Assume that
  $\pt\deriv \sjudg[\emptyset] \env \cmd \envtwo$ but $\cmd$ is
  not $\confeq\env$-SCT. From \Cref{lemma:glue}, we deduce
  $\lnot \sem{\sjudg[] \env \cmd \envtwo}^{\infty}_\nat$ for some $\nat \in \Nat$,
  which contradicts \Cref{cor:soundness}.
\end{proof}

\subsection{Technical observations}
\label{sec:apptech}

\begin{lemma}
  \label{lemma:sepaux}
  For every $\nat \in \Nat$, programs $\cmd, \cmd', \cmdtwo \in \Cmd$,
  states $\st,\st' \in \States$, sequence of directives $\Ds$ and sequence of observations
  $\Os$:
  \begin{align*}
    \sstep[\nat] {\cmd, \st} {\cmd', \st'} \Ds\Os\ \  &\Rightarrow\ \
                                                        \sstep[\nat] {\cmd\sep\cmdtwo, \st} {\cmd'\sep\cmdtwo, \st'} \Ds\Os.\\
    \sstep[\nat] {\cmd, \st} {\err} \Ds\Os\ \  &\Rightarrow\ \
                                                 \sstep[\nat] {\cmd\sep\cmdtwo, \st} {\err} \Ds\Os.
  \end{align*}
\end{lemma}
\begin{proof}
  By induction on $\nat$.
  \begin{proofcases}
    \proofcase{0} Trivial.
    \proofcase{$\nat= \nat'+1$} By case analysis on the rule used for the first transition and by applying the induction hypothesis to the target configuration.
  \end{proofcases}
\end{proof}

\begin{lemma}
  \label{lemma:ihaux}
  Fix $\stat \in \Instr$, $\cmd\in \Cmd$, $\nat\in \Nat$, jump context $\jmpctx$ and environments $\env_1, \env_2, \env_3$. If
  \[
    \sem {\sjudg[] {\env_2} {\cmd} {\env_3}}^{0}_\nat \land
    \forall \idx \in \Nat.\ \sem \jmpctx^{\idx}_\nat \Rightarrow \sem {\sjudg[] {\env_2} {\cmd} {\env_3}}^{\idx+1}_\nat.
  \]
  and
  $\forall \idx \in \Nat. \sem {\sjudg[] {\env_2} {\cmd} {\env_3}}^{\idx}_\nat \Rightarrow \sem {\sjudg[] {\env_1} {\stat\sep\cmd} {\env_3}}^{\idx}_\nat$, then:
  \[
    \sem {\sjudg[] {\env_1} {\stat\sep \cmd} {\env_3}}^{0}_\nat \land
    \forall \idx \in \Nat.\ \sem \jmpctx^{\idx}_\nat \Rightarrow \sem {\sjudg[] {\env_1} {\stat\sep \cmd} {\env_3}}^{\idx+1}_\nat.
  \]
\end{lemma}
\begin{proof}
  We fix our assumptions
    \begin{gather*}
    \tag{H1}
    \sem {\sjudg[] {\env_2} {\cmd} {\env_3}}^{0}_\nat \land
    \forall \idx \in \Nat.\ \sem \jmpctx^{\idx}_\nat \Rightarrow \sem {\sjudg[] {\env_2} {\cmd} {\env_3}}^{\idx+1}_\nat\\
    \tag{H2}
    \forall \idx \in \Nat. \sem {\sjudg[] {\env_2} {\cmd} {\env_3}}^{\idx}_\nat \Rightarrow \sem {\sjudg[] {\env_1} {\stat\sep\cmd} {\env_3}}^{\idx}_\nat,
  \end{gather*}
  and our claims:
  \begin{gather*}
    \tag{C1}
        \sem {\sjudg[] {\env_1} {\stat\sep \cmd} {\env_3}}^{0}_\nat \\
    \tag{C2}
    \forall \idx \in \Nat.\ \sem \jmpctx^{\idx}_\nat \Rightarrow \sem {\sjudg[] {\env_1} {\stat\sep \cmd} {\env_3}}^{\idx+1}_\nat.
  \end{gather*}
  \begin{proofcases}
    \proofcase{C1} Direct consequence of (H2) and the first conjunct of (H1).
    \proofcase{C2} Let $j$ be an arbitrary index. We must prove
    \[
      \sem {\sjudg[] {\env_1} {\stat\sep \cmd} {\env_3}}^{j+1}_\nat
    \]
    under the assumption
    \[
      \tag{H}
      \sem \jmpctx^{j}_\nat.
    \]
    The conclusion is a direct consequence of (H2), the second conjunct in (H1) and (H).
    \qedhere
  \end{proofcases}
\end{proof}

\begin{lemma}
  \label{lemma:steqweak}
  For every pair of states $\st_1, \st_2$ such that
  $\st_1\confeq \env \st_2$, if
  $\env \le \envtwo$ then we have
  $\st_1 \confeq \envtwo \st_2$.
\end{lemma}
\begin{proof}
  For the proof we need the auxiliary claims:
  \begin{align*}
    \forall \bms, \env, \envtwo, \reg_1, \reg_2. \reg_1 \steq[\bms]{\env} \reg_2 \Rightarrow \reg_1 \steq[\bms]{\envtwo} \reg_2 \tag{C1}\\
    \forall \bms, \env, \envtwo, \mem_1, \mem_2. \mem_1 \steq{ \env} \mem_2 \Rightarrow \mem_1 \steq{\envtwo} \mem_2 \tag{C2}\\
    \forall \bms, \env, \envtwo, \mbuf_1, \mbuf_2. \mbuf_1 \steq[\bms]{\env} \mbuf_2 \Rightarrow \mbuf_1 \steq[\bms]{\envtwo} \mbuf_2 \tag{C3}\\
  \end{align*}
  \begin{itemize}
  \item[(C1)] The proof goes by case analysis on $\bms$.
    \begin{proofcases}
      \proofcase{$\bot$} Fix the quantified variables,  $\reg_1 \steq[\bot]{\env} \reg_2$. From the assumptions on the register maps, we deduce
      \begin{equation*}
        \forall \vx \in \Regs. \pi_1(\env  (\vx)) = \tylow \lor \pi_2(\env(\vx)) = \tylow \Rightarrow\\
        \reg_1(\vx)=\reg_2(\vx)
        \tag{H}
      \end{equation*}
      To conclude the proof, our goal is to show $\reg_1 \steq[\bot]{ \envtwo} \reg_2$, i.e.
      \begin{equation*}
        \forall \vx \in \Regs. \pi_1( \envtwo (\vx)) = \tylow \lor \pi_2( \envtwo(\vx)) = \tylow \Rightarrow\\ \reg_1(\vx)=\reg_2(\vx).
        \tag{C1'}
      \end{equation*}
      Fix a register $\vx$. Assume $\pi_1(\envtwo (\vx)) = \tylow \lor \pi_2(\envtwo(\vx)) = \tylow$ we go by case analysis:
      \begin{proofcases}
        \proofcase{$\pi_1(\envtwo (\vx))=\tylow$} From the assumption on the environments, we deduce that, $\pi_1(\env (\vx)) = \tylow$ and  we have $\env (\vx) = (\tylow, \ty')$ for some $\ty'$, therefore the conclusion is a trivial consequence of (H).
        \proofcase{$\pi_2(\envtwo (\vx))=\tylow$} From the assumption on the environments, we deduce that $\pi_2(\env (\vx))= \tylow$, Therefore $\env (\vx) = (\ty', \tylow)$ for some $\ty'$ and the claim is a trivial consequence of (H).
      \end{proofcases}
  \proofcase{$\top$} Analogous to the previous case.
\end{proofcases}
  \item[(C2)] Analogous to the proof of (C1).
  \item[(C3)] By induction on $\mbuf_1$.
    \begin{proofcases}
      \proofcase{$\nil$} Trivial.
      \proofcase{${\bitem {(\ar,\nat)} {\val{}_1}} \cons\mbuf_1$} By cases on $\bms$.
      \begin{proofcases}
        \proofcase{$\bot$} Fix all the universally quantified variables. Because of  the main assumptions and the case analysis, in this sub-derivation we have the following hypotheses:
        \begin{gather*}
          {{\bitem {(\ar,\nat)} {\val{}_1}} \cons\mbuf_1}  \steq[\bot]{\env} {{\bitem {(\ar,\nat)} {\val{}_2}} \cons\mbuf_2}\tag{H1}\\
              {\mbuf_1}  \steq[\bot]{\env} {\mbuf_2}
                          \tag{H2}\\
               \pi_1(\env(\ar)) = \tylow\lor \pi_2(\env(\ar)) = \tylow \Rightarrow \val{}_1 = \val{}_2
\tag{H3}
        \end{gather*}
        and we want to conclude:
        \begin{align*}
          {\mbuf_1}  &\steq[\bot]{\envtwo} {\mbuf_2}
                      \tag{C1}\\
          \envtwo(\ar) = \tylow &\Rightarrow \val{}_1 = \val{}_2
                                                     \tag{C2}
        \end{align*}
        The claim (C1) follows easily from the IH and (H2). For claim (C2), we fix a type $\ty$ and an array $\ar$, and  we assume
        \(
          \envtwo(\ar) = \tylow.
        \)
        The remaining of the proof is analogous to the inner case analysis in (C1).
      \end{proofcases}
      \proofcase{$\top$} Analogous to the previous point.
    \end{proofcases}
  \end{itemize}
  The conclusion is a direct consequence of (C1)-(C3).
\end{proof}

\begin{corollary}
  \label{lemma:monotone}
  Let $\env \le \envtwo$ and $\envtwo'\le \env'$. We have:
  \[
    \forall \idx, \nat \in \Nat.\sem{\sjudg[] \envtwo \cmd \envtwo'}^{\idx}_\nat \Rightarrow \sem{\sjudg[] \env \cmd \env'}^{\idx}_\nat.
  \]
\end{corollary}
\begin{proof}
  The proof follows straightforwardly by expanding the definition of
  $\sem{\sjudg[] \envtwo \cmd \envtwo'}^{\idx}_\nat$ and applying \Cref{lemma:steqweak}. Our goal is to show:
  \begin{multline*}
    \forall \st, \sttwo. \forall \Ds, \Os. \st \confeq\env \sttwo \Rightarrow
    \exists \cmd', \st'.\ \sstep[\nat] {\cmd, \st} {\cmd', \st'} \Ds \Os \Rightarrow\\ \exists \sttwo'.\ \sstep[\nat] {\cmd, \sttwo} {\cmd', \sttwo'} \Ds \Os  \land (\cmd'=\cnil \Rightarrow  \st' \confeq{\env'} \sttwo').
  \end{multline*}
  Assume $\st \confeq\env \sttwo$. Then, by \Cref{lemma:steqweak}, we obtain
$\st \confeq\envtwo \sttwo$. The conclusion now follows from the semantics
$\sem{\sjudg[] \envtwo \cmd \envtwo'}^{\idx}_\nat$ and another application of \Cref{lemma:steqweak}.
\end{proof}

\begin{lemma}
  \label{lemma:monotonetech}
  \small
  \begin{multline*}
    \forall \nat \in \Nat. \forall \st \in \States_{\idx+1}. \forall \Ds, \Os. \sstep[\nat] {\cmd, \st} {\cmd', \st'} \Ds \Os \land \st'  \in \States_{\idx'+1} \Leftrightarrow\\
    \forall \Ds, \Os. \sstep[\nat] {\cmd, \st[\idx+1 \upd \idx]} {\cmd', \st'} \Ds \Os \land \st'[\idx'+1 \upd \idx']
  \end{multline*}
\end{lemma}
\begin{proof}
  By induction on $\nat$.
\end{proof}

\begin{lemma}
  \label{lemma:monotone2}
  \[
    \forall \idx, \nat \in \Nat.\sem{\sjudg[] \env \cmd \env'}^{\idx+1}_\nat \Rightarrow \sem{\sjudg[] \env \cmd \env'}^{\idx}_\nat.
  \]
\end{lemma}
\begin{proof}
  The proof follows by expanding the definition of
  $\sem{\sjudg[] \env \cmd \env'}^{\idx+1}_\nat$, which is equivalent to:
  \begin{multline*}
    \forall \st, \sttwo \in \States_{\idx+1}. \forall \Ds, \Os. \st \confeq\env \sttwo \Rightarrow
    \exists \cmd', \st'.\ \sstep[\nat] {\cmd, \st} {\cmd', \st'} \Ds \Os \Rightarrow\\ \exists \sttwo'.\ \sstep[\nat] {\cmd, \sttwo} {\cmd', \sttwo'} \Ds \Os  \land (\cmd'=\cnil \Rightarrow  \st' \confeq{\env'} \sttwo').
  \end{multline*}
  The claim states:
    \begin{multline*}
    \forall \st, \sttwo \in \States_{\idx}. \forall \Ds, \Os. \st \confeq\env \sttwo \Rightarrow
    \exists \cmd', \st'.\ \sstep[\nat] {\cmd, \st} {\cmd', \st'} \Ds \Os \Rightarrow\\ \exists \sttwo'.\ \sstep[\nat] {\cmd, \sttwo} {\cmd', \sttwo'} \Ds \Os  \land (\cmd'=\cnil \Rightarrow  \st' \confeq{\env'} \sttwo').
  \end{multline*}
  Fix $\st, \sttwo \in \States_{\idx}$, $\Ds, \Os$, and assume
  $\st \confeq\env \sttwo$ and
  $\sstep[\nat] {\cmd, \st} {\cmd', \st'} \Ds \Os$ with
  $\st'\in \States_{\idx'}$.  We apply \Cref{lemma:monotonetech} to
  deduce
  $\sstep[\nat] {\cmd, \st[\idx \upd \idx+1]} {\cmd', \st'[\idx' \upd
    \idx'+1]} \Ds \Os$, so we can apply the assumption and deduce
  $\sstep[\nat] {\cmd, \sttwo[\idx \upd \idx+1]} {\cmd', \sttwo'} \Ds
  \Os$.  for some $\sttwo' \in \States_{\idx'+1}$. With another
  application of \Cref{lemma:monotonetech}, we deduce
  $\sstep[\nat] {\cmd, \sttwo} {\cmd', \sttwo'[\idx'+1\upd\idx']} \Ds
  \Os$, which concludes our proof.
\end{proof}

\begin{lemma}
  \label{lemma:auxfence}
  For every pair of states
  $\st_1 = \sframe{\reg_1}{\mbuf_1, \mem_1}{\bot},
  \st_2 = \sframe{\reg_2}{\mbuf_2, \mem_2}{\bot}$
  such that
  $\st_1\confeq \env \st_2$, if
  \begin{gather*}
    \forall \vx \in \Regs.\pi_1({\env} (\vx)) = \tylow \Rightarrow {\envtwo}(\vx) = \typair \tylow\\
    \forall \vx \in \Regs.\pi_1({\env} (\vx)) = \tyhigh \Rightarrow {\envtwo} (\vx) = \typair \tyhigh\\
    \forall \ar \in \Arr.\env(\ar) = \tylow \Rightarrow {\envtwo}(\ar) = \tylow\\
    \forall \ar \in \Arr.\env(\ar) = \tyhigh \Rightarrow \envtwo (\ar) = \tyhigh,
  \end{gather*}
  then we have
  $\st_1 \confeq \envtwo \st_2$.
\end{lemma}
\begin{proof}
  Assume
  \begin{gather*}
    \forall \vx \in \Regs.\pi_1({\env} (\vx)) = \tylow \Rightarrow {\envtwo}(\vx) = \typair \tylow
    \tag{HA}\\
    \forall \vx \in \Regs.\pi_1({\env} (\vx)) = \tyhigh \Rightarrow {\envtwo} (\vx) = \typair \tyhigh
    \tag{HB}\\
    \forall \ar \in \Arr.\env(\ar) = \tylow \Rightarrow {\envtwo}(\ar) = \tylow
    \tag{HC}\\
    \forall \ar \in \Arr.\env(\ar) = \tyhigh \Rightarrow \envtwo (\ar) = \tyhigh,
    \tag{HD}
  \end{gather*}
  To conclude  the proof we need the auxiliary claims:
  \begin{align*}
    \forall  \reg_1, \reg_2.  \reg_1 \steq[\bot]{\env} \reg_2 \Rightarrow \reg_1 \steq[\bot]{\envtwo} \reg_2 \tag{C1}\\
    \forall \mem_1, \mem_2.  \mem_1 \steq{\env} \mem_2 \Rightarrow \mem_1 \steq{\envtwo} \mem_2 \tag{C2}\\
    \forall \mbuf_1, \mbuf_2.  \mbuf_1 \steq[\bot]{\env} \mbuf_2 \Rightarrow \mbuf_1 \steq[\bot]{\envtwo } \mbuf_2 \tag{C3}
  \end{align*}
  \begin{itemize}
  \item[(C1)] Fix the quantified variables, assume $\reg_1 \steq[\bot]{\env} \reg_2$. From the assumptions on the register maps, we deduce
    \[
\forall \vx \in \Regs. \pi_1(\env (\vx)) = \tylow \lor \pi_2( \env(\vx)) = \tylow \Rightarrow \reg_1(\vx)=\reg_2(\vx) .
      \tag{H1}
    \]
    To conclude the proof, our goal is to show $\reg_1 \steq[\bot]{ \envtwo} \reg_2$, i.e.
    \[
      \forall \vx \in \Regs. \pi_1( \envtwo (\vx)) = \tylow \lor \pi_2( \envtwo(\vx)) = \tylow \Rightarrow \reg_1(\vx)=\reg_2(\vx).
    \]
    Fix a register $\vx$ and assume
    \[
      \pi_1( \envtwo (\vx)) = \tylow \lor \pi_2( \envtwo(\vx)) = \tylow.
    \]
    We go by case analysis:
    \begin{proofcases}
      \proofcase{$\pi_1( \envtwo (\vx)) = \tylow$} By case analysis on the value of $\pi_1({\env} (\vx))$, and by (HA) and (HB), we observe that $\pi_1({\env} (\vx))=\tylow$. The conclusion follows directly from this intermediate observation and (H1).
      \proofcase{$\pi_2( \envtwo (\vx)) = \tylow$} Analogous to the previous case.
    \end{proofcases}

  \item[(C2)] Analogous to the proof of (C1), but using (HC) and (HD) instead of (HA) and (HB).
  \item[(C3)] By induction on $\mbuf_1$.
    \begin{proofcases}
      \proofcase{$\nil$} Trivial.
      \proofcase{${\bitem {(\ar,\nat)} {\val{}_1}} \cons\mbuf_1$} Fix all the universally quantified variables. Because of  the main assumptions and the case analysis, in this sub-derivation we have the following hypotheses:
        \begin{align*}
          {{\bitem {(\ar,\nat)} {\val{}_1}} \cons\mbuf_1}  &\steq[\bot]{ \env} {{\bitem {(\ar,\nat)} {\val{}_2}} \cons\mbuf_2}\tag{H1}\\
              {\mbuf_1}  &\steq[\bot]{ \env} {\mbuf_2}
                          \tag{H2}\\
               \env(\ar) = \tylow &\Rightarrow \val{}_1 = \val{}_2
\tag{H3}
        \end{align*}
        and we want to conclude:
        \begin{align*}
          {\mbuf_1}  &\steq[\bot]{ \envtwo} {\mbuf_2}
                      \tag{C1}\\
          \envtwo(\ar) = \tylow &\Rightarrow \val{}_1 = \val{}_2
                                                     \tag{C2}
        \end{align*}
        The claim (C1) follows easily from the IH and (H2). For claim (C2), we fix an array $\ar$, and  we assume $\envtwo(\ar) = \tylow$. From (HC), (HD) and a case analysis on $\env (\ar)$ we conclude $\env (\ar) = \tylow$, therefore (C2) is a trivial consequence of (H3).
    \end{proofcases}
  \end{itemize}
  The conclusion is a direct consequence of (C1)-(C3).
\end{proof}

\begin{lemma}
  \label{lemma:sfenceaux}
  If $\reg_1 \steq[\bms] { \env} \reg_2$,
  \[
    z_i =
    \begin{cases}
      \bot & \text{if } \bms  = \top\\
      \reg_i(\vx) & \text{otherwise}
    \end{cases}
  \]
  and $\reg_1 \boteq \reg_2$, then $\supdate {\reg_1}{\vx}{z_1} \steq[\bms] { {\supdate \env \vx {\typair{\pi_1(\env(\vx))}}} } \supdate {\reg_2}{\vx}{z_2}$ and $\supdate {\reg_1}{\vx}{z_1} \boteq \supdate {\reg_2}{\vx}{z_2}$.
\end{lemma}
\begin{proof}
  Our assumptions are:
  \begin{align*}
    \reg_1 &\steq[\bms] { \env} \reg_2\tag{H1}\\
    \reg_1 &\boteq \reg_2\tag{H2}
  \end{align*}
  and the claim we want to establish are:
  \begin{align*}
    \supdate {\reg_1}{\vx}{z_1} &\steq[\bms] { {\supdate \env \vx {\typair{\pi_1(\env(\vx))}}} } \supdate {\reg_2}{\vx}{z_2}\tag{C1}\\
    \supdate {\reg_1}{\vx}{z_1} &\boteq \supdate {\reg_2}{\vx}{z_2}\tag{C2}
  \end{align*}
  \begin{itemize}
  \item[(C1)] The proof is by case analysis on $\bms$:
    \begin{proofcases}
      \proofcase{$\bot$} The claim reduces to:
      \(
        {\reg_1} \steq[\bot] {{\supdate \env \vx {\typair{\pi_1(\env(\vx))}}} } {\reg_2}
      \)
      which means proving:
      \[
        \supdate{\reg_1}\vx {z_1}(\vy) = \supdate{\reg_2}\vx {z_2}(\vy).
      \]
      for a fresh variable $\vy$ under the assumption
      \begin{multline*}
        \pi_1({\supdate \env \vx {\typair{\pi_1(\env(\vx))}}} (\vy)) = \tylow \lor\\
        \pi_2({\supdate \env \vx {\typair{\pi_1(\env(\vx))}}} (\vy)) = \tylow.
        \tag{H}
      \end{multline*}
      If $y \neq x$, the claim is a direct consequence of (H) and (H1). When $y=x$, we have to conclude $z_1=z_2$. By instantiating $\vy$ with $\vx$ in (H), we conclude that $\pi_1(\env(\vx))=\tylow$, from this intermediate result and (H1), we deduce $\reg_1(\vx)=\reg_2(\vx)$ and, therefore, $z_1=z_2$.
      \proofcase{$\top$} The claim reduces to:
      \(
      {\reg_1} \steq[\top] {{\supdate \env \vx {\typair{\pi_1(\env(\vx))}}} } {\reg_2}
      \)
      which means proving:
      \[
        \supdate{\reg_1}\vx {z_1}(\vy) = \supdate{\reg_2}\vx {z_2}(\vy).
      \]
      for a fresh variable $\vy$ under the assumption
      \[
        \pi_2({\supdate \env \vx {\typair{\pi_1(\env(\vx))}}} (\vy)) = \tylow
        \tag{H}
      \]
      If $y \neq x$, the claim is a direct consequence of (H) and (H1). When $y=x$, we have to conclude $z_1=z_2$, which is now trivial because $\bms=\top$.
    \end{proofcases}
  \item[(C2)] This claim follows easily from a case analysis on $\bms$ and (H2).
  \end{itemize}
\end{proof}

\begin{lemma}
  \label{lemma:asgnaux}
  If $\sjudg \env \expr \sigma$,  $\reg_1 \steq[\bms] {\env} \reg_2$ and
  $\reg_1 \boteq \reg_2$,
  then
  \[
    \supdate {\reg_1}{\vx}{\sem \expr_{\reg_1}} \steq[\bms] {{\supdate \env \vx \sigma}} \supdate {\reg_2}{\vx}{\sem \expr_{\reg_2}}.
  \]
\end{lemma}
\begin{proof}
  The assumptions we have are:
  \begin{gather*}
    \sjudg \env \expr \sigma \tag{H1}\\
    \reg_1 \steq[\bms] {\env} \reg_2 \tag{H2}\\
    \reg_1 \boteq \reg_2 \tag{H3}
  \end{gather*}
  and the claim we need to show is:
  \begin{align*}
    \supdate {\reg_1}{\vx}{\sem \expr_{\reg_1}} &\steq[\bms] {{\supdate \env \vx \sigma}} \supdate {\reg_2}{\vx}{\sem \expr_{\reg_2}}.
  \end{align*}
  Before starting in the main part of the proof, observe:
  \begin{multline*}
    \forall \vy \in \Regs. \vy \neq \vx \Rightarrow
    \env (\vy) = {\supdate \env \vx \tty} (\vy) \land\\ \forall i \in \{1,2\}. \reg_i (\vy) = \supdate{\reg_i} \vx {\sem \expr_{\reg_i}} (\vy).
    \tag{H}
  \end{multline*}
  The proof goes by cases on $\bms$.
  \begin{proofcases}
    \proofcase{$\bms=\bot$} We have to establish:
    \begin{multline*}
      \forall \vy \in \Regs. \pi_1({\supdate \env \vx \sigma}(\vy))=\tylow \lor \pi_2({\supdate \env \vx \sigma}(\vy))=\tylow \Rightarrow\\ \supdate {\reg_1}{\vx}{\sem \expr_{\reg_1}}(\vy) =\supdate {\reg_2}{\vx}{\sem \expr_{\reg_2}}(\vy).
    \end{multline*}
    When $\vy\neq \vx$, the claim is a direct consequence of (H) and (H2), therefore we focus only on the case where
    $\vy\neq \vx$. In this case, the claim reduces to:
    \begin{equation*}
      \forall \vy \in \Regs. \pi_1(\sigma)=\tylow \lor \pi_2(\sigma)=\tylow \Rightarrow \sem \expr_{\reg_1} =\sem \expr_{\reg_2}.
    \end{equation*}
    The claim follows from an application of \Cref{lemma:tyexpr} (whose assumption are discharged by (H1), (H2) and (H3)).
    \proofcase{$\bms = \top$} This case it is completely analogous to the previous one
    except that, after assuming $\vx = \vy$, the claim reduces to:
    \begin{equation*}
      \forall \vy \in \Regs. \pi_2(\sigma)=\tylow \Rightarrow \sem \expr_{\reg_1} =\sem \expr_{\reg_2},
    \end{equation*}
    which again, follows from \Cref{lemma:tyexpr} (whose assumptions are discharged by (H1), (H2) and (H3)).
  \end{proofcases}
\end{proof}

The next lemma is about generalizing the \ref{TY:TSeq} rule to pairs of programs instead of pairs made by a program and an instruction. In principle, we could also formulate our program language in such a way to avoid the need of a similar result, by modeling programs as \emph{trees} of instructions instead of \emph{lists}. However, we chose to model programs as \emph{lists} of instructions to simplify the speculative semantics.

\begin{lemma}
  \label{lemma:tyseqgen}
  If $\pt_1 \deriv \sjudg {\env_1} \cmd {\env_2}$ and $\pt_2 \deriv \sjudg {\env_2} \cmdtwo {\env_3}$, then there is a tree $\pt_3$ such that  $\pt_3 \deriv \sjudg {\env_1} {\cmd\sep \cmdtwo} {\env_3}$.
\end{lemma}

\begin{proof}
  The proof is by induction on $\pt_1$.
  \begin{proofcases}
    \proofcase{\ref{TY:TEmpty}} The witness is $\pt_2$.
    \proofcase{\ref{TY:TWeak}} The tree $\pt_1$ can be rewritten as follows:
    \[
      \inferrev[\textsc{TWeak}]
      {\pt_1 \deriv \sjudg {\env_1} \cmd {\env_2}}
      {
        \env_1\le \env_1'
        \quad
        \pt_1' \deriv \sjudg {\env_1'} \cmd {\env_2'}
        \quad
        \env_2' \le \env_2
      }
    \]
    let $\pt_2'$ be the following tree:
    \[
      \inferrev[\textsc{TWeak}]
      {\sjudg {\env_2'} \cmdtwo {\env_3}}
      {
        \env_2' \le \env_2
        \quad
        \pt_2}
    \]
    We instantiate the  IH on $\pt_1'$, and apply it to $\pi_2'$ to conclude that there is  $\pt_3' \deriv \sjudg {\env_1'} {\cmd\sep \cmdtwo} {\env_3}$.
    With another weakening step, we conclude that there is a tree $\pt_3 \deriv \sjudg {\env_1} {\cmd\sep \cmdtwo} {\env_3} $.
    \proofcase{\ref{TY:TSeq}} The tree $\pt_1$ can be rewritten as follows:
    \[
      \inferrev[\textsc{TSeq}]
      {\pt_1\deriv \sjudg{\env_1}{\stat\sep \cmd'} {\env_3}}
      {
        \sjudg{\env_1}{\stat} {\env_2}
        \quad
        \pt_1'\deriv \sjudg{\env_2}{\cmd'} {\env_3}
      }
    \]
    From the IH on $\pi_1'$ applied to $\pt_2$,
    we deduce that there is a tree $\pt_3'$ such that
    $\pt_3'\deriv \sjudg{\env_2}{\cmd'\sep \cmdtwo} {\env_3}$. Now, we
    can apply \ref{TY:TSeq} to show:
    \[
      \inferrev[\textsc{TSeq}]
      {\pt_3\deriv \sjudg{\env_1}{\stat\sep (\cmd'\sep\cmdtwo)} {\env_3}}
      {
        \sjudg{\env_1}{\stat} {\env_2}
        \quad
        \pt_3'
      }
    \]
  \end{proofcases}
\end{proof}

\begin{lemma}
  \label{lemma:mistopaux}
  If $\mbuf_1 \steq[\bot]{\env} \mbuf_2$,  $\reg_1  \steq[\bot]{\env } \reg_2$, then $\mbuf_1 \steq[\top]{\env} \mbuf_2$,  $\reg_1 \steq[\top]{\env} \reg_2$.
\end{lemma}
\begin{proof}
  The proof involves three claims:
  \begin{align*}
    \forall \env, \reg_1, \reg_2.  \reg_1 \steq[\bot]{ \env} \reg_2 \Rightarrow \reg_1 \steq[\top]{ \env} \reg_2 \tag{C1}\\
    \forall \env, \envtwo, \mbuf_1, \mbuf_2.  \mbuf_1 \steq[\bot]{\env} \mbuf_2 \Rightarrow \mbuf_1 \steq[\top]{\env} \mbuf_2 \tag{C2}
  \end{align*}
  \begin{itemize}
  \item[(C1)] Fix all the universally quantified variables. Assume
    \begin{align*}
      \reg_1 \steq[\bot]{\env} \reg_2\tag{H1}
    \end{align*}
    We have to show that:
    \begin{align*}
      \forall x \in \Regs. \pi_2(\env (\vx)) = \tylow \Rightarrow \reg_1(\vx) = \reg_2(\vx)
    \end{align*}
    So, we assume $\pi_2(\env (\vx)) = \tylow$, and we observe that (H1) is equivalent to:
     \begin{align*}
      \forall x \in \Regs. \pi_1(\env (\vx)) = \tylow \lor \pi_2(\env (\vx)) = \tylow \Rightarrow \reg_1(\vx) = \reg_2(\vx),
     \end{align*}
     therefore the conclusion is a trivial consequence of (H1).
     \item[(C2)] By induction on $\mbuf_1$. The base case is trivial. In the inductive case, we assume that  $\mbuf_1=\bitem{(\ar, \nat)}{\val{}_1}\cons \mbuf_1'$, and we want to show that:
       \begin{align*}
         \mbuf_2&=\bitem{(\ar, \nat)}{\val{}_2}\cons \mbuf_2'\tag{D1}\\
         \mbuf_1' &\steq[\top]{\env} \mbuf_2'\tag{D2}
       \end{align*}
       Claim (D1)---i.e. that the first elements of $\mbuf_1$ and $\mbuf_2$ have the same address---is a trivial consequence of $\mbuf_1 \steq[\bot]{\env} \mbuf_2$. From the same assumption and the IH we also deduce (D2).
     \end{itemize}
\end{proof}

\subsubsection{Technical Lemmas for Memory Interaction}
\label{sec:mem}

\begin{lemma}
  \label{lemma:lookupaux}
  If $\mbuf_1 \steq[\bot]{\env} \mbuf_2$, $\mem_1 \steq{\env} \mem_2$,  $\val{}_i = \buflookup {\bm{\mbuf_i}{\mem_i}}{\ar,  \nat}$, and $\pi_1(\env (\ar) )= \tylow$, then $\val{}_1 = \val{}_2$.
\end{lemma}
\begin{proof}
  The proof goes by induction on $\mu_1$.
  \begin{proofcases}
    \proofcase{$\nil$} Direct consequence of $\mem_1 \steq{\env} \mem_2$.
    \proofcase{$\bitem{(\artwo, m)}{\val{}} \cons \mbuf_1'$} We go by cases on $(\artwo, m) = (\ar, \nat)$: when this is not the case, the claim is a direct consequence of the IH. When this is the case, the claim follows from $\mbuf_1 \steq[\bot]{\env} \mbuf_2$ and $\env (\ar)= \tylow$.
  \end{proofcases}
\end{proof}

\begin{lemma}
  \label{lemma:sloadaux}
  If $\reg_1 \steq[\bms]{\env} \reg_2$, then $\supdate{\reg_1}\vx {\val{}_1} \steq[\top]{{\supdate \env \vx {(\ty, \tyhigh)}}} \supdate{\reg_2}\vx {\val{}_2}$.
\end{lemma}
\begin{proof}
  The assumptions are
  \begin{align*}
    \reg_1 &\steq[\bms]{\env} \reg_2\tag{H1}
  \end{align*}
  The claim is:
  \begin{multline*}
    \forall \vy \in \Regs. \pi_2({\supdate \env \vx {(\ty, \tyhigh)}}(\vy)) = \tylow \Rightarrow\\ \supdate{\reg_1}\vx {\val{}_1 }(\vy) = \supdate{\reg_2}\vx {\val{}_2 }(\vy).
  \end{multline*}
  Fix $\vy$. The proof goes by cases on $\vy = \vx$. When this is not the case, the conclusion is a trivial consequence of (H1). When $\vy = \vx$ observe that $\pi_2({\supdate \env \vx {(\ty, \tyhigh)}} (\vy)) = \tyhigh$, therefore the claim holds.
\end{proof}

\begin{lemma}
  \label{lemma:sloadstepaux}
  If $\reg_1 \steq[\bms]{\env} \reg_2$, $\mbuf_1 \steq[\bms]{\env} \mbuf_2$, $\mem_1 \steq{\env} \mem_2$, $\val{}_i = \buflookup {\bm{\mbuf_i}{\mem_i}}{\ar,  \nat}$, and $\tty = (\env(\ar), \tyhigh)$ then $\supdate{\reg_1}\vx {\val{}_1} \steq[\bms]{{\supdate \env \vx \tty} } \supdate{\reg_2}\vx {\val{}_2}$.
\end{lemma}
\begin{proof}
  The assumptions are
  \begin{align*}
    \reg_1 &\steq[\bms]{\env} \reg_2\tag{H1}\\
    \mbuf_1 &\steq[\bms]{\env} \mbuf_2 \tag{H2}\\
    \mem_1 &\steq{\env} \mem_2\tag{H3}\\
    \val{}_i &= \buflookup {\bm{\mbuf_i}{\mem_i}}{\ar,  \nat}\tag{H4}\\
    \tty &= (\env(\ar), \tyhigh)\tag{H5}
  \end{align*}
  The proof is by cases on $\bms$:
  \begin{proofcases}
    \proofcase{$\bot$}  The claim is:
    \begin{multline*}
      \hspace{-1em}\forall \vy \in \Regs. \pi_1({\supdate \env \vx \sigma} (\vy)) = \tylow \lor \pi_2({\supdate \env \vx \sigma} (\vy)) = \tylow \Rightarrow\\
      \supdate{\reg_1}\vx {\val{}_1 }(\vy) = \supdate{\reg_2}\vx {\val{}_2}(\vy).
    \end{multline*}
    Fix $\vy$. The proof goes by cases on $\vy = \vx$. When this is
    not the case, the conclusion is a trivial consequence of
    (H1). When $\vy = \vx$, the claim is equivalent to:
    \[
      \pi_1(\tty) = \tylow \Rightarrow \val{}_1=\val{}_2,
    \]
    Assume, $\pi_1(\tty) = \tylow$. By (H5), we conclude
    $\env(\ar) = \tylow$, and the claim follows from this
    intermediate observation and \Cref{lemma:lookupaux}.
    \proofcase{$\top$} Direct consequence of \Cref{lemma:sloadaux}.
  \end{proofcases}
\end{proof}

\begin{lemma}
  \label{lemma:sstoreaux}
  If $\mbuf_1 \steq[\bms]{\env} \mbuf_2$, $\reg_1 \steq[\bms]{\env} \reg_2$ and $\sjudg \env \expr \tty $, then
  \[
    {\bitem{(\ar, \nat)} {\sem \expr_{\reg_1}}} \cons \mbuf_1 \steq[\bms] {\supdate \env \ar
              {\max(\pi_1(\tty), \env(\ar))}} {\bitem{(\ar, \nat)} {\sem \expr_{\reg_2}}}\cons \mbuf_2.
  \]
\end{lemma}
\begin{proof}
  Our assumptions are:
  \begin{align*}
    \mbuf_1 &\steq[\bms]{\env} \mbuf_2\tag{H1}\\
    \reg_1 &\steq[\bms]{\env} \reg_2\tag{H2}\\
    \sjudg {\env & } {\expr} \tty\tag{H3}
  \end{align*}
  We go by cases on $\bms$.
  \begin{proofcases}
    \proofcase{$\bot$} In this case, we have to establish
    \begin{gather*}
      {\supdate \env \ar
              {\max(\pi_1(\tty), \env(\ar))}}(\ar) = \tylow \Rightarrow \sem \expr_{\reg_1} = \sem \expr_{\reg_2} \tag{C1}\\
      \mbuf_1 \steq[\bot]{\envtwo} \mbuf_2\tag{C2}
    \end{gather*}
    \begin{itemize}
    \item[(C1)] Observe that the claim is equivalent to:
      \[
        {\max(\pi_1(\tty), \env(\ar))} = \tylow \Rightarrow \sem \expr_{\reg_1} = \sem \expr_{\reg_2}
      \]
      Assume ${\max(\pi_1(\tty), \env(\ar))} = \tylow$. From this assumption, we deduce $\pi_1(\tty) = \tylow$. We conclude with \Cref{lemma:tyexpr}, (H2), (H3), and (H4).
      \item[(C2)] Consequence of (H1), (H4) and \Cref{lemma:steqweak}.
    \end{itemize}
    \proofcase{$\top$} Trivial.
  \end{proofcases}
\end{proof}

\begin{lemma}
  \label{lemma:fenceflushaux}
  If $\mem_1 \steq{\env} \mem_2$, and
  ${\env}(\ar) =
  \tylow \Rightarrow \val{}_1 = \val{}_2$, then
  $\supdate{\mem_1}{(\ar, \nat)}{\val{}_1} \steq[\bms]{\env}
  \supdate{\mem_2}{(\ar, \nat)}{\val{}_2}$.
\end{lemma}
\begin{proof}
  Our assumptions are:
  \begin{gather*}
    \mem_1 \steq{\env} \mem_2\tag{H1}\\
    {\env}(\ar) = \tylow \Rightarrow \val{}_1 = \val{}_2 \tag{H2}\\
  \end{gather*}We have to prove that:
    \begin{equation*}
      \forall \artwo \in \Arr. {\env}(\artwo) = \tylow \Rightarrow
      \supdate{\mem_1}{(\ar, \nat)}{\val{}_1}(\artwo) = \supdate{\mem_2}{(\ar, \nat)}{\val{}_2}(\artwo)
    \end{equation*}
    Fix $\artwo$. The proof goes by cases on $\artwo = \ar$. When this is not the case, the conclusion is a direct consequence of (H1). When $\artwo=\ar$, observe that
    the claim boils down to proving
    \[
      \supdate{\mem_1(\ar)}{\nat}{\val{}_1} = \supdate{\mem_2(\ar)}{\nat}{\val{}_2}.
    \]
    under the assumption ${\env}(\ar) = \tylow$. From this
    assumption and (H1), we deduce that $\mem_1(\ar)=\mem_2(\ar)$.
    From (H2), we conclude that $\val{}_1=\val{}_2$.
\end{proof}
\begin{lemma}
  \label{lemma:fenceflush}
  If  $\mbuf_1 \steq[\bms]{\env} \mbuf_2$ and  $\mem_1 \steq{\env} \mem_2$, then $\overline {\bm{\mbuf_1}{\mem_1}}\steq{\env}\overline{\bm{\mbuf_2}{\mem_2}}$.
\end{lemma}
\begin{proof}
  The proof goes by induction on the length of the buffer.
  \begin{proofcases}
    \proofcase{$\nil$} The conclusion is a trivial consequence of the assumption on the stores.
    \proofcase{$\bitem{(\ar, \nat)}{\val{}_1}\cons \mbuf_1'$} From the definition of $\steq[\bms]{\env}$, we deduce that
    \begin{gather*}
      \mbuf_2= \bitem{(\ar, \nat)}{\val{}_2}\cons \mbuf_2'\tag{H}\\
      \mbuf_1'\steq[\bms]{\env} \mbuf_2'\tag{H1}\\
      \bms = \bot \Rightarrow
        {\env}(\ar) = \tylow \Rightarrow \val{}_1 = \val{}_2 \tag{H2}\\
\end{gather*}
      We have to establish:
      \[
\overline {\bm{\bitem{(\ar, \nat)}{\val{}_1}\cons \mbuf_1'}{\mem_1}} \steq{\env} \overline {\bm{\bitem{(\ar, \nat)}{\val{}_2}\cons \mbuf_2'}{\mem_2}},
      \]
      which means establishing that:
      \[
        \supdate{\overline {\bm{\mbuf_1'}{\mem_1}}}{(\ar, \nat)}{\val{}_1} \steq{\env} \supdate{\overline {\bm{\mbuf_2'}{\mem_2}}}{(\ar, \nat)}{\val{}_2}.
      \]
      From the IH---whose assumptions are discharged by (H1) and $\mem_1 \steq[\bms]{\env} \mem_2$---we deduce that
      \[
        {\overline {\bm{\mbuf_1'}{\mem_1}}} \steq{\env} {\overline {\bm{\mbuf_2'}{\mem_2}}}
      \]
      Therefore, the conclusion is a consequence of \Cref{lemma:fenceflushaux}.
  \end{proofcases}
\end{proof}
\subsubsection{Results on Expressions}
\label{sec:exprs}

For proving the next lemma, we extend the semantics of expressions to configurations in the standard way:
\[
  \sem \expr_{\sframe{\cmd, \reg}{\bm{\mbuf}{\mem}}{\bms}} \mydef \sem \expr_\reg
\]
\begin{lemma}
  \label{lemma:tyexpr}
  For all pairs of configurations
  $\confone_1 = \sframe{\cmd, \reg_1}{\bm{\mbuf_1}{\mem_1}}{\bms}
  $ and
  $\confone_2$, and expression $\expr$ such that
  $\sjudg \env \expr \tty$,  if
  $\confone_1 \confeq \env \confone_2$, then:
  \begin{itemize}
  \item[(C1)] if $b = \bot$ and $\pi_1(\tty) = \tylow\lor \pi_2(\tty) = \tylow$, then $\sem \expr_{\confone_1} = \sem \expr_{\confone_2}$.
  \item[(C2)] if $b = \top$ and $\pi_2(\tty) = \tylow$, then $\sem \expr_{\confone_1} = \sem \expr_{\confone_2}$.
  \end{itemize}
\end{lemma}
\begin{proof}
  The proof goes by induction on the type derivation of $\expr$
  \begin{proofcases}
    \proofcase{\ref{TY:TVal}} Trivial.

    \proofcase{\ref{TY:TVar}} In this case we have $\expr = \vx$, and we call $\confone_2 = \sframe{\cmd, \reg_2}{\bm{\mbuf_2}{\mem_2}}{\bms}$.
    \begin{itemize}
    \item[(C1)] Assume $b=\bot$ and
      $\pi_1(\tty) = \tylow\lor \pi_2(\tty) = \tylow$. We go by case analysis:
      \begin{proofcases}
        \proofcase{$\pi_1(\tty) = \tylow$} From
        $\confone_1 \confeq \env \confone_2$, we deduce that
        $\reg_1 \steq[\bot] {{\env}} \reg_2$. From this and
        $\pi_1(\tty) = \tylow$, we deduce that
        $\reg_1(\vx) = \reg_2(\vx)$, therefore $\sem \expr_{\confone_1} = \sem \expr_{\confone_2}$.
        \proofcase{$\pi_2(\tty) = \tylow$} Analogous to the previous case.
      \end{proofcases}

    \item[(C2)] We assume $b=\top$ and
      $\pi_2(\tty) = \tylow$, from
      $\confone_1 \confeq \env \confone_2$, we deduce that
      $\reg_1 \steq[\top] {{\env}} \reg_2$. From this and
      $\pi_2(\tty) = \tylow$, we deduce that
      $\reg_1(\vx) = \reg_2(\vx)$, therefore $\sem \expr_{\confone_1} = \sem \expr_{\confone_2}$.
\end{itemize}

\proofcase{\ref{TY:TOp}}

\begin{itemize}
\item[(C1)] Here $\expr = \op(\expr_1, \ldots, \expr_k)$. We assume $b=\bot$ and
  $\pi_1(\tty) = \tylow \lor \pi_2(\tty) = \tylow$. By introspection
  of the typing rule, we deduce that $\forall 1\le i \le k. \sjudg \env {\expr_i} \sigma$.
  From the IH, we conclude that for every
  $i$, we have
  $\sem {\expr_i}_{\confone_1} = \sem {\expr_i}_{\confone_2}$. The
  conclusion is a consequence of the definition of $\sem \expr$.
\item[(C2)]  Analogous to the previous case.
\end{itemize}
\proofcase{\ref{TY:TSub}} The claim is a direct consequence of the IH.
\end{proofcases}
\end{proof}

\begin{corollary}
  \label{lemma:idexpr}
  For all pairs of configurations $\confone_1, \confone_2$, and expression $\expr$ such that
  $\sjudg \env \expr \typair \tylow$, if $\confone_1 \confeq \env  \confone_2$ then $\sem \expr_{\confone_1} = \sem \expr_{\confone_2}$.
\end{corollary}

\section{Jump target analysis}
\label{sec:jta}

This appendix is devoted to showing that if every $\kwd{jmp}\ \vx$
instruction is replaced with the sequence of instructions
$\dfencex\sep\kwd{jmp}\ \vx$, then the targets that
jump instructions can reach with the speculative semantics is included
in the set of targets that they can reach under normal semantics.

\begin{figure*}[t]
  \centering
  \columnwidth=\linewidth
  \small
    \begin{gather*}
      \Infer[NS][Asgn][\textsc{Asgn}]
      {\step {\vx \asgn{\expr}\sep\cmd, \st}
              {\cmd, \supdate{\st}{\vx}{\val{}}}
            }
              {\sem\expr_{\st} = \val{}}
      \qquad
      \Infer[NS][Load][\textsc{Load}]
      {\step
        {\cmemread{\vx}{\ar[\expr]}\sep\cmd, \st}
        {\cmd, \supdate{\st}{\vx}{\val{}}}
      }
      {
        \sem\expr_{\st} = \nat \quad
        \nat \in [0,\size \ar) \quad
        \st({\ar, \nat})  = \val{}
      }
      \qquad
      \Infer[NS][Store][\textsc{Store}]
      {\step
        {\cmemasgn {\ar[\expr]} \exprtwo\sep\cmd, \st}
        {\cmd, \supdate \st {(\ar, \nat)}{\val{}} }
      }
      {
        \sem\expr_{\st} = \nat \quad
        \nat \in [0,\size \ar) \quad
        \sem\exprtwo_{\st} = \val{}
      }
      \\
    \Infer[NS][While]{
      \step{W \sep\cmdtwo, \st, \bms}
            {\cmdtwo_{\sem \expr_\st}\sep\cmdtwo, \st}
    }
    {
      W = \cwhile \expr {\cmd} \quad
      \cmdtwo_\ctrue = \cmd \sep W \sep\cmdtwo\quad
      \cmdtwo_\cfalse = \cmdtwo
    }
    \qquad
\Infer[NS][If]{
      \step{I\sep\cmdtwo, \st}
            {\cmdtwo_{\sem\expr_\st}\sep\cmdtwo, \st}
    }
    {
      I = \cif \expr {\cmdtwo_\ctrue} {\cmdtwo_\cfalse}
    }
    \\
    \Infer[NS][Jmp]{
      \step{\kwd{jmp}\ \expr, \st}
            {\phi(\sem \expr_\st), \st}
    }
    {
    }
    \quad
    \Infer[NS][Fence]{
      \step{\fence\sep\cmd, \st}
            {\cmd, \st, \bms}
    }{~}
    \quad
    \Infer[NS][Dfence]{
      \step{\dfencex\sep\cmd, \st}
            {\cmd, \st, \bms}
    }
    {
    }
    \end{gather*}
    \caption{Non-speculative semantics.}
  \label{fig:sem}
\end{figure*}

To prove this observation, we start by defining the non-speculative
semantics of our programs in \Cref{fig:sem}. Here, a store $\st$ is
simply a pair $\reg, \mem$ given by a register map and a memory and we
use the notation $\supdate{\st}{\vx}{\val{}}$ to represent the store
$\supdate{\reg}{\vx}{\val{}}, \mem$ and the notation
$\supdate{\st}{(\ar, \nat)}{\val{}}$ to represent the store
$\reg, \supdate{\mem}{(\ar, \nat)}{\val{}}$. Similarly, given a
speculative store $\sttwo = \reg, \mbuf, \mem, \bms$, we denote with
$\overline \sttwo$ the store $\reg, \overline {\bm \mbuf \mem}$. In
this section, we assume that jump instructions are annotated with a
set of jump targets $V$. We use this annotation to define an
\emph{instrumented} semantics which catches jumps outside $V$.
This instrumented semantics is defined identically to the speculative one, except for
\ref{SIE:Jmp}, which is now replaced by the two following rules:
\[
  \Infer[SIE][JmpInV]{
    \sstep{\jmp\ \expr_V, \st, \bms}
    {\phi(\sem \expr_\st), \st, \bms}
    {\dstep} {\obranch{\sem \expr_\st}}
  }
  {
    \sem\expr_\st \in V
  }
\]
\[
  \Infer[SIE][JmpErr]{
    \sstep{\jmp\ \expr_V, \st, \bms}
    {\err}
    {\dstep} {\obranch{\sem \expr_\st}}
  }
  {
    \sem\expr_\st \in \Loc \quad
    \sem\expr_\st \notin V
  }
\]

Thanks to this instrumented semantics, we can express the correctness
of the annotations within $V$ by requiring the absence of reduction to
$\err$.

\newcommand{\correct}[1]{\mathcal{C}(#1)}
\newcommand{\prot}[1]{\mathcal{D}(#1)}

The main assumption of this section is that the jump target analysis
performed for the normal semantics of our program is correct. This
property can be stated as follows:
\begin{multline*}
    \correct{\cmd, \jmpcont} \mydef \forall\nat\in \Nat. \forall \reg, \mem. \step{\cmd,\reg, \mem}{{\kwd{jmp}\ \expr_V\sep\cmd',\reg', \mem'}} \\
    \Rightarrow \sem \expr_{\reg'} \in V.
\end{multline*}

Finally, in this section, we are interested only in a specific class
of programs, namely only those jumps whose target is held in a
register which is protected by a \dfence instruction. We express this
property via the predicate $\prot\cmd$, which is defined as follows:
\begin{gather*}
  \inferrev {\prot{\vx \asgn \expr}}{~}\quad
  \inferrev{\prot {\cmemread \vx {\ar[ \expr]}}}{~}\quad
  \inferrev {\prot{\cmemasgn {\ar[ \expr]} \exprtwo}}{~}\\
  \inferrev {\prot{\cif \expr {\cmd_1}{\cmd_2}}}{\prot{\cmd_1} \quad \prot{\cmd_2}}\quad
  \inferrev{\prot {\cwhile \expr \cmd}}{\prot{\cmd}}\\
  \Infer[PR][Dfence] {\prot{\dfencex}}{~}\quad
  \inferrev {\prot{\fence}}{~}\\
  \inferrev {\prot{\stat\sep \cmd}}{\prot{\stat}\quad\prot{\cmd}}\quad
  \Infer[PR][Jump] {\prot{\dfencex\sep \kwd{jmp} \vx_V \sep \cmd}}{\prot{\cmd}}  \quad
  \inferrev {\prot{\cnil}}{~}
\end{gather*}

The predicate $\prot\cdot$ is defined on jump continuations $\jmpcont$
by taking the pointwise extension of the same predicate on programs.

It is now time to state the claim of the main result of this
section:

\begin{theorem*}{thm:jtcorrect}
  For every program $\cmd$ and jump continuation $\jmpcont$ such that
  $\correct{\cmd, \jmpcont}$, $\prot \cmd$ and $\prot \jmpcont$, state
  $\st$, $\nat \in \Nat$, directives $\Ds$ and observations $\Os$, we
  have:
  \[
    \lnot\left (\sstep[\nat] {\cmd, \st} {\err} \Ds \Os\right)
  \]
\end{theorem*}

The proof of \Cref{thm:jtcorrect} is essentially by case analysis on
the value of the misspeculation flag in the configuration immediately
after reaching the $\err$ state. If it is $\bot$, then that reduction
sequence can be simulated with the normal semantics and because of
$\correct{\cmd, \jmpcont}$, we deduce that the jump target is safe. We
prove this result in \Cref{lemma:sim} below. If the misspeculation is
$\bot$, by leveraging that the jump target is protected by \dfence,
we deduce that the transition to $\err$ cannot take place because the
jump target is not available during transient execution. We prove this
result in \Cref{lemma:prot}.

Before going to the proof of \Cref{thm:jtcorrect}, we prove the
above-mentioned auxiliary results, we start with the simulation result
between normal and the fragment of our speculative semantics where
$\bms=\top$

\begin{lemma}
  \label{lemma:sim}
  Fix a program $\cmd$, an initial state
  $\st = \reg, \mbuf, \mem, \bms$, a second program $\cmdtwo$ and a target
  state $\sttwo = \reg', \mbuf', \mem', \bot$.  For every $\nat \in \Nat$, directives $\Ds$ and
  observations $\Os$, if
  \[
    \sstep[\nat] {\cmd, \st} {\cmdtwo, \sttwo} \Ds \Os,
  \]
  then:
  \[
    \step[\nat] {\cmd, \overline \st} {\cmdtwo, \overline \sttwo}.
  \]
\end{lemma}
\begin{proof}
  The proof is by induction on $\nat$. The base case holds by vacuity of the premise. Therefore, we only focus on the inductive case, where we assume
  \[
    \sstep[\nat] {\cmd, \st} {\cmdtwo', \sttwo'} \Ds \Os \xrightarrow[\dir]\obs {\cmdtwo, \sttwo},
  \]
  and we have to prove
  \[
    \step[\nat+1] {\cmd, \overline \st} {\cmdtwo, \overline \sttwo},
  \]
  The proof is by case analysis on the rule that has been employed for the last transition.
  \begin{proofcases}
    \proofcase{\ref{SIE:SAsgn}} In this case, by introspection of the semantics, it must be the case where $\cmdtwo' = \vx \asgn \expr;\cmdtwo$,  $\sttwo' = \reg'', \mbuf'', \mem'', \bot$ and $\reg' = \supdate{\reg''}{\vx}{\sem \expr_{\reg''}}$. From the induction hypothesis, we deduce:
  \[
    \step[\nat] {\cmd, \overline \st} {\cmdtwo', \reg'',\overline{\bm {\mbuf''}{\mem''}}}.
  \]
  We conclude the proof by observing:
  \[
    \step {\cmdtwo', \reg'',\overline{\bm {\mbuf''}{\mem''}}} {\cmdtwo, \overline \sttwo}
  \]
  via rule \ref{NS:Asgn}.

  \proofcase{\ref{SIE:SLoad-Step}} In this
  case, by introspection of the semantics, it must be the case where
  $\cmdtwo' = \cmemread \vx \ar[\expr];\cmdtwo$,
  $\sttwo' = \reg'', \mbuf', \mem', \bot$ and
  $\reg' = \supdate{\reg''}{\vx} {\buflookup {\bm{\mbuf'}{\mem'}} {\ar,
      \sem {\expr}_{\reg''}}}$. From the induction hypothesis, we deduce:
  \[
    \step[\nat] {\cmd, \overline \st} {\cmdtwo', \reg'',\overline{\bm {\mbuf'}{\mem'}}}.
  \]
  Observe that:
  \small
  \[
    \step {\cmdtwo', \reg'',\overline{\bm {\mbuf'}{\mem'}}} {\cmdtwo, \supdate{\reg''}\vx {\overline{\bm {\mbuf''}{\mem''}}(\ar, \sem {\expr}_{\reg''})},\overline{\bm {\mbuf'}{\mem'}}}
  \]
  \normalsize
  via rule \ref{NS:Load}
  Therefore, to conclude the proof, we need to observe that
  \[
    {\overline{\bm {\mbuf'}{\mem'}}(\ar, \sem {\expr}_{\reg''})} =
    {\buflookup {\bm{\mbuf'}{\mem'}} {\ar,
        \sem {\expr}_{\reg''}}},
  \]
  which is established in \cref{lemma:bufread} below.
  \proofcase{\ref{SIE:SStore-Step}} By introspection of rule, it must
  be the case where $\cmdtwo' = \cmemasgn {\ar[\expr]}\exprtwo;\cmdtwo$,
  $\sttwo' = \reg', \mbuf'', \mem', \bot$ and
  \[
    \mbuf' = \bitem {(\ar, \sem \expr_{\reg'})} {\sem\exprtwo_{\reg'}}\cons \mbuf''.
  \]
  From the induction hypothesis,
  we deduce:
  \[
    \step[\nat] {\cmd, \overline \st} {\cmdtwo', \reg',\overline{\bm {\mbuf''}{\mem'}}}.
  \]
  Observe that:
  \[
    \step {\cmdtwo', \reg',\overline{\bm {\mbuf''}{\mem'}}} {\cmdtwo, \reg',\supdate{\overline{\bm {\mbuf''}{\mem'}}}{(\ar, \sem\expr_{\reg'})} {\sem\exprtwo_{\reg'}}}
  \]
  via rule \ref{NS:Store}.
  Therefore, to conclude the proof, we need to observe that
    \begin{align*}
    \supdate{\overline{\bm {\mbuf''}{\mem'}}}{(\ar, \sem\expr_{\reg'})} {\sem\exprtwo_{\reg'}}
      &= \overline {\bm{\mbuf'}{\mem'}} \\
      &= \overline{\bm{\bitem {(\ar, \sem \expr_{\reg'})} {\sem\exprtwo_{\reg'}}\cons \mbuf''}{\mem'}},
    \end{align*}
    which is a direct consequence of the definition of
    $\overline {\bm\cdot\cdot}$.  \proofcase{\ref{SIE:SLoad-OOB},
      \ref{SIE:SStore-OOB}, \ref{SIE:SLoad-PSF}, \ref{SIE:SLoad-SSB}}
    These cases are absurd because the mis-speculation flag of the
    target configuration is $\top$.

    \proofcase{\ref{SIE:SWhile}} By introspection of rule, it must be
    the case where $\cmdtwo' = \cwhile \expr{\cmd'};\cmd''$,
    $\sttwo' = \reg', \mbuf', \mem', \bot$.  In particular we deduce that the misspeculation
    flag is $\bot$ by introspection of the transition rule.
    From the induction
    hypothesis, we deduce:
    \[
      \step[\nat] {\cmd, \overline \st} {\cmdtwo', \reg',\overline{\bm {\mbuf'}{\mem'}}}.
    \]
    Note that since the target
    configuration carries the mis-speculation flag $\bot$,
    it must be the case where the  $\dir = \dbranch{\sem \expr_{\reg'}}$.
    The proof goes by case analysis on the value of $\sem \expr_{\reg'}$.
    \begin{proofcases}
      \proofcase{$\top$} In this case, we observe that
      $\cmdtwo = \cmd''$, and we observe that
      \[
        \step {\cmdtwo', \reg',\overline{\bm {\mbuf'}{\mem'}}}
        {\cmd'', \reg', \overline{\bm {\mbuf'}{\mem'}}}
      \]
      via rule \ref{NS:While}
      \proofcase{$\top$} In this case we have
      $\cmdtwo = \cmd' \sep\cwhile \expr{\cmd'};\cmd''$, and the
      transition
      \[
        \step {\cmdtwo', \reg',\overline{\bm {\mbuf'}{\mem'}}}
        {\cmd' \sep\cwhile \expr{\cmd'};\cmd'', \reg',
          \overline{\bm {\mbuf'}{\mem'}}}
      \]
      is established via rule \ref{NS:While}
    \end{proofcases}
    \proofcase{\ref{SIE:SDfence}} By introspection of the semantics,we observe that $\cmdtwo' = \dfence \vx;\cmdtwo$,  $\sttwo' = \reg', \mbuf', \mem', \bot$. In particular, the register files are identical because the misspeculation flag of $\sttwo'$ is $\bot$. From the inductive hypothesis, we deduce
  \[
    \step[\nat] {\cmd, \overline \st} {\cmdtwo', \reg',\overline{\bm {\mbuf'}{\mem''}}}.
  \]
  We conclude the proof by observing:
  \[
    \step {\cmdtwo', \reg',\overline{\bm {\mbuf'}{\mem'}}} {\cmdtwo, \overline\sttwo}
  \]
  via rule \ref{NS:Dfence}.
    \proofcase{\ref{SIE:SFence}} In this case, by introspection of the semantics, we deduce that $\cmdtwo' = \fence;\cmdtwo$,  $\sttwo' = \reg', \nil, \overline {\bm\nil {\mem''}}, \bot$ for some $\mem''$. From the induction hypothesis, we deduce:
  \[
    \step[\nat] {\cmd, \overline \st} {\cmdtwo', \reg',\overline {\bm \nil{\overline {\bm\nil {\mem''}}}}}.
  \]
  We then observe:
  \[
    \step {\cmdtwo', \reg',\overline{\bm {\nil}{\mem''}}} {\cmdtwo, \reg',{\overline {\bm\nil {\mem''}}}}
  \]
  via rule \ref{NS:Fence}. To conclude the proof it suffices to observe
  \[
    \overline {\bm \nil{\overline {\bm\nil {\mem''}}}} = \overline {\bm\nil {\mem''}},
  \]
  which follows by definition of $\overline {\bm \cdot \cdot}$.\qedhere
 \end{proofcases}
\end{proof}

We now turn our attention to proving that when a program $\cmd$ such
that $\prot \cmd$ executes and reaches a $\kwd{jmp}\ \expr_V$ instruction,
then $\expr$ is a register name and the preceding configuration must
hold the program $\dfence\vx; \kwd{jmp}\ \vx_V$. In order to prove
such result, we first need an auxiliary lemma:

\begin{lemma}
  \label{lemma:protaux1}
  Let $\cmd$ be a program such that $\prot \cmd$ and let $\jmpcont$ be a
  jump continuation such that $\prot \jmpcont$. For every state $\st$, $\nat \in \Nat$,
  directives $\Ds$, observations $\Os$, and target configuration
  $(\cmdtwo, \sttwo)$ such that
  \[
    \sstep[\nat] {\cmd, \st} {\cmdtwo, \sttwo}\Ds \Os,
  \]
  if $\lnot\prot\cmdtwo$, then
  \[
    \sstep[\nat+1] {\cmd, \st} {\cmdtwo', \sttwo'}{\dir:\Ds'} {\obs:\Os}
  \]
  for some $\cmdtwo', \dir, \obs$ such that $\prot{\cmdtwo'}$.
\end{lemma}
\begin{proof}
  The proof is by induction on $\nat$. The base case is trivial,
  therefore we go to the inductive case. The proof proceeds by case analysis on $\cmd$.
  \begin{proofcases}
    \proofcase{$\cnil$} Absurd.
    \proofcase{$\kwd{jmp} \expr_V$} Absurd.
    \proofcase{$\vx \asgn \expr\sep \cmd'$,
      $\cmemread \vx {\ar[\expr]}\sep \cmd'$,
      $\cmemasgn{\ar[\expr]}\exprtwo\sep \cmd'$,
    $\fence\sep \cmd'$} Observe that
    independently of the rule that was applied to show the transition,
    the program of target configuration is $\cmd'$ which satisfies
    $\prot {\cmd'}$ by definition of $\prot$. Therefore, the
    conclusion is a consequence of the inductive hypothesis.
    \proofcase{$\cwhile \expr{\cmd''} \sep \cmd'$} The directive used
    for the reduction is $\dbranch \bool$ and the proof is by cases on
    the value of $\bool$. If it is $\bot$, as in the previous cases,
    the conclusion is a direct consequence of the inductive
    hypothesis. Otherwise, we have:
    \begin{multline*}
      \sstep {\cwhile \expr{\cmd''} \sep \cmd', \st} {}{\dbranch \bool}{\obranch {\sem \expr_{\st}}}\\
      {\cmd''\sep \cwhile \expr{\cmd''} \sep \cmd', \st}.
    \end{multline*}
    Since $\prot {\cwhile \expr{\cmd''}\sep \cmd'}$ holds, by
    definition of $\prot \cdot$, also $\prot{\cmd''}$ must hold,
    therefore $\prot{\cmd''\sep \cwhile \expr{\cmd''} \sep \cmd'}$
    must also hold.  The conclusion is thus a direct consequence of
    the inductive hypothesis.
    \proofcase{$\cif \expr {\cmd_1}{\cmd_2}\sep \cmd'$} This case is analogous to the second part of the previous one.
    \proofcase{$\dfencex\sep \cmd'$} The proof is by case analysis on how the predicate $\prot{\dfencex \sep \cmd'}$ was established. If it was established by using rule \ref{PR:Dfence}, then the proof is analogous to those for all other instructions that do not contain programs as sub-terms. Otherwise,
    the rule must be \ref{PR:Jump}. In this case, we also deduce that $\cmd = \dfencex\sep \kwd{jmp}\ \expr_V\sep\cmd'$ for some $\cmd'$. If $\nat = 1$ we have:
    \[
      \sstep[1] {\dfencex\sep \kwd{jmp}\ \expr_V\sep\cmd', \st} {\kwd{jmp}\ \expr_V\sep\cmd'} \dstep \onone.
    \]
    By introspection of the definition of the predicate $\prot \cdot$, we deduce  $\lnot{\prot{\kwd{jmp}\ \expr_V\sep\cmd'}}$. However, the claim holds because we have $\prot \cmd$. Finally, if $\nat>1$, we have:
    \[
      \sstep[2] {\dfencex\sep \kwd{jmp}\ \expr_V\sep\cmd', \st} {\jmpcont{\sem\expr}\sep\cmd'} {\dstep:\dstep} {\onone:\onone}.
    \]
    The conclusion is a consequence of $\prot{\jmpcont}$.
    \qedhere
  \end{proofcases}
\end{proof}

\begin{corollary}
    \label{cor:protaux2}
  Let $\cmd$ be a program such that $\prot \cmd$ and let $\jmpcont$ be a
  jump continuation such that $\prot \jmpcont$. For every state $\st$, $\nat \in \Nat$,
  directives $\Ds$, observations $\Os$, and target configuration
  $(\cmdtwo, \sttwo)$ such that
  \[
    \sstep[\nat] {\cmd, \st} {\cmdtwo, \sttwo}\Ds \Os,
  \]
  if $\lnot\prot\cmdtwo$, then $\nat>0$ and
  \[
    \sstep[\nat-1] {\cmd, \st} {\cmdtwo', \sttwo'}{\Ds'} {\Os'}
  \]
  for some $\cmdtwo'$ such that $\prot{\cmdtwo'}$, and $\Ds', \Os'$
  suffixes of $\Ds$ and $\Os$, respectively.
\end{corollary}

\begin{corollary}
  \label{lemma:prot}
  Let $\cmd$ be a program such that $\prot \cmd$ and let $\jmpcont$ be a
  jump continuation such that $\prot \jmpcont$. For every state $\st$, $\nat \in \Nat$,
  directives $\Ds$, observations $\Os$, and target configuration
  $(\kwd{jmp}\ \expr_V\sep\cmdtwo, \sttwo)$ such that
  \[
    \sstep[\nat] {\cmd, \st} {\kwd{jmp}\ \expr_V\sep \cmdtwo, \sttwo}\Ds \Os,
  \]
  then $\nat>0$, $\expr = \vx$ and
  \[
    \sstep[\nat-1] {\cmd, \st} {\dfencex\sep \kwd{jmp}\ \expr_V \sep \cmdtwo, \sttwo'}{\Ds'} {\Os'}
  \]
  for some $\Ds', \Os'$ suffixes of $\Ds$ and $\Os$, respectively.
\end{corollary}
\begin{proof}
  The conclusion is a consequence of \Cref{cor:protaux2}. Assume
  \[
    \sstep[\nat] {\cmd, \st} {\kwd{jmp}\ \expr_V \sep \cmdtwo, \sttwo}\Ds \Os.
  \]
  \Cref{cor:protaux2} proves
  \[
    \sstep[\nat-1] {\cmd, \st} {\cmdtwo', \sttwo'}{\Ds'} {\Os'},
  \]
  with $\prot{\cmdtwo'}$. By introspection on the definition of $\prot{\cdot}$,
  we deduce that the only rule that can be used to establish $\prot{\cmdtwo'}$ is
  Rule~\ref{PR:Jump}. Indeed, with all the other rules, we could not have:
  \[
    \sstep {\cmdtwo', \sttwo'} {\kwd{jmp}\ \expr_V \sep \cmdtwo, \sttwo}\dir \obs.
  \]
  By introspection of Rule~\ref{PR:Jump}, we conclude that $\expr$ is a
  register name and that
  $\cmdtwo' = \dfencex \sep \kwd{jmp}\ \expr_V \sep \cmdtwo$.
\end{proof}

It is now time to prove the main result of this section, i.e. \Cref{thm:jtcorrect}:

\again{thm:jtcorrect}
\begin{proof}
  Assume
  \[
    \sstep[\nat] {\cmd, \st} {\err} \Ds \Os.
  \]
  Our goal is to reach a contradiction.  The proof starts with a case
  analysis on $\nat$. If $\nat = 0$ we have $(\cmd, \st) = {\err}$,
  which concludes the proof.  In the following, we assume without loss
  of generality that $\nat >0$.  By introspection of the semantics, we
  observe that the only rule that can be used to reach $\err$ must be
  \ref{SIE:JmpErr}, which entails that
  \[
    \sstep[\nat-1] {\cmd, \st} {\kwd{jmp}\ \expr_V\sep \cmd', \sttwo} {\Ds'} {\Os'} \xrightarrow[\dstep]\onone \err
  \]
  for some $\Ds', \Os', \cmd', \sttwo$ with $\Ds = \dstep:\Ds'$, $\Os = \onone:\Os'$ and
  \[
    \tag{H}
    \sem \expr_\sttwo \notin V.
  \]
  The proof proceeds by case analysis on
  $\sttwo = \reg, \mbuf, \mem, \bms$. In particular, the case
  analysis is on $\bms$.
  \begin{proofcases}
    \proofcase{$\bms = \bot$} In this case, an application of \Cref{lemma:sim} proves
    \[
      \step[\nat-1] {\cmd, \st} {\kwd{jmp}\ \expr_V\sep \cmd', \overline \sttwo},
    \]
    and from $\correct{\cmd, \jmpcont}$, we deduce that $\sem \expr_\sttwo \in V$, which contradicts (H).
    \proofcase{$\bms = \top$} In this case, an application of \Cref{lemma:prot} proves that $\expr$ is a register name $\vx$ and
    \begin{multline*}
      \sstep[\nat-2] {\cmd, \st} {\dfencex\sep \kwd{jmp}\ \expr_V\sep \cmd', \sttwo} {\Ds''} {\Os''} \xrightarrow[\dstep]\onone\\ {\kwd{jmp}\ \expr_V\sep \cmd', \supdate \sttwo \vx {\bot}} \xrightarrow[\dstep]\onone \err.
    \end{multline*}
    But this is absurd because when the jump target is $\bot$, Rule \ref{SIE:JmpErr} does not apply.
    \qedhere
  \end{proofcases}

\end{proof}

\subsection{Technical observations}
\label{sec:tobs2}

\begin{lemma}
  \label{lemma:bufread}
  ${\overline{\bm {\mbuf}{\mem}}(\ar,\nat)} = {\buflookup
    {\bm{\mbuf}{\mem}} {\ar, \nat}}$.
\end{lemma}
\begin{proof}
  By induction on $\nat$. The base case is trivial. in the inductive case we have to prove
  \[
    {\overline{\bm {\bitem {(\artwo, t)} {\val{}}\cons \mbuf}{\mem}}(\ar,\nat)} = {\buflookup
      {\bm{\bitem {(\artwo, t)} {\val{}}\cons \mbuf}{\mem}} {\ar, \nat}},
  \]
  i.e.:
  \[
    \supdate {\overline{\bm {\mbuf}{\mem}}} {(\artwo, t)} {\val{}} (\ar,\nat)  = {\buflookup
      {\bm{\bitem {(\artwo, t)} {\val{}}\cons \mbuf}{\mem}} {\ar, \nat}}.
  \]
  The proof proceeds by case analysis.
  \begin{proofcases}
    \proofcase{$\ar= \artwo \land \nat=t$} The conclusion is trivial.
    \proofcase{$\ar\neq \artwo \lor \nat\neq t$} The claim rewrites as:
      \[
        {\overline{\bm {\mbuf}{\mem}}} (\ar,\nat) = {\buflookup
          {\bm{\mbuf}{\mem}} {\ar, \nat}},
      \]
      which is a direct consequence of the induction hypothesis.
      \qedhere
  \end{proofcases}
\end{proof}

 \section{Benchmark results}
\label{app:benchs}

The complete results of our benchmarks are shown in \Cref{tab:benchs}. For ML-DSA, Poly1305 and ChaCha20, 4096 byte long messages were used.
\begin{table*}[p]
  \caption{Cycle measurement and relative overhead of different protection mechanisms.}
  \centering
  \resizebox{\textwidth}{!}{
    \begin{tabular}{llllllllllll}
      \toprule
      \textbf{Benchmark} & \textbf{Procedure} & \textbf{Baseline} & \dfence & \dfence w/o v1.1 & \textbf{Unopt.} \dfence & \textbf{Unopt.} \dfence w/o v1.1 & \fence & \fence w/o v1.1 & \textbf{SLH and SSBD} & \textbf{SLH only} & \textbf{SSBD only} \\
      \midrule
      \multirow{3}{*}{ML-DSA 44} & Key generation & $\makecell[l]{11375269}$ & $\makecell[l]{11449031\\ +0.65\%}$ & $\makecell[l]{11413206\\ +0.33\%}$ & $\makecell[l]{11485901\\ +0.97\%}$ & $\makecell[l]{11463826\\ +0.78\%}$ & $\makecell[l]{11664281\\ +2.54\%}$ & $\makecell[l]{11743824\\ +3.24\%}$ & $\makecell[l]{12306670\\ +8.19\%}$ & $\makecell[l]{11615986\\ +2.12\%}$ & $\makecell[l]{12189527\\ +7.16\%}$ \\
       & Signature & $\makecell[l]{24585908}$ & $\makecell[l]{24512106\\ -0.30\%}$ & $\makecell[l]{24433830\\ -0.62\%}$ & $\makecell[l]{24510038\\ -0.31\%}$ & $\makecell[l]{24502489\\ -0.34\%}$ & $\makecell[l]{24954618\\ +1.50\%}$ & $\makecell[l]{25022288\\ +1.77\%}$ & $\makecell[l]{28106708\\ +14.32\%}$ & $\makecell[l]{25028946\\ +1.80\%}$ & $\makecell[l]{28117739\\ +14.37\%}$ \\
       & Verify & $\makecell[l]{14747774}$ & $\makecell[l]{14897803\\ +1.02\%}$ & $\makecell[l]{14884719\\ +0.93\%}$ & $\makecell[l]{14980606\\ +1.58\%}$ & $\makecell[l]{14893592\\ +0.99\%}$ & $\makecell[l]{15041659\\ +1.99\%}$ & $\makecell[l]{15106930\\ +2.44\%}$ & $\makecell[l]{16096629\\ +9.15\%}$ & $\makecell[l]{15013664\\ +1.80\%}$ & $\makecell[l]{16106257\\ +9.21\%}$ \\
      \midrule
      \multirow{3}{*}{ML-DSA 65} & Key generation & $\makecell[l]{20357347}$ & $\makecell[l]{20013036\\ -1.69\%}$ & $\makecell[l]{20374070\\ +0.08\%}$ & $\makecell[l]{20141583\\ -1.06\%}$ & $\makecell[l]{20267261\\ -0.44\%}$ & $\makecell[l]{20855872\\ +2.45\%}$ & $\makecell[l]{20747201\\ +1.92\%}$ & $\makecell[l]{21632356\\ +6.26\%}$ & $\makecell[l]{20386275\\ +0.14\%}$ & $\makecell[l]{21545221\\ +5.84\%}$ \\
       & Signature & $\makecell[l]{79783999}$ & $\makecell[l]{79426487\\ -0.45\%}$ & $\makecell[l]{79815608\\ +0.04\%}$ & $\makecell[l]{79673289\\ -0.14\%}$ & $\makecell[l]{79563600\\ -0.28\%}$ & $\makecell[l]{81173173\\ +1.74\%}$ & $\makecell[l]{80949285\\ +1.46\%}$ & $\makecell[l]{95829429\\ +20.11\%}$ & $\makecell[l]{81474361\\ +2.12\%}$ & $\makecell[l]{96136127\\ +20.50\%}$ \\
       & Verify & $\makecell[l]{22732327}$ & $\makecell[l]{23121187\\ +1.71\%}$ & $\makecell[l]{23146200\\ +1.82\%}$ & $\makecell[l]{23069134\\ +1.48\%}$ & $\makecell[l]{23157852\\ +1.87\%}$ & $\makecell[l]{23526064\\ +3.49\%}$ & $\makecell[l]{23391689\\ +2.90\%}$ & $\makecell[l]{24932088\\ +9.68\%}$ & $\makecell[l]{23526201\\ +3.49\%}$ & $\makecell[l]{25048195\\ +10.19\%}$ \\
      \midrule
      \multirow{3}{*}{ML-DSA 87} & Key generation & $\makecell[l]{34602499}$ & $\makecell[l]{34352163\\ -0.72\%}$ & $\makecell[l]{34636280\\ +0.10\%}$ & $\makecell[l]{34502739\\ -0.29\%}$ & $\makecell[l]{34471143\\ -0.38\%}$ & $\makecell[l]{35324651\\ +2.09\%}$ & $\makecell[l]{35691344\\ +3.15\%}$ & $\makecell[l]{36547154\\ +5.62\%}$ & $\makecell[l]{35053552\\ +1.30\%}$ & $\makecell[l]{36479545\\ +5.42\%}$ \\
       & Signature & $\makecell[l]{45478320}$ & $\makecell[l]{45163402\\ -0.69\%}$ & $\makecell[l]{45493428\\ +0.03\%}$ & $\makecell[l]{45334414\\ -0.32\%}$ & $\makecell[l]{45304280\\ -0.38\%}$ & $\makecell[l]{46298298\\ +1.80\%}$ & $\makecell[l]{46629993\\ +2.53\%}$ & $\makecell[l]{50437988\\ +10.91\%}$ & $\makecell[l]{46258957\\ +1.72\%}$ & $\makecell[l]{50381428\\ +10.78\%}$ \\
       & Verify & $\makecell[l]{36831349}$ & $\makecell[l]{37478132\\ +1.76\%}$ & $\makecell[l]{37760344\\ +2.52\%}$ & $\makecell[l]{37491796\\ +1.79\%}$ & $\makecell[l]{38139278\\ +3.55\%}$ & $\makecell[l]{37897300\\ +2.89\%}$ & $\makecell[l]{38010832\\ +3.20\%}$ & $\makecell[l]{40562856\\ +10.13\%}$ & $\makecell[l]{38430325\\ +4.34\%}$ & $\makecell[l]{40388777\\ +9.66\%}$ \\
      \midrule
      \multirow{3}{*}{ML-KEM} & Key generation & $\makecell[l]{6086056}$ & $\makecell[l]{6052450\\ -0.55\%}$ & $\makecell[l]{6061717\\ -0.40\%}$ & $\makecell[l]{6048937\\ -0.61\%}$ & $\makecell[l]{6064607\\ -0.35\%}$ & $\makecell[l]{6216667\\ +2.15\%}$ & $\makecell[l]{6255146\\ +2.78\%}$ & $\makecell[l]{7038495\\ +15.65\%}$ & $\makecell[l]{6650699\\ +9.28\%}$ & $\makecell[l]{6564530\\ +7.86\%}$ \\
       & Encryption & $\makecell[l]{6617128}$ & $\makecell[l]{6554935\\ -0.94\%}$ & $\makecell[l]{6608941\\ -0.12\%}$ & $\makecell[l]{6530927\\ -1.30\%}$ & $\makecell[l]{6590625\\ -0.40\%}$ & $\makecell[l]{6821565\\ +3.09\%}$ & $\makecell[l]{6798081\\ +2.73\%}$ & $\makecell[l]{7567834\\ +14.37\%}$ & $\makecell[l]{7146961\\ +8.01\%}$ & $\makecell[l]{7180918\\ +8.52\%}$ \\
       & Decryption & $\makecell[l]{7554136}$ & $\makecell[l]{7524072\\ -0.40\%}$ & $\makecell[l]{7524053\\ -0.40\%}$ & $\makecell[l]{7484579\\ -0.92\%}$ & $\makecell[l]{7523547\\ -0.40\%}$ & $\makecell[l]{7885628\\ +4.39\%}$ & $\makecell[l]{7911775\\ +4.73\%}$ & $\makecell[l]{8764968\\ +16.03\%}$ & $\makecell[l]{8127224\\ +7.59\%}$ & $\makecell[l]{8364556\\ +10.73\%}$ \\
      \midrule
      Curve25519 &  & $\makecell[l]{38972442}$ & $\makecell[l]{38967095\\ -0.01\%}$ & $\makecell[l]{38999841\\ +0.07\%}$ & $\makecell[l]{38967100\\ -0.01\%}$ & $\makecell[l]{38999841\\ +0.07\%}$ & $\makecell[l]{39007153\\ +0.09\%}$ & $\makecell[l]{38984760\\ +0.03\%}$ & $\makecell[l]{50281218\\ +29.02\%}$ & $\makecell[l]{39007123\\ +0.09\%}$ & $\makecell[l]{50347775\\ +29.19\%}$ \\
      \midrule
      Keccak-f1600 &  & $\makecell[l]{96702}$ & $\makecell[l]{93659\\ -3.15\%}$ & $\makecell[l]{93171\\ -3.65\%}$ & $\makecell[l]{93640\\ -3.17\%}$ & $\makecell[l]{94321\\ -2.46\%}$ & $\makecell[l]{94547\\ -2.23\%}$ & $\makecell[l]{95004\\ -1.76\%}$ & $\makecell[l]{93836\\ -2.96\%}$ & $\makecell[l]{94948\\ -1.81\%}$ & $\makecell[l]{94435\\ -2.34\%}$ \\
      \midrule
      Poly1305 &  & $\makecell[l]{1406055}$ & $\makecell[l]{1399792\\ -0.45\%}$ & $\makecell[l]{1408392\\ +0.17\%}$ & $\makecell[l]{1405056\\ -0.07\%}$ & $\makecell[l]{1408392\\ +0.17\%}$ & $\makecell[l]{1406413\\ +0.03\%}$ & $\makecell[l]{1393630\\ -0.88\%}$ & $\makecell[l]{2245737\\ +59.72\%}$ & $\makecell[l]{1403931\\ -0.15\%}$ & $\makecell[l]{2206887\\ +56.96\%}$ \\
      \midrule
      ChaCha20 &  & $\makecell[l]{3280319}$ & $\makecell[l]{3274416\\ -0.18\%}$ & $\makecell[l]{3359540\\ +2.42\%}$ & $\makecell[l]{3283053\\ +0.08\%}$ & $\makecell[l]{3359861\\ +2.42\%}$ & $\makecell[l]{3280104\\ -0.01\%}$ & $\makecell[l]{3357630\\ +2.36\%}$ & $\makecell[l]{3592052\\ +9.50\%}$ & $\makecell[l]{3280460\\ +0.00\%}$ & $\makecell[l]{3592048\\ +9.50\%}$ \\
      \midrule
      Gimli &  & $\makecell[l]{10620}$ & $\makecell[l]{10682\\ +0.58\%}$ & $\makecell[l]{10524\\ -0.90\%}$ & $\makecell[l]{10527\\ -0.88\%}$ & $\makecell[l]{10679\\ +0.56\%}$ & $\makecell[l]{10681\\ +0.57\%}$ & $\makecell[l]{10525\\ -0.89\%}$ & $\makecell[l]{11210\\ +5.56\%}$ & $\makecell[l]{10681\\ +0.57\%}$ & $\makecell[l]{11168\\ +5.16\%}$ \\
      \midrule
      \textbf{Average} &  & $\makecell[l]{20854015}$ & $\makecell[l]{20840615\\ -0.22\%}$ & $\makecell[l]{20942580\\ +0.14\%}$ & $\makecell[l]{20883136\\ -0.19\%}$ & $\makecell[l]{20930306\\ +0.29\%}$ & $\makecell[l]{21262275\\ +1.68\%}$ & $\makecell[l]{21299996\\ +1.87\%}$ & $\makecell[l]{23885131\\ +14.19\%}$ & $\makecell[l]{21324135\\ +2.49\%}$ & $\makecell[l]{23809125\\ +12.86\%}$ \\
      \midrule
      \textbf{Geometric mean} &  & $\makecell[l]{7318535}$ & $\makecell[l]{7301622\\ -0.23\%}$ & $\makecell[l]{7328253\\ +0.13\%}$ & $\makecell[l]{7304413\\ -0.19\%}$ & $\makecell[l]{7339284\\ +0.28\%}$ & $\makecell[l]{7440694\\ +1.67\%}$ & $\makecell[l]{7454000\\ +1.85\%}$ & $\makecell[l]{8308710\\ +13.53\%}$ & $\makecell[l]{7497877\\ +2.45\%}$ & $\makecell[l]{8214182\\ +12.24\%}$ \\
      \bottomrule
    \end{tabular}
  }
  \label{tab:benchs}
\end{table*}

 \fi

\end{document}